\documentclass[draft]{econsocart}
\RequirePackage{silence}
\RequirePackage[colorlinks,citecolor=blue,linkcolor=blue,urlcolor=blue,hypertexnames=false]{hyperref}

\startlocaldefs

\makeatletter
\def\ps@workingpaper{%
  \let\@mkboth\@gobbletwo
  \def\@oddhead{}%
  \def\@evenhead{}%
  \def\@oddfoot{\hfil\thepage\hfil}%
  \def\@evenfoot{\hfil\thepage\hfil}%
}
\let\ps@plain\ps@workingpaper
\makeatother

\makeatletter
\let\@font@warning\@gobble
\makeatother
\AtBeginDocument{%
  \hbadness=10000
  \vbadness=10000
  \hfuzz=16000pt
  \vfuzz=16000pt
  \raggedbottom
}

\usepackage{booktabs}
\usepackage{longtable}
\usepackage{tabularx}
\usepackage{array}
\usepackage{pdflscape}
\usepackage{threeparttable}
\usepackage{threeparttablex}
\usepackage{enumitem}

\usepackage{algorithm}
\usepackage{algpseudocode}

\makeatletter
\@ifundefined{abovecaptionskip}{\newskip\abovecaptionskip}{ }
\@ifundefined{belowcaptionskip}{\newskip\belowcaptionskip}{ }
\makeatother

\theoremstyle{plain}

\newtheorem{theorem}{Theorem}

\newtheorem{lemma}[theorem]{Lemma}

\theoremstyle{definition}
\newtheorem{definition}{Definition}
\newtheorem*{example}{Example}
\newtheorem{remark}{Remark}
\newtheorem{assumption}{Assumption}
\newtheorem{proposition}{Proposition}
\newtheorem{corollary}{Corollary}

\endlocaldefs

\begin{document}
\hbadness=10000
\vbadness=10000
\hfuzz=16000pt
\vfuzz=16000pt
\raggedbottom

\begin{frontmatter}

\title{Choosing What to Calibrate and What to Estimate in
Structural Models}
\runtitle{Choosing What to Calibrate}

\begin{aug}
%
%
%
\author[add1]{\fnms{Joan}~\snm{Alegre Canton}\ead[label=e1]{joan.alegre.canton@gmail.com}}
\address[add1]{%
\orgdiv{Department of Economics},
\orgname{Universidad Carlos III de Madrid}}

\end{aug}

\begin{funding}
I am grateful to Juan Carlos Escanciano and Andrey Ramos for their advice and helpful comments. 
Research funded by Ministerio de Ciencia e Innovación, grants PGC2018-096732-B-I00 and PID2021-127794NB-I00, and Comunidad de Madrid, grants EPUC3M11 (VPRICIT) and H2019/HUM-589.
\end{funding}
\begin{abstract}
Structural models often fix (calibrate) some parameters and estimate the rest, but this calibration--estimation partition is usually chosen by convention. This paper treats that choice as an econometric partition-selection problem. For each admissible partition, we construct a scalar sensitivity statistic measuring the local response of a target object---such as a policy effect, welfare measure, impulse response, or treatment effect---to perturbations of the calibrated parameters. The selected partition minimizes this statistic and therefore minimizes worst-case local bias from calibration errors. We first illustrate the decision problem in two canonical examples. We then apply it to the New Keynesian model of Nakamura and Steinsson (2018), where the partition choice has large implications for credibility: some partitions remain reliable under sizeable miscalibrations, whereas others generate large bias from small calibration errors. The procedure requires only local derivatives, avoids repeated re-estimation, and applies to a broad class of structural models.
\end{abstract}



\end{frontmatter}

\makeatletter
\@ifundefined{nolinenumbers}{}{\nolinenumbers}
\makeatother

\pagestyle{workingpaper}
\thispagestyle{workingpaper}

\section{Introduction}

A common practice in structural model estimation is to calibrate some structural parameters---that is, to hold them fixed at prespecified values---and estimate the rest.\footnote{Throughout the paper, a parameter is \emph{calibrated} if it is held fixed during the estimation problem. This usage is narrower than in some applied and econometric work, where \emph{calibration} refers to choosing parameters to match empirical moments, i.e. to minimum-distance estimation; see, for example, \citet{cocci2021standard}. In our terminology, the latter is estimation.} This approach is often justified when full estimation is computationally infeasible \citep{jorgensen2023sensitivity}, when the model is not fully identified \citep{iskrev2019expect}, or when population values are known. In practice, however, the calibration--estimation partition is rarely chosen by a formal econometric criterion. It is often inherited from previous studies, justified informally, or dictated by computational convenience. 

This paper provides an econometric criterion for choosing what to calibrate and what to estimate in a general class of structural models. We treat the calibration--estimation split as a partition-selection problem. For each admissible partition, we construct a scalar sensitivity statistic that measures the local response of the object of interest to perturbations of the calibrated parameters. We then select the admissible partition that minimizes this sensitivity. The resulting rule chooses the partition that delivers the smallest worst-case local bias induced by calibration error. This perspective connects the paper to robust estimation under local misspecification, especially \cite{bonhomme2022minimizing}.

A natural alternative would be to select the partition by minimizing an MSE criterion that trades off calibration bias and sampling variance. We focus instead on bias, through the sensitivity statistic, because implementing an MSE criterion requires the full joint variance--covariance matrix of the stacked moments or reduced-form estimates. This requirement is often restrictive when different moments are estimated from different datasets and not all cross-covariances are identified.

A key feature of the implementation is that the least-sensitive partition is computed relative to two inputs. The first is a reference point for the structural parameters, typically obtained from a first-stage estimation. This reference point is where local sensitivities are evaluated. It also allows the researcher to compare candidate partitions without re-estimating the model separately for each one. The second input is a normalization of the calibrated parameters. This normalization specifies the scale on which calibration errors are measured and makes perturbations comparable across parameters expressed in different units. Hence, the least-sensitive partition should be interpreted as conditional on both the reference point and the chosen normalization of calibration errors. 

The method is developed in a minimum-distance framework \citep{newey1994large}, broad enough to cover several estimation approaches commonly used in structural work, including the Generalized Method of Moments (GMM), Classical Minimum Distance (CMD), Simulated Method of Moments (SMM), or Indirect Inference (II). 

We illustrate the importance of partition choice in three settings. The first two are New Keynesian models, based on \citet{an2007bayesian} and \citet{nakamura2018high}. The third is the dynamic discrete-choice model in \citet{kalouptsidi2021identification}. Across these examples, we show that the calibration--estimation partition has large implications for robustness. In the Nakamura--Steinsson application (Section~\ref{sec: empirical application}), where the object of interest is the monetary policy's information-effect parameter, the sensitivity statistic varies substantially across admissible partitions. For some partitions, miscalibrations of only a few percent are sufficient to overturn the qualitative conclusion about the object of interest. By contrast, miscalibrations of comparable magnitude generate only a small worst-case bias under the least-sensitive partition. These results show that fragility can be avoided purely by choosing more carefully which parameters to fix and which to estimate.

The paper is related to several strands of literature. It builds most directly on work on sensitivity analysis for structural models. \citet{iskrev2019expect} and \citet{jorgensen2023sensitivity} study how estimates and model outputs respond to perturbations in calibrated parameters, while \citet{lau2024sensitivity} develops local and global sensitivity measures for dynamic discrete choice models with fixed parameters. Relative to \cite{jorgensen2023sensitivity}, our proposed statistic is scalar and independent of calibrated parameter units. Related work studies sensitivity to moments or identifying assumptions, including \citet{andrews2017measuring}, \citet{honore2020informativeness}, and \citet{escanciano2013set}. Our contribution differs in that sensitivity is not the final object of interest. Instead, it is the input into a decision rule that selects which parameters should be fixed and which should be estimated. 

The paper is also related to the literature on identification and partial identification of counterfactuals in dynamic discrete choice models. Building on \cite{aguirregabiria2010another}, \cite{aguirregabiria2014identification}, \cite{norets2014semiparametric}, and \cite{arcidiacono2020identifying}, \cite{kalouptsidi2021identification} characterize which counterfactuals are point identified despite payoff nonidentification, while \cite{kalouptsidi2026counterfactual} characterize the corresponding identified set for the counterfactual. We reverse the direction of this exercise. In that literature, the calibration or identifying strategy is taken as given, and the question is which counterfactuals are point or set identified. Here, the counterfactual exercise is taken as given, and the question is which normalization—or calibration–estimation partition—makes that same object of interest least-sensitive to local calibration error. 

The minimax interpretation connects the paper to the literature on estimation and inference under local misspecification (e.g., \citealp{newey1985generalized,conley2012plausibly,kitamura2013robustness}). The closest paper is \citet{bonhomme2022minimizing}, who choose an estimator, through its influence function, to minimize worst-case mean squared error under local misspecification. We differ in two main respects: our criterion targets worst-case bias rather than worst-case MSE, and our decision variable is the calibration--estimation partition rather than the influence function. Also closely related is \citet{armstrong2021sensitivity}, who show that GMM weights should account for both sampling variance and misspecification bias when moments may be only approximately valid. We apply a similar local-robustness logic to a different margin: the choice of which parameters to calibrate and which to estimate.

The remainder of the paper is organized as follows. Section~\ref{sec: motivational example section} presents two motivational examples. Section~\ref{sec: Framework} defines the framework and discusses the construction of the partition-selection methodology in detail. Section~\ref{sec: Justification} offers a risk-minimization interpretation and relates the methodology to \cite{bonhomme2022minimizing}. Section~\ref{sec: Implementation} provides a step-by-step implementation guide. Section~\ref{sec: empirical application} applies the method to \cite{nakamura2018high}, and Section~\ref{sec: Monte carlo Simulation} presents Monte Carlo evidence.
\section{Two Motivational Examples}\label{sec: motivational example section}
This section illustrates the importance of choosing the calibration–estimation partition. In
both examples, the model requires fixing at least one structural parameter before the remaining
parameters can be estimated. The central point is that different admissible choices of the fixed
parameter can generate very different sensitivity of the object of interest to calibration error.

\subsection{A dynamic discrete choice model}
\label{subsec:ddc_motivation}

Consider the entry--exit model in \citet{kalouptsidi2021identification}. In each period, a firm chooses whether to be inactive, \(a=0\), or active, \(a=1\). The state is \(x=(k,w)\), where \(k\in\{0,1\}\) is the firm's lagged activity status and \(w\in\{w_L,w_M,w_H\}\) is an exogenous demand state. Up to the idiosyncratic choice shock, flow payoffs are
\[
\pi_a(k,w;\eta)=
\begin{cases}
k\,sv, & a=0,\\[2mm]
k\{\bar\pi(w)-fc\}-(1-k)ec, & a=1,
\end{cases}
\qquad
\eta=(ec,fc,sv).
\]
Here \(ec\) is the entry cost, \(fc\) is the fixed cost, and \(sv\) is the scrap value. Given \(\eta\), the choice-specific value and CCP are
\[
\begin{aligned}
v_a(x;\eta)
&=
\pi_a(x;\eta)+\rho\,\mathbb E[V(x';\eta)\mid a,x],\\
p_a(x;\eta)
&=
\frac{\exp\{v_a(x;\eta)\}}
{\exp\{v_0(x;\eta)\}+\exp\{v_1(x;\eta)\}},
\end{aligned}
\]
where the second expression uses the type-I extreme-value normalization.\footnote{The numerical primitives are those in \citet{kalouptsidi2021identification}. Writing \(\xi\) for the inverse-demand slope, residual demand is \(P_t=w_t-\xi Q_t\), marginal cost is \(c=11\), \(\xi=1.5\), \(w=(20,17,12)\), and \(\bar\pi(w)=(w-c)^2/(4\xi)\). The remaining primitives are \(\rho=0.95\), \(ec=9\), \(fc=5.5\), \(sv=10\), and
\[
F_w=
\begin{pmatrix}
0.40 & 0.35 & 0.25\\
0.30 & 0.40 & 0.30\\
0.20 & 0.20 & 0.60
\end{pmatrix}.
\]
The idiosyncratic shocks are type-I extreme value.}

The counterfactual policy is a \(20\%\) proportional tax on the active operating payoff. Let \(\tau=0.20\). The counterfactual flow payoff is
\[
\pi^\tau_0(k,w;\eta)=\pi_0(k,w;\eta),
\qquad
\pi^\tau_1(k,w;\eta)
=
(1-\tau)\,k\{\bar\pi(w)-fc\}-(1-k)ec .
\]
The object of interest is the counterfactual CCP vector,
\[
p^\tau_1(x;\eta).
\]
The model is not point identified without one additional normalization: one of \(ec\), \(fc\), or \(sv\) must be fixed \citep{aguirregabiria2014identification}. This is not merely a computational detail. Observationally equivalent parameterizations of the model can imply different values of \(p^\tau_1(x;\eta)\). In applications, researchers most often fix either the scrap value or the fixed cost. Table~\ref{tab:DDC-calibration} summarizes this choice in related studies.

\begin{table}[H]
    \centering
    \begin{threeparttable}
        \caption{Calibration Strategy in Recent Studies}
        \label{tab:DDC-calibration}
        \begin{tabular}{l|c}
            \toprule
            \textit{Calibrate the scrap value} & \textit{Calibrate the fixed cost} \\ 
            \midrule
            \cite{das2007market}              & \cite{ryan2012costs} \\
            \cite{kryukov2010dynamic}         & \cite{sweeting2013dynamic} \\
            \cite{collard2011productivity}    & \cite{santos2017sunk} \\
            \cite{collard2013demand}          & \cite{beresteanu2024dynamics} \\
            \cite{aguirregabiria2012dynamic}  & \\
            \cite{suzuki2013land}             & \\
            \cite{barwick2015costs}           & \\
            \cite{lin2015quality}             & \\
            \cite{igami2017estimating}        & \\
            \cite{varela2018costs}            & \\
            \bottomrule
        \end{tabular}
        \begin{tablenotes}
            \small
            \item Note: No study in this list calibrated the entry cost.
        \end{tablenotes}
    \end{threeparttable}
\end{table}

Figure~\ref{fig:ddc_calibration_sensitivity} plots the object of interest as a function of the fixed parameter. For each curve, either \(sv\) or \(fc\) is fixed at a value around its population value, the remaining cost parameters are then estimated, and the counterfactual CCP is recomputed. The horizontal axis is centered at the true value of the fixed parameter, so zero corresponds to correct calibration. The horizontal line is the population counterfactual CCP; the vertical distance from this line is the bias induced by miscalibration.

\begin{figure}[htbp]
    \centering
    \caption{Counterfactual CCP as a function of the fixed parameter}
    \includegraphics[width=0.72\textwidth,keepaspectratio]{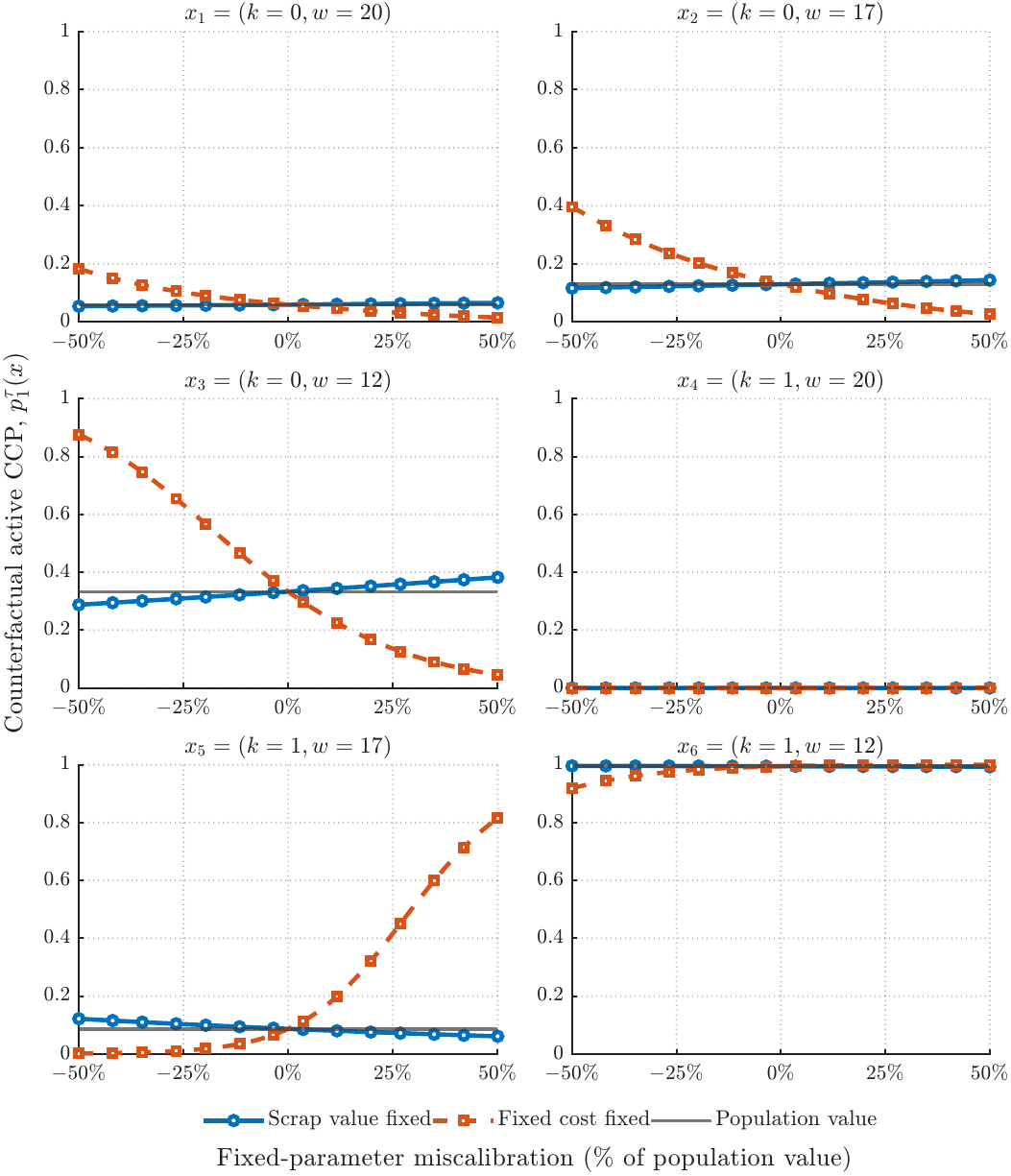}
    \label{fig:ddc_calibration_sensitivity}
\end{figure}

The figure shows that fixing the scrap value produces a relatively flat curve, while fixing the fixed cost produces a much steeper one. Thus, two common ways of identifying the same entry--exit model imply very different robustness properties for the same counterfactual object.

This motivational example is related to \cite{kalouptsidi2021identification} and \cite{kalouptsidi2026counterfactual}. Those papers take the calibration strategy as given and ask which counterfactuals are point identified, or what set of counterfactual values is implied when point identification fails. We ask the dual question. We take the counterfactual object, $p^\tau_1(x;\eta)$, as given and compare the admissible calibration strategies themselves. In this example, fixing sv, fc, or ec generates a different local range for the same counterfactual CCP under calibration errors; our method selects the strategy that makes this local range smallest.

\subsection{A canonical New Keynesian model}
\label{subsec:nk_motivation}

The same issue arises in a standard macroeconomic environment. Consider the three-equation forward-looking New Keynesian model used in \citet{an2007bayesian}:

\begin{align*}
	\textbf{Euler Equation:} \quad & 
	y_t = \mathbb{E}_t y_{t+1} + g_t - \mathbb{E}_t g_{t+1} - \tfrac{1}{\tau}(r_t - \mathbb{E}_t \pi_{t+1} - \mathbb{E}_t z_{t+1}). \\
	\textbf{Phillips curve:} \quad & \pi_t = \rho \mathbb{E}_t \pi_{t+1} + \kappa(y_t - g_t), \\
	& c_t = y_t - g_t. \\
	\textbf{Monetary Policy:} \quad & r_t = \rho_r r_{t-1} + (1-\rho_r)\psi_1 \pi_t + (1-\rho_r)\psi_2 (y_t - g_t) + \epsilon_{rt}, \\
	& g_t = \rho_g g_{t-1} + \epsilon_{gt}, \\
	& z_t = \rho_z z_{t-1} + \epsilon_{zt}.
\end{align*}  

The shocks \(\varepsilon_{r,t}\), \(\varepsilon_{g,t}\), and \(\varepsilon_{z,t}\) are mutually uncorrelated, with variances \(\sigma_r^2\), \(\sigma_g^2\), and \(\sigma_z^2\)\footnote{We use the calibration in \citet{an2007bayesian} as the data-generating process:
\[
(\rho,\tau,\kappa,\psi_1,\psi_2,\rho_r,\rho_g,\rho_z,\sigma_r^2,\sigma_g^2,\sigma_z^2)
=
(0.9975,2,0.33,1.5,0.125,0.75,0.95,0.9,0.04,0.36,0.09).
\]}.

Let the object of interest be \(\psi_1\), the coefficient on inflation in the Taylor rule. In this model, \(\psi_1\) cannot be jointly identified with all elements of \(\{\psi_2,\rho_r,\sigma_r^2\}\); one of these three parameters must be fixed in order to estimate the others \citep[e.g.][]{iskrev2010local,komunjer2011dynamic,qu2012identification,qu2017global,kocikecki2023solution}. The researcher must therefore decide whether to fix \(\psi_2\), \(\rho_r\), or \(\sigma_r^2\).

Because these parameters have different units, Figure \ref{fig:nk_calibration_sensitivity} reports miscalibrations in normalized percent-of-range units. Thus, \(x=0\) means correct calibration, and \(x=0.1\) means that the fixed parameter is set \(10\%\) of its plausible range above its baseline value\footnote{To compute the normalization ranges we use the priors for \(\psi_2\), \(\rho_r\), and \(\sigma_r\) in \cite{an2007bayesian}. We compute its respective $0.5\%$ and $99.5\%$ quantiles. Obtaining:
\[
\psi_2 \in [0,0.98], \qquad
\rho_r \in [0,1], \qquad
\sigma_r^2 \in [0,0.35].
\]}. For each curve, we fix one parameter at nearby values, re-estimate the remaining parameters by MLE, and record the implied value of \(\psi_1\).

\begin{figure}[htbp]
	\centering
	\caption{Value of \(\psi_1\) as a function of miscalibration in the fixed parameter}
	\includegraphics[width=0.78\textwidth,keepaspectratio]{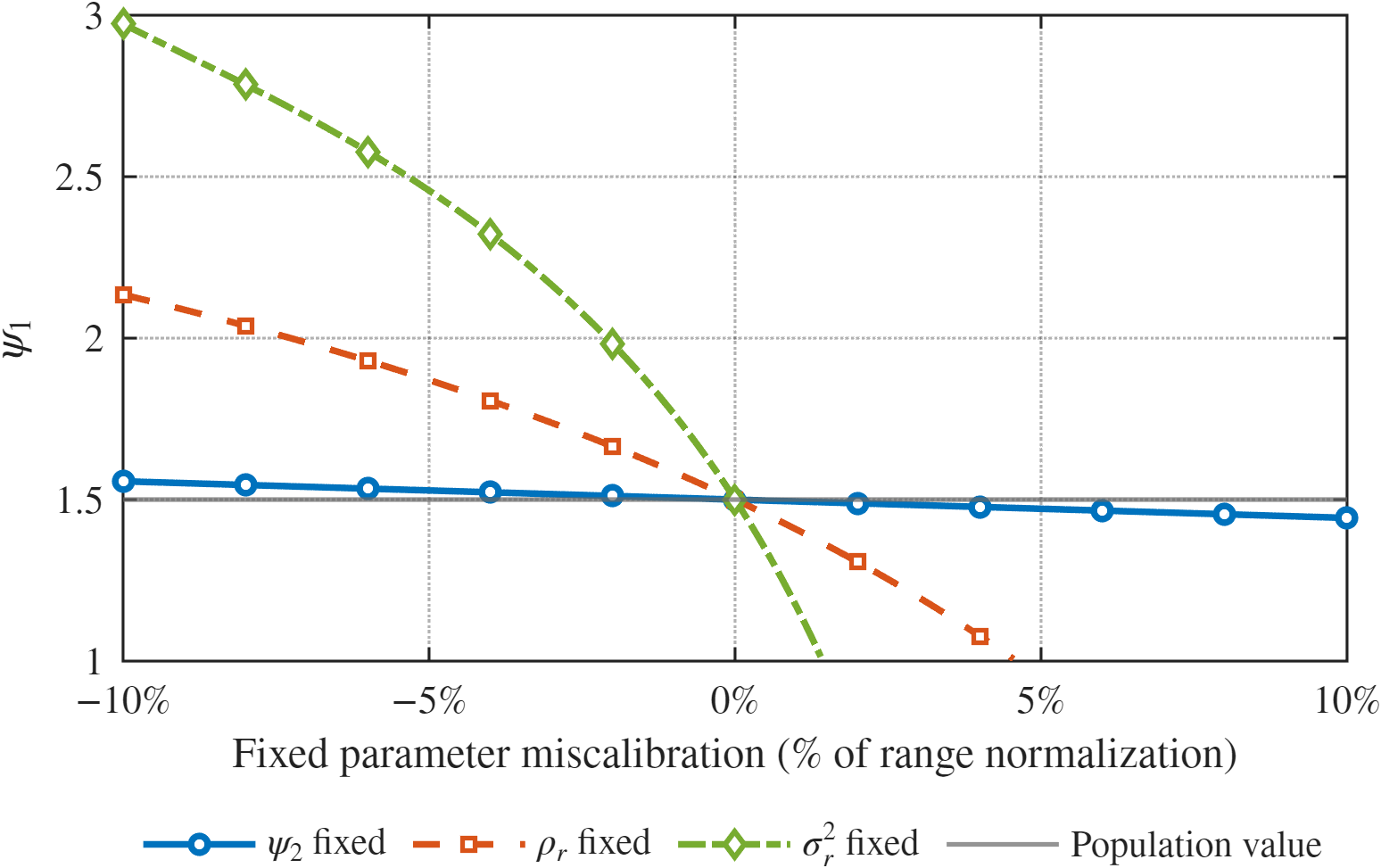}
	\label{fig:nk_calibration_sensitivity}
\end{figure}

The true value is \(\psi_1=1.5\). Hence, the vertical distance from \(1.5\) is the bias induced by miscalibration. Fixing \(\psi_2\) is relatively robust: the curve is nearly flat around zero. Fixing \(\rho_r\) or \(\sigma_r^2\), in contrast, can generate large changes in \(\psi_1\) even for small normalized perturbations.

These two examples motivate the criterion developed below. Full re-estimation curves are informative, but they are costly to compute, difficult to compare when the target or the calibrated block is multidimensional, and sensitive to the units of the fixed parameters. The statistic \(K_S\) replaces these curves with a local, scalar measure of worst-case sensitivity of the object of interest to calibration error.

\section{Framework and Partition Selection Methodology}\label{sec: Framework}

Let \(\eta \in \mathcal N \subset
\mathbb R^{d_\eta}\) denote the structural parameter vector  and let
$
\gamma=\Gamma(\eta)\in\mathbb R^{d_\gamma}
$
denote the object of interest. A partition is an index set $S\subseteq\{1,\dots,d_\eta\}$ selecting which components of $\eta$
are estimated. Define
\[
\alpha := (\eta_j)_{j\in S},
\qquad
\beta := (\eta_j)_{j\in S^c},
\]
so that $\eta=(\alpha,\beta)$, where $\alpha$ denotes the estimated parameters and $\beta$ the
calibrated (fixed) parameters\footnote{Equivalently, $S$ induces $\mathcal N=\mathcal A\times\mathcal B$ with
$\mathcal A:=\prod_{j\in S}\mathcal N_j$ and $\mathcal B:=\prod_{j\in S^c}\mathcal N_j$, where $\mathcal N_j$
is the parameter space associated with the $j$th element of $\eta$.}. For a given partition, fixing
\(\beta\) and solving the minimum-distance problem yields
\(\alpha(\beta;S)\), which in turn induces the partition-specific map
\[
\gamma(\beta;S)=\Gamma(\alpha(\beta;S),\beta).
\]
This induced map is the central object for partition selection: it describes
how the object of interest changes when the fixed parameters are perturbed and
the remaining parameters are re-estimated.

Our comparison is local. Partitions are evaluated around a reference point
\(\eta^*\), typically obtained from a first-stage estimation, which provides
the common baseline at which sensitivities and admissibility are assessed.

The goal of this section is to compare these induced maps across partitions and
select the one that is least-sensitive to local miscalibration of the fixed
parameters. Formally, after defining the admissible set of partitions $\mathcal{S}_{\mathrm{AD}}$, and the
sensitivity statistic \(K_S\), we choose
\[
S^* \in \arg\min_{S\in \mathcal{S}_{\mathrm{AD}}} K_S.
\]
The rest of the section builds the ingredients of this criterion step by step.
Subsection~\ref{subsec:framework} introduces the conditional
minimum-distance problem, the induced map \(\gamma(\beta;S)\), and the reference point $\eta^*$. Subsection~\ref{sec: normalization and sensitivity} derives
the sensitivity statistic \(K_S\). Subsection~\ref{sec: admissible_partitions}
defines the admissible set, \(\mathcal{S}_{\mathrm{AD}}\), and the least-sensitive partition $S^*$.

\subsection{Framework}\label{subsec:framework}

 We study the case in which a structural economic model implies a system of $d_g$ equations
\[
g(\eta): \mathcal{N} \longrightarrow \mathcal{G}\subseteq \mathbb{R}^{d_g}.
\]

The identified set is defined by all structural parameters that solve the model equations:
\begin{equation}
    \mathcal{N}_0 := \left\{\eta \in \mathcal{N} \;\middle|\; g(\eta)=0\right\}.
\end{equation}
Throughout the paper we assume the following conditions about the system of equations and the object of interest are satisfied.

 \begin{assumption} \label{ass: regularity conditions}
    $i)$ There exists a true population parameter $\eta_0\in \mathcal{N}_0$;  $ii)$ $g(\alpha,\beta;S)$ is twice continuously differentiable in $(\alpha,\beta)$; $iii)$ $\Gamma(\alpha,\beta)$ is continuously differentiable.
\end{assumption}

Given $\beta$, the conditional minimum-distance solution for the estimated parameters $\alpha$ is given by
\begin{equation}\label{eq: MD opt problem}
    \alpha(\beta;S)\in
\arg\min_{\alpha\in\mathcal{A}}
g(\alpha,\beta)'Wg(\alpha,\beta),
\end{equation}
where $W$ is a positive-definite weighting matrix.

To compute the sensitivity statistics and the set of admissible partitions, we work around a \emph{reference point} $\eta^*$, which may differ from the true parameter $\eta_0$. Typically, it is obtained from a first-stage estimation. For any partition we write $\eta^*=(\alpha^\ast,\beta^\ast)$ for its corresponding relabeled components. See Section \ref{sec:first-stage-reference-point} for a detailed discussion on the reference point.

\begin{assumption} \label{ass: regularity assumptions reference point}
    $(i)$ The reference point $\eta^*$ is an interior point of $\mathcal{N}$; $(ii)$ the reference point satisfies $\eta^* \in \mathcal{N}_0$.
\end{assumption}

Throughout the main text, we focus on cases in which the system of equations $g(\cdot)$ does
not depend on the partition S, except through the relabeling of coordinates induced by the
split $(\alpha, \beta)$. This condition is satisfied in many applications of GMM, CMD, SMM, and II.
There are, however, important cases in which the system of equations does depend on the
partition. A leading example is MLE: deciding which parameters are estimated and which are
fixed changes the score equations used in the estimation problem. In Appendix \ref{appendix: MLE reference point} we extend
the methodology to the MLE case.

\subsection{Normalization and Sensitivity Statistic} \label{sec: normalization and sensitivity}
We now formalize the local sensitivity measure used to compare candidate partitions. We proceed in three steps: derive how the MD solution responds to perturbations in the fixed parameters, translate this into sensitivity of the object of interest, and normalize the result so it can be compared across partitions.

The Jacobians of the system of equations, evaluated at the reference point, are denoted by:
\[
G_\alpha:=\frac{\partial g(\alpha^\ast,\beta^\ast)}{\partial \alpha},
\qquad
G_\beta:=\frac{\partial g(\alpha^\ast,\beta^\ast)}{\partial \beta}.
\]
Throughout the paper, we suppress the dependence of these Jacobians on the reference point and on the partition $S$.

\begin{lemma}
    Suppose that Assumptions \ref{ass: regularity conditions} and \ref{ass: regularity assumptions reference point} hold, and the Jacobian for the estimated parameters satisfies $Rank(G_{\alpha}) = |S|$. Then,
\begin{equation} \label{eq: der_est_fixed}
    \frac{\partial \alpha(\beta^\ast;S)}{\partial \beta}
=
-
\left(G_\alpha'WG_\alpha\right)^{-1}
G_\alpha'WG_\beta.
\end{equation}
Accordingly, by the chain rule,
\begin{equation} \label{eq: der_gamma_fixed}
    \frac{\partial \gamma(\beta^\ast;S)}{\partial \beta}
=
\frac{\partial \Gamma}{\partial \alpha}
\frac{\partial \alpha(\beta^\ast;S)}{\partial \beta}
+
\frac{\partial \Gamma}{\partial \beta},
\end{equation}
where $\partial \Gamma/\partial \alpha$ and $\partial \Gamma/\partial \beta$ are evaluated at $(\alpha^\ast,\beta^\ast)$. 
\end{lemma}
\begin{proof}
    This follows directly from applying the Implicit Function Theorem (IFT) to the first-order conditions of the MD problem in (\ref{eq: MD opt problem}).
\end{proof}
\cite{jorgensen2023sensitivity} derives equivalent expressions in the GMM setting. Since the derivative of the object of interest depends on the units of the fixed parameters, we normalize them using the diagonal matrix
\[
\Sigma_S:=\operatorname{diag}\bigl((\Delta_j)_{j\in S^c}\bigr),
\]
where \(\Delta_j\) determines the scale used for parameter \(j\).  Section~\ref{sec: practical guidance Sigma_S} provides further discussion on how to construct it in practice.
\begin{definition}
We define the sensitivity statistic as
\begin{equation}\label{eq: sensitivity-statistic}
K_S (\eta^*, \Sigma_S)
:= \sqrt{|S^c|}\cdot\Bigl\lVert \frac{\partial \gamma(\beta^\ast;S)}{\partial \beta}\,\Sigma_{S} \Bigr\rVert_{2}.
\end{equation}
Here, $\lVert A \rVert_{2} = \sigma_{\max}(A)$ denotes the spectral norm, where $\sigma_{\max}(A)$ is the largest singular value of $A$. In Section~\ref{sec: Justification}, we provide the formal justification for using $K_S$: we show that $K_S$ is proportional to the worst-case bias induced by a local miscalibration. Equivalently, $K_S$ can be read as the largest first-order change in the object of interest generated by a normalized local perturbation of the fixed parameters, so it summarizes worst-case normalized bias in a single scalar.
\end{definition}

Note that $K_S$ depends on the reference point
$\eta^*$, because the derivative
$\partial\alpha(\beta^\ast;S)/\partial\beta$ is obtained from the IFT using Jacobians
evaluated at $\eta^*$. It also depends on the normalization matrix.

Under range normalization, $\Sigma_S$ plays two roles: 
$(i)$ it removes dependence on arbitrary units of measurement, and 
$(ii)$ it incorporates external information about plausible parameter values 
through bounds, priors, or confidence intervals. As a result, $K_S$ captures 
both the model-implied sensitivity of $\gamma$ to $\beta$ and the uncertainty 
surrounding $\beta^\ast$.

\begin{example}
Suppose both the object of interest and the fixed parameter are scalars, and let
$S^c=\{j\}$. Under range normalization, the sensitivity statistic reduces to
\begin{equation}
\begin{aligned}
    K_S
    &=
    \underbrace{
    \left|
    \frac{\partial \gamma(\beta^\ast;S)}{\partial \beta}
    \right|
    }_{\scriptstyle \text{Model-implied sensitivity}}
    \cdot
    \overbrace{
    \left(\eta_{j,\max}-\eta_{j,\min}\right)
    }^{\scriptstyle \text{Uncertainty about fixed parameter}} .
\end{aligned}
\end{equation}
Thus, the statistic is simply the product of two components: 
$(i)$ the sensitivity of the object of interest to the parameter, and 
$(ii)$ the uncertainty surrounding that parameter's value.
\end{example}

\subsection{Admissible Partitions and Least-sensitive Partition}
\label{sec: admissible_partitions}
This subsection completes the partition-selection rule. We first define the admissible class of partitions and then select, within that class, the partition that minimizes the sensitivity statistic $K_S$.

Admissibility imposes three requirements: the object of interest must depend nontrivially on the estimated parameters, the estimated parameters must be conditionally locally identified, and the estimation error in the object of interest must vanish in mean square as $n$ grows.

\begin{assumption}[Admissibility] \label{ass: admissibility} The partition, $S$, and the reference point, $\eta^*$, satisfy:

\begin{enumerate}[label=\roman*)]
\item Non-trivial partition: $S\in \mathcal S$ where $\mathcal S$ is 
\[
\mathcal S
:=
\Big\{
S\subseteq\{1,\ldots,d_\eta\}:
0<|S|,\ 
\frac{\partial \Gamma(\alpha^*,\beta^*)}{\partial \alpha}\neq 0
\Big\}.
\]
\item Conditional local identification: $\mathrm{rank}\big(G_\alpha\big)=|S|$.
\item Declining second moments of the estimation error: 
\begin{equation}
\label{eq:target_variance_decay}
\mathbb{E}\!\left[
\left\lVert
\widehat{\gamma}(\beta^*;S)-\gamma(\beta^*;S)
\right\rVert_2^2
\right]=o(1).
\end{equation}
\end{enumerate}
\end{assumption}

Assumption \ref{ass: admissibility}$.i)$ ensures the object of interest is not only (locally) explained by the fixed parameters. In the case where a particular structural parameter is the object of interest it would rule out partitions where that parameter is fixed. 

 Assumption \ref{ass: admissibility}$.ii)$ would be the typical regular local identification rank condition for the estimated parameters. 

Third, Section~\ref{sec: Justification} shows that the least-sensitive partition minimizes
worst-case local bias and, under an additional higher-order condition, also minimizes
worst-case mean squared error asymptotically. For this MSE interpretation, the sampling
variation in the estimated object of interest must be asymptotically negligible, so that the
criterion is driven by calibration bias. We therefore restrict attention to partitions satisfying
Assumption~\ref{ass: admissibility}$.iii)$.

Finally, let $\mathcal R$ encode any hard researcher restrictions (for example, parameters
that must always be estimated or always fixed, or upper and lower bounds on the number of
estimated parameters). We define the admissible set as
\[
\mathcal S_{\mathrm{AD}}(\eta^*)\; := \;\{S\in \mathcal{R} \; \textit{ s.t. } \; \textit{Assumption }\ref{ass: admissibility} \textit{ holds}\}.
\]

Given a reference point
$\eta^*$, and a normalization matrix $\Sigma_S$ we evaluate the sensitivity statistic $K_S(\eta^*,\Sigma_S)$ for each partition
$S\in\mathcal S_{\mathrm{AD}}(\eta^*)$ and choose a partition with the smallest value.

\begin{definition}[Least-sensitive partition]
\label{def:least_sensitive_partition}
Conditional on the reference point $\eta^*$, a least-sensitive partition is any solution to
\[
S^*(\eta^*, \{\Sigma_S\})
\in
\arg\min_{S\in\mathcal S_{\mathrm{AD}}(\eta^*)} K_S(\eta^*, \Sigma_S).
\]
\end{definition}

\section{Justification of the Partition Selection Method} \label{sec: Justification}

This section first gives a local minimax-bias justification and then derives an asymptotic worst-case-MSE interpretation. For a given admissible partition $S$, we interpret the choice of partition through a minimax criterion over a local neighborhood of miscalibrations, and relate this perspective to \citet{bonhomme2022minimizing}'s approach to local robust estimation under misspecification.

For a fixed partition $S$, define the local neighborhood around $\beta^*$ by
\begin{equation}
    B_{\epsilon}(\beta^*;S)
    \;:=\;
    \left\{
        \beta_0 \in \mathcal{B}
        \; \textit{s.t.} \;
        \left\lVert \Sigma_S^{-1}(\beta^*-\beta_0) \right\rVert_{2}
        \leq \epsilon \sqrt{|S^c|}
    \right\}.
\end{equation}
The factor $\sqrt{|S^c|}$ scales the radius by the number of fixed parameters, so that local perturbations are comparable across partitions with different cardinalities. Thus, $\epsilon$ denotes the magnitude of the normalized miscalibration of a single fixed parameter. For a given $\epsilon$-neighborhood, define the worst-case bias by
\begin{equation} \label{eq:WorstBias}
    \operatorname{WorstBias}(\beta^*,\epsilon;S)
    \;:=\;
    \max_{\beta_0 \in B_{\epsilon}(\beta^*;S)}
    \left\lVert \gamma(\beta^*; \; S) - \gamma(\beta_0;\; S) \right\rVert_{2}.
\end{equation}

The following lemma connects worst-case bias to the sensitivity statistic $K_S$.

\begin{lemma} \label{lemma: WorstBias}
     Suppose Assumptions \ref{ass: regularity conditions} and \ref{ass: regularity assumptions reference point} hold, and $S\in \mathcal{S}_{\mathrm{AD}}(\eta^*)$. The worst-case bias and the sensitivity statistic are locally proportional:
    \begin{equation}
        \operatorname{WorstBias}(\beta^*,\epsilon;S) = \epsilon K_S + o(\epsilon).
    \end{equation}
\end{lemma}

\begin{proof}
    See Appendix \ref{appendix:main-proofs}.
\end{proof}

\begin{assumption} \label{ass: uniq}
    The optimization problem $\min_{S \in \mathcal{S}_{AD}} K_S$ has a unique solution $S^*$.
\end{assumption}

The next theorem shows that minimizing $K_S$ is equivalent to minimizing worst-case local bias.

\begin{theorem} \label{thm: sufficiency}
    Suppose Assumptions \ref{ass: regularity conditions}, \ref{ass: regularity assumptions reference point}, and \ref{ass: uniq} hold. Then there exists a scalar $\mathcal{E} > 0$ such that for all $0 < \epsilon \leq \mathcal{E}$,
    \begin{equation}
        S^*
        \;:=\;
        \arg\min_{S \in \mathcal{S}_{AD}} K_S
        \;=\;
        \arg\min_{S \in \mathcal{S}_{AD}} \operatorname{WorstBias}(\beta^*,\epsilon;S).
    \end{equation}
\end{theorem}

\begin{proof}
    See Appendix \ref{appendix:main-proofs}.
\end{proof}

In the boundary case where all parameters are estimated ($|S^c| = 0$), no parameters are subject to calibration error. The worst-case bias is exactly zero, making the full-estimation partition strictly preferred under our criterion (for full estimation $K_S=0$) whenever it is empirically admissible.

Theorem~\ref{thm: sufficiency} also yields a worst-case mean squared error interpretation. Let
\begin{equation}
\begin{aligned}
MSE(\beta^*,\beta_0,S)
&:= \mathbb{E}\left[
\left(\widehat{\gamma}(\beta^*)-\gamma(\beta_0)\right)'
\left(\widehat{\gamma}(\beta^*)-\gamma(\beta_0)\right)
\right] \\
&=
\operatorname{tr}\left\{
\operatorname{Var}\left(\widehat{\gamma}(\beta^*)\right)
\right\}
+
\left\|
\gamma(\beta^*)-\gamma(\beta_0)
\right\|_2^2 \\
&\quad
+
2
\left(\gamma(\beta^*)-\gamma(\beta_0)\right)'
\left(
\mathbb{E}\left[\widehat{\gamma}(\beta^*)\right]
-
\gamma(\beta^*)
\right) \\
&\quad
+
\left\|
\mathbb{E}\left[\widehat{\gamma}(\beta^*)\right]
-
\gamma(\beta^*)
\right\|_2^2,
\end{aligned}
\end{equation}
where \(\mathbb{E}[\cdot]\) is taken under the fixed sampling law
governing \(\widehat\gamma(\beta^*;S)\). 

When $\widehat{\gamma}(\beta^*;S)$ satisfies the decay condition in
\eqref{eq:target_variance_decay}, namely
\[
\mathbb{E}\!\left[
\left\lVert
\widehat{\gamma}(\beta^*;S)-\gamma(\beta^*;S)
\right\rVert_2^2
\right]=o(1),
\]
the sampling component of the target estimator is asymptotically negligible. Hence, for fixed
$\epsilon$,
\begin{equation}
\textit{MSE}(\beta^*,\beta_0,S)
=
\left\lVert
\gamma(\beta^*;S)-\gamma(\beta_0;S)
\right\rVert_2^2
+
o(1),
\end{equation}
uniformly over $\beta_0\in B_\epsilon(\beta^*;S)$.

This decay requirement is closely related to weak identification. Section~\ref{sec: implementing Admisible set} explains
how we operationalize this restriction using spectral Jacobian information, following
\citet{antoine2024gmm} and \citet{forneron2024detecting}.

\begin{theorem}\label{thm: minimax} Under Assumptions \ref{ass: regularity conditions}, \ref{ass: regularity assumptions reference point}, and \ref{ass: uniq}, the least-sensitive partition minimizes, asymptotically, the mean squared error under worst-case local miscalibration. Specifically, there exists a scalar $\mathcal{E} > 0$ such that for all $0 < \epsilon \leq \mathcal{E}$, 
\begin{equation}
	\min_{S\in\mathcal S_{\mathrm{AD}}(\eta^*)}
\max_{\beta_0\in B_\epsilon(\beta^\ast;S)}
\operatorname{MSE}(\beta^\ast,\beta_0,S)
=
\operatorname{WorstBias}(\beta^\ast,\epsilon;S^*)^2
+
o(1).
\end{equation}
\end{theorem} 
\begin{proof} See Appendix \ref{appendix:main-proofs}. \end{proof}

This interpretation is closely related to \citet{bonhomme2022minimizing}'s approach to robust estimation under misspecification.\footnote{Their notion of misspecification is more general than the classical formulation considered here and includes our setting as a special case.} In both frameworks, the objective can be viewed as minimizing a worst-case mean squared error. The main difference is the decision variable: \citet{bonhomme2022minimizing} optimize over influence functions, whereas we optimize over parameter partitions. A second difference is the asymptotic regime. Their misspecification neighborhoods shrink with the sample size, so bias and variance are asymptotically comparable, whereas in our setting the neighborhood is not indexed by $n$, so the asymptotic MSE is driven by the bias component.

A natural alternative to our sensitivity criterion is to select the partition by minimizing an MSE criterion. Nevertheless, implementing an MSE ranking would require a credible estimate of the full joint covariance matrix of all moments or reduced-form estimates. This is often unavailable when the moments are constructed from multiple datasets, as in many CMD applications. One could use an MSE-based criterion when this covariance matrix is available, or when a reliable joint bootstrap is feasible. In the \cite{nakamura2018high} application, however, different datasets are used to compute the reduced-form parameters entering the CMD estimation, which motivates our sensitivity-based approach.

\begin{remark} [\textbf{Why a local analysis?}]
\emph{First}, local methods are typically much easier to implement than global sensitivity methods, which is important in structural applications where each re-estimation can be costly. \emph{Second}, anchoring the analysis at a reference point is natural empirically, since applied work almost always begins with a first-stage estimation that delivers a candidate value for $\eta^*$. \emph{Third}, many structural analyses proceed under an implicit assumption of \emph{local} miscalibration: the calibrated (fixed) parameters are believed not to be too far from their population values, since otherwise counterfactual and policy conclusions would be difficult to interpret. This perspective is plausible in some applications---for example, for parameters such as discount factors that are often thought to lie in a relatively narrow range---but it need not hold universally. If $\eta^*$ is far from $\eta_0$, a local analysis becomes uninformative, and global sensitivity tools are more appropriate.
\end{remark} 

\begin{remark}[\textbf{Using the selected partition under recalibration}]
    After selecting the least-sensitive partition $S^*(\eta^*)$, the researcher will often want to use it together with an alternative calibration $\tilde{\beta}\neq\beta^*$. Keeping the selected partition fixed, this leads to the updated point
\[
\tilde{\eta}
\;:=\;
\big(\alpha(\tilde{\beta};S^*(\eta^*)),\tilde{\beta}\big).
\]
A natural question is whether the same partition remains optimal when the criterion is re-evaluated at $\tilde{\eta}$. Appendix~\ref{sec: identification S^*} studies this issue and gives sufficient conditions under which a nearby recalibration does not change the selected partition. In particular, it shows that if $\tilde{\eta}$ remains sufficiently close to $\eta^*$, then
\[
S^*(\tilde{\eta}) = S^*(\eta^*).
\]
This provides a formal justification for selecting the partition at a reference point and then using it with nearby alternative calibrations.
\end{remark}

\section{Estimation and Implementation}\label{sec: Implementation}

This section describes how to implement the partition-selection procedure in finite samples.
We begin by presenting the full algorithm.

\begin{algorithm}[htbp]
\caption{Least-Sensitive Partition}
\label{alg:ls-partition}

\begin{algorithmic}[1]

\State \textbf{Input:}
$(S_1,\beta_1)$, $\mathcal R$, $\{\Sigma_S\}_{S\in\mathcal S}$,
$\widehat g$, $\Gamma$, $\widehat W$, $n$

\vspace{0.3em}

\State \textbf{Step 1: First-stage estimation}
\State Compute
\[
\widehat\alpha(\beta_1;S_1)
\in \arg\min_{\alpha}
\widehat g(\alpha,\beta_1;S_1)'\widehat W \widehat g(\alpha,\beta_1;S_1)
\]
\State Set $\widehat\eta := (\widehat\alpha(\beta_1;S_1),\beta_1)$

\vspace{0.3em}

\State \textbf{Step 2: Derivatives at reference point}
\State Compute
\[
\widehat G_\eta = \left.\frac{\partial \widehat g(\eta)}{\partial \eta'}\right|_{\widehat\eta},
\qquad
\widehat \Gamma_\eta = \left.\frac{\partial \Gamma(\eta)}{\partial \eta'}\right|_{\widehat\eta}
\]

\vspace{0.3em}

\For{each $S \in \mathcal S \cap \mathcal R$}

    \State Compute singular values of $\widehat W^{1/2}\widehat G_{\alpha,S}$
    \State Define $\widehat r_S = \#\{j:\widehat\sigma_{j,S}>\tau_n\}$

    \If{$\widehat r_S = |S|$}
        \State Compute
        \[
        \widehat D_{\alpha\beta,S}
        =
        -(\widehat G_{\alpha}'\widehat W\widehat G_{\alpha})^{-1}
        \widehat G_{\alpha}'\widehat W\widehat G_{\beta}
        \]
        \State Compute
        \[
        \widehat D_{\gamma\beta,S}
        =
        \widehat\Gamma_{\alpha}\widehat D_{\alpha\beta,S}
        +
        \widehat\Gamma_{\beta}
        \]
        \State Compute sensitivity
        \[
        \widehat K_S
        =
        \sqrt{|S^c|}
        \left\|\widehat D_{\gamma\beta,S}\Sigma_S\right\|_2
        \]
    \EndIf

\EndFor

\vspace{0.3em}

\State \textbf{Step 5: Selection}
\State
\[
\widehat S^* \in \arg\min_S \widehat K_S
\]

\vspace{0.3em}

\State \textbf{Output:}
$\widehat{\mathcal S}_{AD}$, $\{\widehat K_S\}$, $\widehat S^*$

\end{algorithmic}
\end{algorithm}

The rest of this section discusses the three implementation choices that enter the algorithm:
how to obtain the reference point, how to construct the admissible set of partitions, and how
to choose the normalization matrix used in the sensitivity computation. 

Formal results on the consistency of the estimated least–sensitive partition are given in
Appendix \ref{sec: consistency}.

\subsection{First-stage estimation and the reference point}
\label{sec:first-stage-reference-point}

Starting from a first-stage calibration $(S_1,\beta_1)$, the reference point is
\[
\widehat{\eta}
:=
\bigl(\widehat\alpha(\beta_1;S_1),\beta_1\bigr),
\]
where $\widehat\alpha(\beta_1;S_1)$ is the minimum-distance estimate obtained under the
first-stage partition and calibration. We assume that the sample reference point satisfies
\[
\widehat{\eta} \xrightarrow{p} \eta^*.
\]

A primitive condition ensuring that $\eta^*\in\mathcal N_0$, as required in
Assumption~\ref{ass: regularity assumptions reference point}, is that the first-stage calibration
is compatible with the population equations. That is, there exists
$\alpha\in\mathcal A$ such that
\begin{equation}\label{eq: primitive of identified set reference point}
    g(\alpha,\beta_1;S_1)=0.
\end{equation}

This assumption that the reference point lies in the identified set is crucial for three reasons. 
First, it guarantees \emph{comparability} across partitions: for any candidate partition $S$ 
calibrating the same $\beta^*$, the conditional MD solution $\alpha(\beta^*;S)$ coincides with 
$\alpha^*$. This ensures that the points at which the local derivatives are evaluated for different 
MD problems are comparable. Second, it allows us to \emph{avoid re-estimating} the model for 
every candidate partition. Third, it allows us to omit the Hessian terms in the IFT formula, which are asymptotically negligible at a reference point in the identified set. The asymptotic justification of the direct derivative computation, 
which is detailed in the next subsection, relies on these properties.

In overidentified GMM or CMD settings, the  condition in (\ref{eq: primitive of identified set reference point}) is empirically testable. One can
use identification-robust calibration-validity tests, as in \cite{alegre2023robust}, or
related weak-identification robust methods such as \cite{forneron2024detecting}.

\subsection{Derivative computation without re-estimation}
The statistic $\widehat{D}_{\alpha \beta,S}$, proposed in the algorithm, should be interpreted as a direct plug-in computation of the
local derivative, not as the exact finite-sample derivative of every partition-specific
re-estimation map. To see the distinction, fix a candidate partition \(S\) and write the
reference point under the \(S\)-specific ordering as
\[
    \widehat \eta = (\widehat\alpha,\widehat\beta).
\]
A literal finite-sample computation of the slope of the re-estimation map would first solve
\[
\widehat\alpha(\widehat\beta;S)
\in
\arg\min_{\alpha\in A_S}
\widehat g(\alpha,\widehat\beta;S)'
\widehat W
\widehat g(\alpha,\widehat\beta;S),
\]
and then differentiate the corresponding first-order conditions with respect to the fixed
parameters. Doing this separately for each partition would require re-estimating the model many
times. Importantly, $\widehat\alpha$, the $S$-coordinates of the reference point, need not coincide with
$\widehat\alpha(\widehat\beta;S)$, the MD solution obtained after fixing $\widehat\beta$. 

The exact finite-sample derivative would also include additional Hessian terms weighted by the residuals
\(\widehat g(\widehat{\alpha}(\widehat{\beta}),\widehat{\beta})\); under Assumption~\ref{ass: regularity assumptions reference point}, which places the reference point in the identified set, these terms vanish asymptotically.

The algorithm avoids this re-estimation step. Instead, for each partition \(S\), it uses the derivative
bundle evaluated once at the common reference point \(\widehat\eta\) and computes
\[
\widehat D_{\alpha\beta}(S)
=
-
\left(
\widehat G_{\alpha}(\widehat\eta;S)'
\widehat W
\widehat G_{\alpha}(\widehat\eta;S)
\right)^{-1}
\widehat G_{\alpha}(\widehat\eta;S)'
\widehat W
\widehat G_{\beta}(\widehat\eta;S).
\]
The justification is asymptotic. Under the regularity conditions in
Appendix~\ref{app:direct-derivative}, and using the reference point in the identified set, 
\(\widehat D_{\alpha\beta}(S)\) provides a consistent asymptotic approximation to the population slope,
\[
\left.
\frac{\partial\widehat\alpha(\beta;S)}{\partial \beta'}
\right|_{\beta=\widehat\beta}
=
\widehat D_{\alpha\beta}(S)+o_p(1),
\qquad
\widehat D_{\alpha\beta}(S)
=
\frac{\partial \alpha(\beta^*;S)}{\partial \beta'}+o_p(1).
\]
Therefore, re-estimating \(\widehat\alpha(\widehat\beta;S)\) for every
candidate partition would only change the derivative by an \(o_p(1)\) term. The same
argument applies to the derivative of the target object.

\subsection{Practical guidance on the implementation of the admissible set of partitions}\label{sec: implementing Admisible set}
The goal of this subsection is to explain how we implement the conditional local identification and declining-variance conditions used in the definition of the admissible set of partitions. We implement both restrictions using singular values of the weighted Jacobian of the estimating equations. The basic idea is that small singular values of \(W^{1/2}G_{\alpha}\) make the inverse problem unstable, which can prevent the variance of the induced estimator of \(\gamma\) from declining with the sample size.

For a candidate partition \(S\), let
\[
C_S := (W^{1/2}G_{\alpha})^+,
\]
where \((\cdot)^+\) denotes the Moore--Penrose generalized inverse. The norm
\(\|\cdot\|_2\) denotes the spectral norm, i.e., the largest singular value of a matrix. When \(W^{1/2}G_{\alpha}\) has full column rank,
\begin{equation}
\label{eq:minimum-singular-value}
\|C_S\|_2
=
\frac{1}{\sigma_{\min}\{W^{1/2}G_{\alpha}\}}.
\end{equation}
Thus, small singular values of the weighted Jacobian correspond to a large Moore--Penrose inverse and indicate that the partition is close to a weak-identification region.

To keep notation light, we suppress the local-sequence index in the main text.\footnote{Strictly speaking, the population Jacobian and the matrix \(C_S\) should be indexed by \(n\) along the local sequence. Appendix \ref{appendix: weak identification and Sad} gives the formal local-sequence argument.} 

Under local curvature restrictions on \(g\) and regularity conditions on \(\Gamma\), Appendix \ref{appendix: weak identification and Sad} shows that if
\[
\|C_S\|_2=o(\sqrt{n}),
\]
then
\begin{equation}\label{eq: second moment of gamma declining bounds}
    \mathbb{E}\!\left[
    \|\widehat{\gamma}(\beta^*)-\gamma(\beta^*)\|^2
    \right]
    =
    o(1),
    \qquad
    \operatorname{tr}\!\left\{
    \operatorname{Var}\big(\widehat{\gamma}(\beta^*)\big)
    \right\}
    =
    o(1).
    \end{equation}

    Equations \eqref{eq:minimum-singular-value} and \eqref{eq: second moment of gamma declining bounds} motivate a hard-thresholding rule for the singular values of the estimated weighted Jacobian. For each partition \(S\), let
    \[
    \widehat\sigma_{1,S}\geq \widehat\sigma_{2,S}\geq \cdots
    \]
    denote the singular values of \(\widehat W^{1/2}\widehat G_{\alpha,S}\). We estimate the conditional rank of partition \(S\) by
    \begin{equation}
    \widehat r_S
    =
    \#\left\{
    j:\widehat\sigma_{j,S}>\tau_n
    \right\},
    \qquad
    \tau_n=\sqrt{\frac{\log n}{n}},
    \end{equation}
    where the threshold follows \cite{forneron2024detecting}. A partition is kept only if
    \[
    \widehat r_S=|S|.
    \]
    This rule removes partitions with exact rank failure and partitions whose weighted Jacobian is numerically close to rank deficient. The estimated admissible set is therefore
    \[
    \widehat{\mathcal S}_{AD}(\widehat{\eta})
    =
    \left\{
    S\in \mathcal S\cap \mathcal{R}:
    \widehat r_S=|S|
    \right\}.
    \]

\subsection{Practical Guidance on the Normalization Matrix}
\label{sec: practical guidance Sigma_S}

The normalization matrix $\Sigma_S$ defines the metric used to measure local miscalibrations of the fixed parameters. Its purpose is to make perturbations comparable across structural parameters that may be expressed in different units. This subsection explains how we construct this metric, and how we conduct robustness analysis.

Before evaluating candidate partitions, we assign each structural parameter $j$ an interval
\[
    [\eta_{j,\min},\eta_{j,\max}].
\]
We use the range $\eta_{j,\max}-\eta_{j,\min}$. To obtain the ranges, we can use one of the following options.

\begin{remark}[\textbf{Sources of Bounds}]
The intervals $(\eta_{j,\min},\eta_{j,\max})$ should be chosen before evaluating the candidate partitions and should be based on conservative restrictions. In the applications below we use three sources.

\begin{enumerate}[label=\roman*)]
\item \textbf{Parameter-space bounds.} When the model imposes natural bounds on a parameter, we can use those bounds directly. 

\item \textbf{Model-implied economic bounds.} When a parameter is more naturally disciplined through an economic object—for example, the steady-state interest rate—, we impose bounds on that object and map them into the structural parameter, say the discount factor. For example, bounds on steady-state markups can be translated into bounds on the elasticity of substitution across varieties.

\item \textbf{Well-established external intervals.} When direct external evidence is available, we can use published intervals or widely used conservative quantiles from the relevant literature. For example, in the empirical application we use external evidence on the slope of the Phillips curve given by \cite{hazell2022slope}.
\end{enumerate}
\end{remark}

Because the normalization matrix encodes what the researcher regards as a plausible
calibration error, robustness to this choice is an important part of the implementation.
The baseline normalization should therefore be interpreted as the main specification, not
as the only admissible metric. When several normalizations are equally defensible, the
researcher should verify whether the selected partition and the substantive conclusions
are stable across them.

A first robustness exercise is to replace the baseline normalization by a family of
plausible normalization matrices. Let \(\mathcal{M}_S\) denote the set of normalization
matrices considered plausible for partition \(S\). Instead of selecting the partition that
minimizes \(K_S(\eta^*,\Sigma_S)\) for a single matrix, one can compute
\[
S^{rob} \in \arg\min_{S\in \mathcal{S}_{\mathrm{AD}}(\eta^*)}
\max_{\Sigma_S\in \mathcal{M}_S} K_S(\eta^*,\Sigma_S).
\]

A second exercise is diagnostic. For each partition, define the normalized Jacobian
\[
A_S := D_{\gamma\beta,S}\Sigma_S,
\qquad
D_{\gamma\beta,S}:=\frac{\partial \gamma(\beta^*;S)}{\partial \beta'}.
\]
Let \(v_S\) be a right singular vector associated with the largest singular value of
\(A_S\), equivalently an eigenvector of \(A_S'A_S\) associated with its largest eigenvalue. The squared entries of $v_S$ quantify the proportion of the worst-case normalized perturbation allocated to each calibrated parameter.

These contribution measures identify which calibrated parameters drive the sensitivity
statistic. A targeted robustness check is then to widen the normalization bounds only for
the parameters with the largest contributions and recompute \(K_S\), the selected
partition, and the implied worst-case bounds. This stress test asks whether the empirical
conclusion depends on tight bounds for the parameters that matter most for worst-case
miscalibration. If the selected partition and the substantive conclusions remain stable
after widening those influential bounds, the robustness of the result is more credible.

\section{Empirical Application}\label{sec: empirical application}

\citet{nakamura2018high} provide evidence of monetary non-neutrality using high-frequency responses of real interest rates, expected inflation, and output growth to Federal Reserve announcements. A central implication of their analysis is that monetary policy affects the economy not only through conventional channels, but also by shifting market beliefs about broader economic fundamentals via \emph{information effects}. They study this mechanism by extending a New Keynesian model to incorporate an information effect and estimating the model by classical minimum distance (CMD), matching model-implied impulse response functions (IRFs) to reduced-form IRFs obtained via local projections.

We use their environment as a laboratory for our partition-selection criterion. Our object of interest is the information-effect parameter $\gamma$. Throughout, we hold fixed the model, the CMD objective, and the reduced-form moments; what varies across implementations is only the partition of structural parameters into estimated and calibrated (fixed) components. In this setting, the mapping $g(\cdot;\, S)$ depends on $S$ only through the relabeling of estimated and fixed parameters, while its functional form remains constant across partitions. The number of reduced-form parameters matched in CMD is $d_{m}=33$, and the number of structural parameters is $d_{\eta}=10$.

This application illustrates why we use the sensitivity-based statistic rather than a plug-in MSE criterion. The CMD reduced-form objects are constructed from different datasets, so the full joint covariance matrix of the stacked reduced-form parameters is not identified. In particular, the cross-covariances across moment blocks estimated from different datasets are not generally identified from the reported reduced-form estimates. 

\begin{table}[!t]
	\centering
	\begin{threeparttable}
		\caption{\textbf{Estimated and Fixed Parameters in }\cite{nakamura2018high}}
		\label{tab:parameters_NK18}
		\small
\setlength{\tabcolsep}{3pt}
\begin{tabular*}{\textwidth}{@{\extracolsep{\fill}}l p{0.52\textwidth} c@{}}
			\hline \hline
			\textit{Notation} & \textit{Definition} & \textit{Value} \\
			\hline \hline
			\multicolumn{3}{c}{\textit{Fixed parameters}} \\ \hline
			$\rho$   & \textit{Subjective discount factor} & 0.99 \\
			$\delta$ & \textit{Nominal rigidity} & 0.75 \\
			$\omega$ & \textit{Elasticity of marginal cost to output} & 2 \\
			$\theta$ & \textit{Elasticity of substitution across varieties} & 10 \\
			$\sigma$ & \textit{Intertemporal elasticity of substitution} & 0.50 \\
			$b$      & \textit{Consumption habit} & 0.90 \\
			\hline
			\multicolumn{3}{c}{\textit{Estimated parameters}} \\ \hline
			$\rho_{1}$ & \textit{Autoregressive first root of monetary shock} & 0.90 \\
			$\rho_{2}$ & \textit{Autoregressive second root of monetary shock} & 0.79 \\
			$\kappa \zeta \times 10^{5}$ & \textit{Phillips curve slope} & 11.2 \\
			$\gamma$ & \textit{Information effect} & 0.675 \\
			\hline \hline
		\end{tabular*}
		\begin{tablenotes}
			\item \small Notes: In the original paper, three additional fixed structural parameters (the inflation target shock, the endogenous feedback of the Taylor rule, and the coefficient of lagged inflation) are omitted here to streamline the analysis. 
		\end{tablenotes}
	\end{threeparttable}
\end{table}

Table~\ref{tab:parameters_NK18} reports the structural parameters and the values reported in \citet{nakamura2018high}. We evaluate all local objects at the empirical reference point implied by their estimates, and we denote the corresponding estimate of the information effect by $\widehat \gamma$. \cite{alegre2023robust} test the Nakamura--Steinsson benchmark calibration on a
feasible subset of moments and do not reject it, supporting its use as the local reference point.

\subsection{Informative Range Normalization}
In this New Keynesian model, credible external information about the calibrated parameters is available, so we can use it to enrich the sensitivity analysis by tightening the bounds used in the normalization.

Because ad hoc modifications of the normalization ranges could potentially change the analysis in arbitrary ways, it is important to be transparent about how these tighter bounds are chosen. It is also important to conduct robustness checks for any parameter whose normalization plays a central role in determining the selected partition.

We proceed as follows. First, we explain, parameter by parameter, how the informative bounds are chosen. We then present the main results. Finally, we conduct a robustness exercise for the discount-factor bounds in Subsection \ref{sec:Robustnesssigma_S}, since this parameter is a key driver of the main results.

\subsubsection{Discount factor \texorpdfstring{$\rho$}{rho}.} We follow \citet{schorfheide2000loss} and set $\rho\sim \mathrm{Beta}(\mu=0.993,\ \sigma=0.002)$.
The normalization bounds are the $1\%$ and $99\%$ quantiles of this prior.
Similar tight priors are standard in estimated DSGE work; see, e.g.,
\citet{lubik2005bayesian}, \citet{smets2007shocks}, and \citet{an2007bayesian}.
This reflects the view that the discount factor is tightly pinned down by the steady-state relationship with the steady-state intereste rate.

\subsubsection{Nominal rigidity \texorpdfstring{$\delta$}{delta}.}
We follow \citet{smets2007shocks} and set $\delta\sim \mathrm{Beta}(\mu=0.5,\ \sigma=0.1)$, taking
the $0.5\%$ and $99.5\%$ quantiles as bounds.
Closely related priors are used by 
\citet{rabanal2005comparing}, and \citet{lubik2005bayesian}.

\subsubsection{Intertemporal elasticity parameter \texorpdfstring{$\sigma$}{sigma}.}
Following \citet{smets2007shocks}, we place a Normal prior on $1/\sigma$:
\begin{equation}
\frac{1}{\sigma}\sim \mathcal{N}(1.5,\ 0.375).
\end{equation}
We take the $2.5\%$ and $97.5\%$ quantiles of $1/\sigma$ and invert them to obtain the implied bounds
$[\sigma_{\min},\sigma_{\max}]$ (noting that inversion typically yields an asymmetric interval in
$\sigma$).
\subsubsection{Habit formation \texorpdfstring{$b$}{b}.}
As in \citet{smets2007shocks}, we set $b\sim \mathrm{Beta}(\mu=0.7,\ \sigma=0.1)$ and use the $1\%$
and $99\%$ quantiles as bounds.
These values capture the empirically plausible region in which habits provide internal propagation
without implying nearly unit root consumption dynamics.

\subsubsection{Monetary shock persistence \texorpdfstring{$\rho_1$ and $\rho_2$}{rho1 and rho2}.}
For the autoregressive roots governing the monetary shock, we use the natural positive bounds implied by
stationarity and the model's parametric restrictions.

\subsubsection{Elasticity of marginal cost to output \texorpdfstring{$\omega$}{omega}.}
We follow \citet{smets2007shocks} by anchoring $\omega$ to more primitive objects: the labor share $a$
and the Frisch elasticity $v$.
Using the model-implied mapping $\omega= (v^{-1}+1-a)/a$, we propagate
uncertainty from $(a,v)$ into $\omega$.
Concretely, we set $a\in\{1/2,\ 3/4\}$ and
\begin{equation}
v\sim \mathcal{N}(2,\ 0.75),
\end{equation}
and compute the implied bounds on $\omega$ from conservative quantiles of $v$ ($2.5\%$ and $97.5\%$) and the endpoints for
$a$.
This construction yields economically interpretable variation in $\omega$ while remaining grounded in
well-studied labor-supply and factor-share objects.

\subsubsection{Slope of the Phillips curve.}
We discipline the range for the Phillips-curve slope using
the full-sample tradeable-demand IV estimate from \citet{hazell2022slope}. Since the
units differ, we map it into \cite{nakamura2018high} units using the following model-implied relationship:
\[
\kappa\zeta
=
\frac{v}{1+\omega\theta}\kappa^{Hazell},
\]
where \(v\) is the Frisch labor-supply elasticity, \(\omega\) is the elasticity of
marginal cost with respect to output, and \(\theta\) is the elasticity of substitution
across varieties. Using \(v=1\), which corresponds to a labor share of $2/3$ when $\omega=2$, the confidence
interval yields
\[
\kappa\zeta\in [6.30,52.73]\times 10^{-5}.
\]

\subsubsection{Elasticity of substitution across varieties \texorpdfstring{$\theta$}{theta}.}
We use the steady-state relationship between the markup $M$ and $\theta$,
\begin{equation}
M \;=\; \frac{\theta}{\theta-1},
\end{equation}
and choose $M\in[1.05,\ 1.30]$, i.e.\ markups between $5\%$ and $30\%$.
We then map this interval into bounds for $\theta$.
This pins down $\theta$ using a primitive object (markups) with a clear economic interpretation and
avoids implausible values that would imply either near-perfect competition or excessively high steady-state
profit shares.
\begin{table}[!t]
	\centering
	\begin{threeparttable}
		\caption{\textbf{Ranges used for normalization of parameter units}}
		\label{tab:priors_NK18}
		\renewcommand{\arraystretch}{1.5}
		\begin{tabular}{@{}lccc@{}}
			\hline \hline
			\textit{Parameter} & \textit{Definition} & $\eta_{j,\min}$ & $\eta_{j,\max}$ \\
			\hline \hline
			$\rho$     & \textit{Subjective discount factor} & 0.987 & 0.997 \\
			$\delta$   & \textit{Nominal rigidity} & 0.25 & 0.75 \\
			$\omega$   & \textit{Elasticity of marginal cost to output} & 0.72 & 4.77 \\
			$\theta$   & \textit{Elasticity of substitution across varieties} & 4.3 & 21 \\
			$\sigma$   & \textit{Intertemporal elasticity of substitution} & 0.447 & 1.307 \\
			$\rho_{1}$ & \textit{Autoregressive first root of monetary shock} & 0 & 1 \\
			$\rho_{2}$ & \textit{Autoregressive second root of monetary shock} & 0 & 1 \\
			$\kappa \zeta \times 10^{5}$ & \textit{Phillips curve slope} & 6.3 & 52.73 \\
			$b$        & \textit{Consumption habit} & 0.36 & 0.94 \\
			\hline \hline
		\end{tabular}
	\end{threeparttable}
\end{table}

\subsection{Main Results}
Using the hard-thresholding rank estimator proposed in Section \ref{sec: implementing Admisible set}, we estimate the set of admissible partitions. In the present application, this procedure yields a total of 60 admissible partitions. For readability, we present a representative subset in the main text; full tables are reported in \ref{appendix:MainresultsNS18}.

For each admissible partition $S$, we compute the sensitivity statistic $\widehat{K}(\widehat{\beta};\,S)$, which measures the magnitude of worst-case local bias in $\gamma$ induced by small miscalibrations. To aid interpretation, we report \emph{$5\%$ worst-case bounds}, defined as $\widehat \gamma \pm 0.05\,\widehat{K}(\widehat{\beta};\,S)$, which correspond to a first-order approximation of the worst-case bias of size $\epsilon = 5\%$. 

Table~\ref{tab:ns_sensitivity_informative} reports these statistics for the representative subset of admissible partitions. For comparison, we also report the implementation in \citet{nakamura2018high} (labeled \textit{Original}), even though it is not admissible under our rank hard-thresholding criterion.

\begin{table}[!t]
	\centering
	\begin{threeparttable}
		\caption{\textbf{Admissible partitions and sensitivity statistic $\widehat{K}(\widehat{\beta};\,S)$}}
		\label{tab:ns_sensitivity_informative}
		\renewcommand{\arraystretch}{1.5}
		\begin{tabular}{@{}lcccc@{}}
			\hline \hline
			\textit{Name} & \textit{Estimated} & \textit{Calibrated} & $\widehat{K}(\widehat{\beta};\,S)$ & $5\%$ bounds \\
			\hline \hline
			\textit{Original} & $( \rho_1, \rho_2, \gamma, \kappa\zeta)$ &$(\rho,\delta,\omega,\theta,\sigma,b) $ & 1.82 & [0.58, 0.768] \\
			$S^*$ & $(\delta, \sigma, \rho_{1}, \rho_{2}, \gamma, b)$ & $(\rho, \omega, \theta, \kappa\zeta)$ & 0.06 & [0.67, 0.68] \\
			$S_{2}$ & $(\rho, \delta, \rho_{1}, \rho_{2}, \gamma, b)$ & $(\omega, \theta, \sigma, \kappa\zeta)$ & 0.32 & [0.66, 0.69] \\
			$S_{3}$ & $(\rho, \rho_{1}, \rho_{2}, \gamma, b)$ & $(\delta, \omega, \theta, \sigma, \kappa\zeta)$ & 0.36 & [0.66, 0.70] \\
			$S_{4}$ & $(\rho, \delta, \rho_{2}, \gamma, b)$ & $(\omega, \theta, \sigma, \rho_{1}, \kappa\zeta)$ & 0.38 & [0.66, 0.70] \\
			$S_{12}$ & $(\rho, \delta, \sigma, \rho_{1}, \rho_{2}, \gamma)$ & $(\omega, \theta, \kappa\zeta, b)$ & 1.05 & [0.63, 0.73] \\
			$S_{14}$ & $(\rho, \sigma, \rho_{1}, \rho_{2}, \gamma, b)$ & $(\delta, \omega, \theta, \kappa\zeta)$ & 1.19 & [0.62, 0.74] \\
			$S_{25}$ & $(\delta, \rho_{1}, \rho_{2}, \gamma)$ & $(\rho, \omega, \theta, \sigma, \kappa\zeta, b)$ & 1.82 & [0.59, 0.77] \\
			$S_{39}$ & $(\rho, \sigma, \rho_{2}, \gamma, b)$ & $(\delta, \omega, \theta, \rho_{1}, \kappa\zeta)$ & 3.70 & [0.49, 0.86] \\
			$S_{46}$ & $(\rho, \delta, \sigma, \gamma)$ & $(\omega, \theta, \rho_{1}, \rho_{2}, \kappa\zeta, b)$ & 9.25 & [0.22, 1.14] \\
			$S_{60}$ & $(\gamma)$ & $(\rho, \delta, \omega, \theta, \sigma, \rho_{1}, \rho_{2}, \kappa\zeta, b)$ & 11.07 & [0.12, 1.23] \\
			\hline \hline
		\end{tabular}
	\end{threeparttable}
\end{table}
Table~\ref{tab:ns_sensitivity_informative} reports the sensitivity statistic
\(\widehat K(\widehat\beta;S)\) for a representative subset of admissible
partitions. Across the full
set of admissible partitions, \(\omega\), \(\theta\), and the Phillips-curve
slope \(\kappa_\zeta\) are never estimated. This suggests that these parameters
are weakly identified in the local rank sense used to construct the admissible
set. 

Conditional on this admissibility restriction, the choice of partition matters
substantially. The least-sensitive partition,
\(S^{\ast}\), estimates
\((\delta,\sigma,\rho_1,\rho_2,\gamma,b)\) and fixes
\((\rho,\omega,\theta,\kappa_\zeta)\). It delivers
\(\widehat K(\widehat\beta;S^{\ast})=0.06\), so a \(5\%\) normalized
miscalibration generates only the local worst-case interval
\([0.67,0.68]\). At the opposite end of the admissible set, the statistic rises
to \(11.07\), producing the much wider interval \([0.12,1.23]\). The original
Nakamura--Steinsson implementation is not admissible under the rank criterion
and has \(\widehat K=1.82\). Hence the original partition is not in the extreme tail, but
it is far from the least-sensitive admissible partition. 

This heterogeneity is not well summarized by the number of estimated
parameters. Among the largest admissible partitions, changing the identity of a
single calibrated parameter moves the statistic from \(0.06\) to values above
one. Fixing \(\rho\), as in \(S^{\ast}\), is much less costly than fixing
\(\sigma\), \(b\), or \(\delta\). The contribution table clarifies why: in the
low- and middle-sensitivity partitions, worst-case bias is typically
concentrated in one preference or propagation parameter, such as
\(\rho\), \(\sigma\), \(\delta\), or \(b\). Robustness therefore depends less on
estimating many parameters per se than on estimating the parameters through
which calibration errors are most strongly transmitted to \(\gamma\).

The role of the discount factor illustrates this point. In \(S^{\ast}\), almost
all remaining sensitivity comes from the calibrated discount factor \(\rho\).
This might suggest that the result is mechanically driven by the tight
normalization range imposed on \(\rho\). The comparison across partitions argues
against this interpretation. Several low-sensitivity alternatives, such as
\(S_2\), \(S_3\), and \(S_4\), estimate \(\rho\) and still deliver small values
of \(\widehat K\). Moreover, the robustness exercise in
Table~\ref{tab: Robustness to discount factor} shows directly how much the
discount-factor range can be widened before \(S^{\ast}\) ceases to be the
least-sensitive partition. Thus, the tight bound on \(\rho\) helps explain the
ranking, but it does not trivially determine the conclusion.

Finally, the high-sensitivity tail has a clear economic structure. Extreme
fragility appears when the persistence of the monetary-policy shock is imposed
by calibration rather than disciplined by the data. Once both persistence roots
are fixed, \(\widehat K\) jumps into the range between roughly \(9\) and \(11\),
and the contribution decomposition shows that \(\rho_1\), with a smaller role
for \(\rho_2\), accounts for almost all of the statistic. Thus, for the
information-effect parameter, the key positive prescription is to estimate the
monetary-shock persistence parameters, especially \(\rho_1\). By contrast,
calibrating the weakly identified Phillips-curve and markup-related block
\((\omega,\theta,\kappa_\zeta)\) appears comparatively harmless for local
robustness.
\begin{table}[!t]
	\centering
	\begin{threeparttable}
		\caption{\textbf{Contribution of fixed parameters to the sensitivity statistic}}
		\label{tab:ns_contributions_informative}
		\renewcommand{\arraystretch}{1.5}
		\begin{tabular}{@{}lcccccccccc@{}}
			\hline \hline
			\textit{Name} & $\rho$ & $\delta$ & $\omega$ & $\theta$ & $\sigma$ & $\rho_{1}$ & $\rho_{2}$ & $\kappa\zeta$ & $\gamma$ & $b$ \\
			\hline \hline
			\textit{Original} & 0.0\% & 0.0\% & 0.0\% & 0.0\% & 0.0\% & - & - & - & - & 100.0\% \\
			$S^*$ & 99.6\% & - & 0.4\% & 0.0\% & - & - & - & 0.0\% & - & - \\
			$S_{2}$ & - & - & 0.0\% & 0.0\% & 100.0\% & - & - & 0.0\% & - & - \\
			$S_{3}$ & - & 22.9\% & 2.0\% & 4.3\% & 70.7\% & - & - & 0.2\% & - & - \\
			$S_{4}$ & - & - & 0.0\% & 0.0\% & 87.6\% & 12.4\% & - & 0.0\% & - & - \\
            $S_{12}$ & - & - & 0.0\% & 0.0\% & - & - & - & 0.0\% & - & 100.0\% \\
            $S_{14}$ & - & 78.0\% & 6.8\% & 14.7\% & - & - & - & 0.5\% & - & - \\
			$S_{25}$ & 0.0\% & - & 0.0\% & 0.0\% & 0.0\% & - & - & 0.0\% & - & 100.0\% \\
			$S_{39}$ & - & 10.1\% & 0.9\% & 1.9\% & - & 87.1\% & - & 0.1\% & - & - \\
			$S_{46}$ & - & - & 0.0\% & 0.0\% & - & 80.0\% & 18.1\% & 0.0\% & - & 1.9\% \\
			$S_{60}$ & 0.0\% & 0.0\% & 0.0\% & 0.0\% & 0.0\% & 82.9\% & 15.0\% & 0.0\% & - & 2.1\% \\
			\hline \hline
		\end{tabular}
	\end{threeparttable}
\end{table}
Table~\ref{tab:ns_contributions_informative} decomposes
\(\widehat K(\widehat\beta;S)\) into the relative contribution of each fixed
parameter. The main message is that sensitivity is typically highly
concentrated: in most partitions, one fixed parameter accounts for almost all
of the statistic. In the least-sensitive partition \(S^{\ast}\), essentially
all residual sensitivity comes from the discount factor \(\rho\), which accounts
for \(99.6\%\) of the total. In \(S_2\), the contribution is entirely due to
\(\sigma\). In \(S_3\), most of the contribution comes from \(\sigma\)
(\(70.7\%\)), with a secondary role for nominal rigidity \(\delta\)
(\(22.9\%\)). The same concentration appears among admissible partitions that
estimate the maximum number of parameters: in \(S_{12}\), the sensitivity is entirely driven
by habit formation \(b\), whereas in \(S_{14}\), most of the sensitivity comes
from \(\delta\), with smaller contributions from \(\omega\) and \(\theta\).

Another important feature is that the parameters that are always fixed because of weak-identification concerns---\(\omega\), \(\theta\), and \(\kappa_\zeta\)---have close to zero contribution to worst-case bias. In particular, the Phillips curve slope \(\kappa_\zeta\) contributes essentially zero, and \(\omega\) and \(\theta\) are almost never dominant sources of sensitivity. Thus, although these parameters must be calibrated, local miscalibration of them would not have large effects on the object of interest. By contrast, in the low- and middle-sensitivity regions of the admissible set, most of the sensitivity comes from preference and propagation parameters such as \(\rho\), \(\sigma\), \(\delta\), and \(b\).

The high-sensitivity tail has a different structure. Once monetary-shock
persistence is fixed rather than estimated, the contribution of \(\rho_1\)
becomes dominant. In \(S_{39}\), for example, \(\rho_1\) accounts for
\(87.1\%\) of total sensitivity. When both persistence roots are fixed, as in
\(S_{46}\), \(\rho_1\) accounts for \(80.0\%\) and \(\rho_2\) for \(18.1\%\).
The same pattern appears in the most fragile partitions: in \(S_{58}\),
\(\rho_1\) and \(\rho_2\) account for \(85.2\%\) and \(14.2\%\), respectively;
in \(S_{59}\), for \(77.4\%\) and \(22.6\%\); and in \(S_{60}\), for \(82.9\%\)
and \(15.0\%\). The economic interpretation is therefore clear: the most
fragile partition strategies are those that force the persistence of the monetary
shock to come from calibration rather than estimation. From a robustness
perspective, the most valuable discipline is to estimate the shock-persistence
parameters, especially \(\rho_1\), while calibration errors in the slope of the Phillips curve,
\(\kappa_\zeta\), and other parameters such as \(\omega\) and \(\theta\) are comparatively harmless.

Because the estimated admissible set depends on the singular-value cutoff used in the rank estimator, Appendix \ref{appendix: weakid_tolerance} reports a robustness exercise with respect to this implementation choice. In that exercise, we replace the baseline cutoff $\tau_n=(\log n/n)^{1/2}$ by $\tau_n(a)=(\log n/n)^a$ for $a\in\{0.50,0.60,0.70,0.80,0.90,1.00\}$, recompute the admissible set, and rank partitions again by the sensitivity statistic. The least-sensitive partition remains $S^*$ for all thresholds considered. For $a\leq 0.80$, both the admissible set and the top-five ranking are unchanged; more permissive cutoffs admit additional partitions but do not overturn the main ranking.

\subsection{Robustness to discount-factor normalization bounds}
\label{sec:Robustnesssigma_S}

The discount-factor normalization is a natural place to check robustness. $\rho$ maps into the steady-state annual nominal interest rate $r$ according to
\begin{equation}
\rho \;=\; \frac{1}{1+r/400}
\qquad\Longleftrightarrow\qquad
r \;=\; 400\left(\frac{1}{\rho}-1\right).
\end{equation}
Thus, small numerical changes in the range for $\rho$ can correspond to large
economic changes in steady-state interest rates.

The exercise below asks how far the discount-factor bound must be relaxed before the
partition ranking changes. For each value of $\rho_{\min}$ we
recompute the sensitivity statistics for the least-sensitive partition, its closest competitors, and the
original partition, and report the annual rate implied by $\rho_{\min}$. The baseline
selected partition remains least-sensitive until the lower bound is pushed below roughly
$\rho_{\min}=0.945$. At that point, the implied steady-state annual nominal
rate is above $23\%$, which is difficult to interpret as a plausible benchmark for the economies
typically targeted by this class of models.

\begin{table}[!t]
\centering
\begin{threeparttable}
\caption{\textbf{Sensitivities with different normalization bounds of the discount factor}}
\label{tab: Robustness to discount factor}
\scriptsize
\setlength{\tabcolsep}{3pt}
\renewcommand{\arraystretch}{1.2}
\begin{tabular*}{\textwidth}{@{\extracolsep{\fill}}lcccccc@{}}
\hline \hline
\textit{$\rho_{min}$} & $S^*$ & $S_2$ & $S_{3}$ & $S_{4}$ & \textit{Original} & \textit{Implied S.S. rate} \\
\hline \hline
0.990 & 0.04 & 0.32 & 0.36 & 0.38 & 1.82 & 4.04\% \\
0.985 & 0.07 & 0.32 & 0.36 & 0.38 & 1.82 & 6.09\% \\
0.980 & 0.10 & 0.32 & 0.36 & 0.38 & 1.82 & 8.16\% \\
0.975 & 0.13 & 0.32 & 0.36 & 0.38 & 1.82 & 10.26\% \\
0.970 & 0.16 & 0.32 & 0.36 & 0.38 & 1.82 & 12.37\% \\
0.965 & 0.19 & 0.32 & 0.36 & 0.38 & 1.82 & 14.51\% \\
0.960 & 0.22 & 0.32 & 0.36 & 0.38 & 1.82 & 16.67\% \\
0.955 & 0.24 & 0.32 & 0.36 & 0.38 & 1.82 & 18.85\% \\
0.950 & 0.27 & 0.32 & 0.36 & 0.38 & 1.82 & 21.05\% \\
0.945 & 0.30 & 0.32 & 0.36 & 0.38 & 1.82 & 23.28\% \\
0.940 & 0.33 & 0.32 & 0.36 & 0.38 & 1.82 & 25.53\% \\
\hline \hline
\end{tabular*}
\end{threeparttable}
\end{table}

\section{Monte Carlo Simulation} \label{sec: Monte carlo Simulation}

This section uses Monte Carlo simulations to study the finite-sample implications of the partition-selection criterion in the \cite{nakamura2018high} application. For details in the implementation of the Monte Carlo simulation see \ref{app:mc-implementation}. 

 The goal of the simulations is to assess how informative the $K_S$-based sensitivity ranking is for finite-sample bias, variance, and MSE. We generate repeated samples, of sizes $n=150$ and $n=500$, from a data-generating process (DGP) calibrated to the reference parameter vector used in the empirical exercise. The Monte Carlo setting is useful because it allows us to compute finite-sample distributions and relevant statistics---such as the variance and mean squared error (MSE) of $\hat{\gamma}$---that are difficult to obtain in the empirical application, where the reduced-form parameters entering the CMD estimation are constructed from different datasets and their joint variance--covariance matrix is not readily available.

The exercise proceeds as follows. For each admissible partition $S$, we first compute the local worst-case direction at the reference point and then scale the perturbation by the miscalibration size $\varepsilon \in \{5\%,10\%,15\%\}$. We then re-estimate the model under partition $S$ using this approximated worst-case miscalibration. Repeating this procedure across Monte Carlo samples delivers, for each
partition, the finite-sample distribution of $\hat{\gamma}$ under its \emph{own} worst-case miscalibration. We
summarize these distributions by reporting the variance and the corresponding
worst-case MSE, and by plotting histograms that make the differences across partitions visually
transparent.

\begingroup
\scriptsize
\setlength{\tabcolsep}{1.8pt}
\renewcommand{\arraystretch}{1.10}

\begin{longtable}{
@{}
>{\raggedright\arraybackslash}p{0.095\textwidth}
>{\raggedright\arraybackslash}p{0.145\textwidth}
>{\raggedright\arraybackslash}p{0.13\textwidth}
!{\hspace{0.2em}\vrule width 0.25pt\hspace{0.2em}}
rrr
!{\hspace{0.2em}\vrule width 0.25pt\hspace{0.2em}}
rrr
!{\hspace{0.2em}\vrule width 0.25pt\hspace{0.2em}}
rrr
@{}}

\caption{Finite-sample performance by partition, sample size, and miscalibration size}
\label{tab:mc_bound_aware_mse_multi_n_reduced}\\

\toprule
Partition 
& Estimated parameters 
& Calibrated parameters 
& \multicolumn{3}{c}{$\varepsilon = 5\%$} 
& \multicolumn{3}{c}{$\varepsilon = 10\%$} 
& \multicolumn{3}{c}{$\varepsilon = 15\%$} \\
\cmidrule(lr){4-6}
\cmidrule(lr){7-9}
\cmidrule(lr){10-12}
& & 
& Bias & Var. & MSE 
& Bias & Var. & MSE 
& Bias & Var. & MSE \\
\midrule
\endfirsthead

\toprule
Partition 
& Estimated 
& Calibrated 
& \multicolumn{3}{c}{$\varepsilon = 5\%$} 
& \multicolumn{3}{c}{$\varepsilon = 10\%$} 
& \multicolumn{3}{c}{$\varepsilon = 15\%$} \\
\cmidrule(lr){4-6}
\cmidrule(lr){7-9}
\cmidrule(lr){10-12}
& & 
& Bias & Var. & MSE 
& Bias & Var. & MSE 
& Bias & Var. & MSE \\
\midrule
\endhead

\midrule
\multicolumn{12}{r}{\emph{Continued on next page}}\\
\endfoot

\bottomrule
\endlastfoot

\multicolumn{12}{@{}l}{\textbf{Panel A.} Simulated sample size $n = 150$}\\
\addlinespace[0.20em]

$S^{\star}$ 
& $(\delta, \sigma, \rho_{1}, \rho_{2}, \gamma, b)$ 
& $(\rho, \omega, \theta, \kappa\zeta)$ 
& 42.5 & 4.6 & 4.7 
& 77.3 & 4.9 & 5.5 
& 115.4 & 5.4 & 6.7 \\

$S_{2}$ 
& $(\rho, \delta, \rho_{1}, \rho_{2}, \gamma, b)$ 
& $(\omega, \theta, \sigma, \kappa\zeta)$ 
& 196.2 & 1.9 & 5.7 
& -269.4 & 1.4 & 8.6 
& -458.6 & 1.4 & 22.4 \\

$S_{3}$ 
& $(\rho, \rho_{1}, \rho_{2}, \gamma, b)$ 
& $(\delta, \omega, \theta, \sigma, \kappa\zeta)$ 
& 303.4 & 1.5 & 10.7 
& -267.2 & 1.7 & 8.8 
& -339.8 & 1.3 & 12.9 \\

$S_{4}$ 
& $(\rho, \delta, \rho_{2}, \gamma, b)$ 
& $(\omega, \theta, \sigma, \rho_{1}, \kappa\zeta)$ 
& 438.8 & 1.8 & 21.1 
& -1867.4 & 0.6 & 349.3 
& -3861.6 & 1.6 & 1492.8 \\

$S_{5}$ 
& $(\sigma, \rho_{1}, \rho_{2}, \gamma, b)$ 
& $(\rho, \delta, \omega, \theta, \kappa\zeta)$ 
& 882.4 & 1.8 & 79.7 
& 2094.4 & 10.3 & 448.9 
& 2427.0 & 13.8 & 602.8 \\

$S_{6}$ 
& $(\delta, \rho_{1}, \rho_{2}, \gamma, b)$ 
& $(\rho, \omega, \theta, \sigma, \kappa\zeta)$ 
& -291.5 & 1.5 & 10.0 
& -576.0 & 1.7 & 34.9 
& -879.5 & 2.0 & 79.4 \\

Original 
& $(\rho_{1}, \rho_{2}, \kappa\zeta, \gamma)$ 
& $(\rho, \delta, \omega, \theta, \sigma, b)$ 
& 951.7 & 0.8 & 91.4 
& 1880.6 & 0.6 & 354.3 
& 2389.4 & 0.3 & 571.3 \\

\midrule
\multicolumn{12}{@{}l}{\textbf{Panel B.} Simulated sample size $n = 500$}\\
\addlinespace[0.20em]

$S^{\star}$ 
& $(\delta, \sigma, \rho_{1}, \rho_{2}, \gamma, b)$ 
& $(\rho, \omega, \theta, \kappa\zeta)$ 
& 32.6 & 1.2 & 1.4 
& 66.0 & 1.3 & 1.8 
& 102.5 & 1.4 & 2.5 \\

$S_{2}$ 
& $(\rho, \delta, \rho_{1}, \rho_{2}, \gamma, b)$ 
& $(\omega, \theta, \sigma, \kappa\zeta)$ 
& 198.7 & 0.6 & 4.5 
& -256.6 & 0.5 & 7.0 
& -453.5 & 0.5 & 21.0 \\

$S_{3}$ 
& $(\rho, \rho_{1}, \rho_{2}, \gamma, b)$ 
& $(\delta, \omega, \theta, \sigma, \kappa\zeta)$ 
& 301.4 & 0.6 & 9.6 
& -259.5 & 0.6 & 7.3 
& -334.7 & 0.5 & 11.7 \\

$S_{4}$ 
& $(\rho, \delta, \rho_{2}, \gamma, b)$ 
& $(\omega, \theta, \sigma, \rho_{1}, \kappa\zeta)$ 
& 445.6 & 0.6 & 20.4 
& -1857.2 & 0.2 & 345.1 
& -3861.7 & 3.1 & 1494.3 \\

$S_{5}$ 
& $(\sigma, \rho_{1}, \rho_{2}, \gamma, b)$ 
& $(\rho, \delta, \omega, \theta, \kappa\zeta)$ 
& 880.0 & 0.6 & 78.1 
& 2051.6 & 17.0 & 437.9 
& 2430.0 & 14.8 & 605.3 \\

$S_{6}$ 
& $(\delta, \rho_{1}, \rho_{2}, \gamma, b)$ 
& $(\rho, \omega, \theta, \sigma, \kappa\zeta)$ 
& -286.3 & 0.5 & 8.7 
& -569.0 & 0.5 & 32.9 
& -870.6 & 0.6 & 76.4 \\

Original 
& $(\rho_{1}, \rho_{2}, \kappa\zeta, \gamma)$ 
& $(\rho, \delta, \omega, \theta, \sigma, b)$ 
& 945.7 & 0.3 & 89.7 
& 1877.9 & 0.2 & 352.9 
& 2392.8 & 0.1 & 572.7 \\

\end{longtable}
\endgroup

Table~\ref{tab:mc_bound_aware_mse_multi_n_reduced} delivers three main lessons. First, the $K_S$-based ranking is highly informative about finite-sample worst-case bias. The least-sensitive partition, $S^{\star}$, has the smallest absolute bias and the lowest MSE for every miscalibration size and for both sample sizes. This is not mechanically imposed by the criterion, which is based on a local worst-case bias approximation rather than on simulated finite-sample MSE. The Monte Carlo results therefore support the interpretation of $K_S$ as a useful finite-sample diagnostic of calibration risk.

Second, the results make clear that sampling variance is not the main source of the differences across partitions. Increasing the sample size from $n=150$ to $n=500$ reduces the variance of $\hat{\gamma}$ substantially, but it leaves the miscalibration-induced bias largely unchanged. As a result, larger samples do not eliminate the consequences of fixing sensitive parameters at incorrect values. Instead, they make the bias component more visible. For $S^{\star}$, the MSE remains small even at $\varepsilon=15\%$, whereas the MSE of more fragile partitions grows rapidly as the miscalibration size increases.

Third, the deterioration in performance is strongly partition-specific. Partitions such as $S_2$, $S_3$, and $S_6$ generate moderate biases relative to the most fragile cases, but they still perform worse than $S^{\star}$ across the table. By contrast, partitions such as $S_4$, $S_5$, and the original Nakamura--Steinsson implementation can generate very large worst-case biases, especially at $\varepsilon=10\%$ and $\varepsilon=15\%$. The original partition is therefore not fragile because of sampling variability; it is fragile because the fixed block contains parameters through which calibration errors are strongly transmitted to $\gamma$.

Figure~\ref{fig:mc_histograms_gamma} visualizes the same pattern. For very small miscalibrations,
the distributions of \(\hat\gamma\) under the admissible partitions are relatively close to one
another, so sampling variation can obscure the gains from reducing worst-case bias. As
\(\epsilon\) increases, the distributions separate sharply. The original partition moves far away
from the true value, while the distribution under \(S^*\) remains much closer to it. Increasing
the sample size tightens the distributions, but it does not remove the displacement generated
by fixing miscalibrated parameters. This explains why the advantage of \(S^*\) becomes clearer
for larger samples and larger miscalibrations: sampling uncertainty falls, while calibration bias
remains.

\begin{figure}[H]
    \centering
    \includegraphics[width=\textwidth]{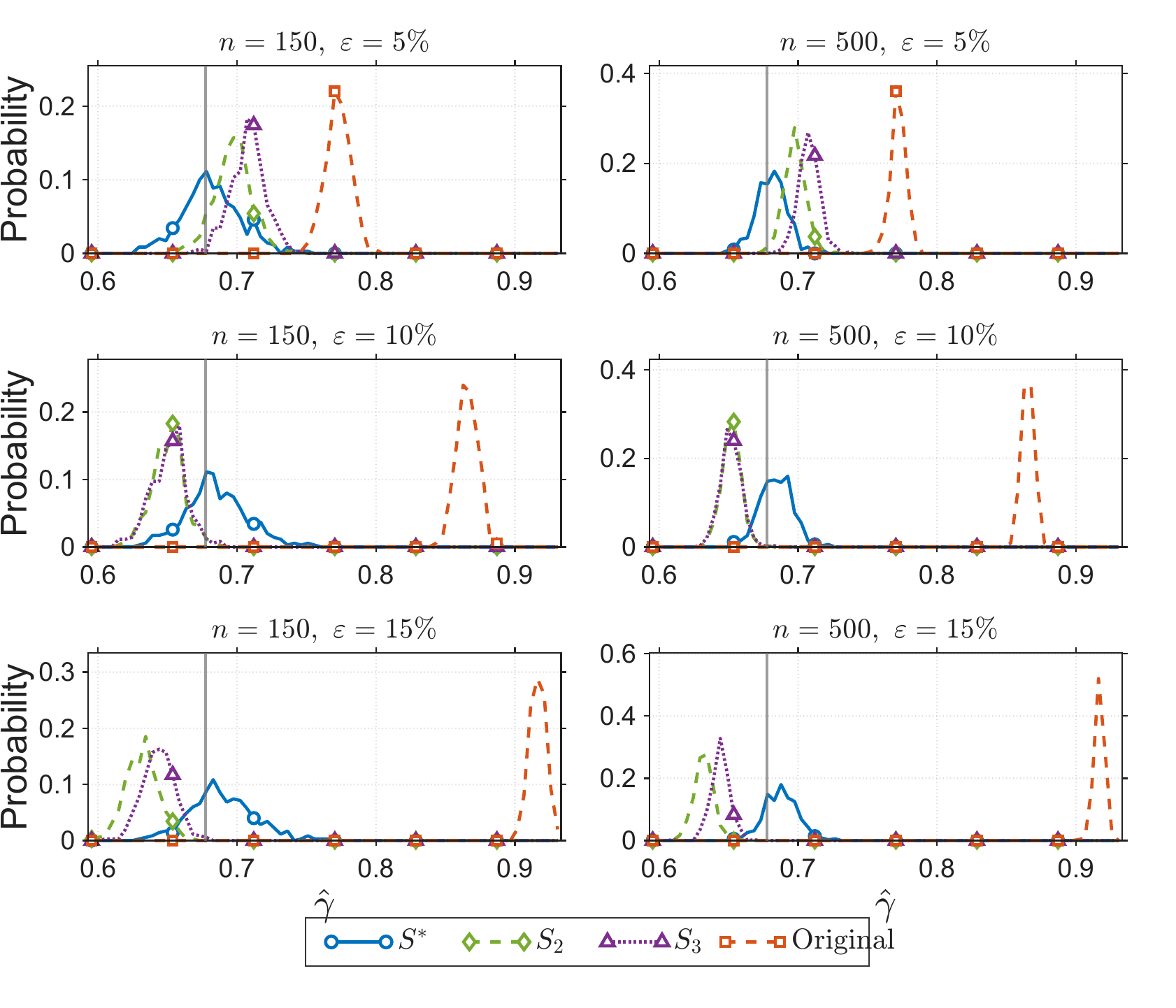}
    \caption{Finite-sample distribution of $\hat{\gamma}$ across selected partitions, sample sizes, and miscalibration sizes.}
    \label{fig:mc_histograms_gamma}
\end{figure}

Figure~\ref{fig: worsecasemiscalibration vs size} reports the worst-case information effect, $\gamma$, as a function of the miscalibration size $\epsilon$ for the main partitions. The figure summarizes how the target object responds under increasing local calibration error, highlighting substantial differences in robustness across partitions: the selected partition remains comparatively stable, while the original implementation exhibits a pronounced deterioration as $\epsilon$ increases, with intermediate partitions lying between these extremes. Overall, the plot provides a direct visual counterpart to the $K_S$-based ranking, illustrating how sensitivity translates into divergence in worst-case outcomes.

\begin{figure}[!t]
    \centering
    \includegraphics[width=\textwidth]{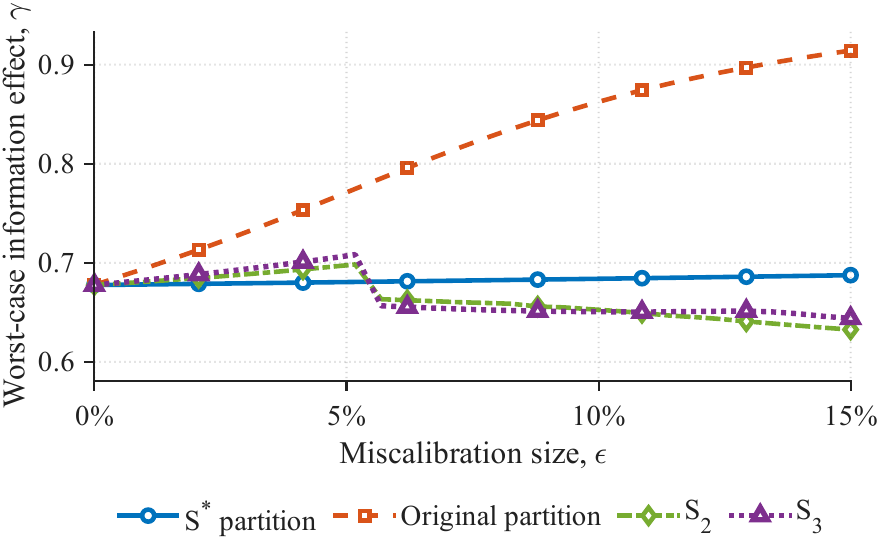}
    \caption{Information effect under worst-case miscalibration over different sizes.}
    \label{fig: worsecasemiscalibration vs size}
\end{figure}

The Monte Carlo results above use the baseline admissible set generated by $\tau_n=(\log n/n)^{1/2}$. Appendix \ref{appendix: weakid_tolerance} repeats the ranking and the worst-case Monte Carlo comparison under alternative rank thresholds. The results are substantively stable: $S^*$ remains the top-ranked partition according to $K_S$ for every cutoff, and the finite-sample MSE patterns continue to support the main conclusion that choosing the least-sensitive partition sharply reduces calibration risk.

\section{Conclusion} \label{sec: conclusion}

This paper studies a common but under-formalized choice in structural estimation: which parameters should be estimated and which should be calibrated. The main contribution is to treat this choice as a partition-selection problem. For each admissible calibration--estimation partition, the paper constructs a scalar sensitivity statistic that measures the worst-case local effect of calibration errors on the object of interest. The selected partition is the one that minimizes this statistic, and therefore minimizes worst-case local bias induced by plausible miscalibration of the fixed parameters.

The strength of the approach is that it turns sensitivity analysis into a decision rule. Rather than reporting that some calibrated parameters matter, the method provides a criterion for deciding which parameters should be fixed and which should be disciplined by the data. The rule is easy to implement, requires only local derivative information, avoids repeated re-estimation across partitions, and applies to a broad class of structural models. The applications show that this choice can matter substantially: partitions that look similar from the perspective of identification or computational convenience can imply very different robustness properties for the same target object.

The analysis also clarifies the limits of the method. The criterion is local: it evaluates robustness around a reference parameter vector and is informative when calibration errors are plausibly small or moderate. If the reference point is far from the relevant population value, or if large calibration errors are empirically plausible, global sensitivity analysis is more appropriate. A second limitation is that the normalization matrix is part of the economic content of the exercise. It defines what counts as a plausible and comparable miscalibration across parameters. Different normalizations can therefore lead to different selected partitions.

These limitations are not defects of the approach, but conditions for its interpretation. Structural work with calibrated parameters already relies on the assumption that the fixed values are credible approximations. The proposed statistic makes this assumption explicit and forces the researcher to state the scale on which calibration errors are judged. Robustness checks over alternative normalizations, together with the contribution decomposition of the sensitivity statistic, provide a transparent way to assess whether the selected partition depends on fragile choices about plausible parameter ranges.

Overall, the paper's message is that calibration is not merely a computational convenience. It is an econometric choice with direct consequences for the credibility of structural conclusions. By choosing the calibration--estimation partition through the lens of worst-case local bias, researchers can make this choice more systematic, transparent, and aligned with the object of interest.

\begin{appendix}

\section{Further Discussions}
\subsection{Framework and notation for the appendix}
\label{appendix:framework-notation}
This appendix uses the notation of the main text. Unless otherwise stated, all local objects
are evaluated at the reference point \(\eta^*\). For a partition \(S\), write
\(\eta=(\alpha_S,\beta_S)\), where \(\alpha_S\) collects the estimated coordinates and
\(\beta_S\) collects the fixed coordinates. When no confusion can arise, the subscript \(S\) is
suppressed and we write \((\alpha,\beta)\). The Jacobian blocks are
\[
G_{\alpha,S}:=\frac{\partial g(\eta^*;S)}{\partial \alpha'},
\qquad
G_{\beta,S}:=\frac{\partial g(\eta^*;S)}{\partial \beta'},
\]
and
\[
\Gamma_{\alpha,S}:=\frac{\partial \Gamma(\eta^*)}{\partial \alpha'},
\qquad
\Gamma_{\beta,S}:=\frac{\partial \Gamma(\eta^*)}{\partial \beta'}.
\]
For any partition satisfying the rank condition in the main text, define
\begin{equation}
D_{\alpha\beta,S}
:=-\left(G_{\alpha,S}'WG_{\alpha,S}\right)^{-1}G_{\alpha,S}'WG_{\beta,S},
\qquad
D_{\gamma\beta,S}:=\Gamma_{\alpha,S}D_{\alpha\beta,S}+\Gamma_{\beta,S}.
\label{eq:appendix-D-definitions}
\end{equation}

We omit the $S$ subscript when the context is clear.

\subsection{Local derivative and first-order expansion}
\label{appendix:local-differentiability}
\label{apppendix: well defined sensitivity}

The next lemma records the only differentiability fact needed repeatedly in the appendix and is used in the proof of Lemma~\ref{lemma: WorstBias}.

\begin{lemma}[Local derivative and uniform first-order expansion]
\label{lemma: derivatives}
\label{lemma: unif convergence first-taylor expansion}
Fix \(S\in\mathcal S_{\mathrm{AD}}(\eta^*)\). Under
Assumptions~\ref{ass: regularity conditions} and
\ref{ass: regularity assumptions reference point}, the following statements hold.
\begin{enumerate}[label=(\roman*)]
\item \(\alpha^\ast=\alpha(\beta^\ast;S)\).
\item \(\partial\alpha(\beta^\ast;S)/\partial\beta'=D_{\alpha\beta,S}\) and
\(\partial\gamma(\beta^\ast;S)/\partial\beta'=D_{\gamma\beta,S}\), where the matrices are
defined in \eqref{eq:appendix-D-definitions}.
\item If
\[
R_S(\beta)
:=
\gamma(\beta;S)-\gamma(\beta^\ast;S)-D_{\gamma\beta,S}(\beta-\beta^\ast),
\]
then
\[
\sup_{0<\|\beta-\beta^\ast\|_2\le\delta}
\frac{\|R_S(\beta)\|_2}{\|\beta-\beta^\ast\|_2}
\longrightarrow0
\qquad\text{as }\delta\downarrow0.
\]
In particular, for
\[
B_\epsilon(\beta^\ast;S)
:=
\left\{\beta_0\in B:
\left\|\Sigma_S^{-1}(\beta^\ast-\beta_0)\right\|_2
\le\epsilon\sqrt{|S^c|}\right\},
\]
one has
\[
\sup_{\beta_0\in B_\epsilon(\beta^\ast;S)}\|R_S(\beta_0)\|_2=o(\epsilon).
\]
\end{enumerate}
\end{lemma}

\begin{proof}
See Appendix \ref{appendix:main-proofs}.
\end{proof}

\subsection{Using the least-sensitive partition}
\label{appendix:stable-selected-partition}
\label{sec: identification S^*}

Let \(S^*:=S^*(\eta^*)\). After selecting \(S^*\), a researcher may use a nearby
alternative calibration \(\widetilde\beta\), re-estimate the free parameters under the same
partition, and form
\[
\widetilde\eta:=\bigl(\alpha(\widetilde\beta;S^*),\widetilde\beta\bigr).
\]
The relevant question is whether the least-sensitive partition remains the same when the
criterion is re-evaluated at \(\widetilde\eta\).

\begin{assumption}[Local stability of the admissible set and sensitivity ranking]
\label{ass:stable-selection}
\label{ass: local identification of S_T}
There exists a neighborhood \(U_0\) of \(\eta^*\) such that
\(\mathcal S_{\mathrm{AD}}(\eta)=\mathcal S_{\mathrm{AD}}(\eta^*)\) for all
\(\eta\in U_0\). Moreover, for every
\(S\in\mathcal S_{\mathrm{AD}}(\eta^*)\), the map
\(\eta\mapsto K_S(\eta)\) is continuous in a neighborhood of \(\eta^*\).
\end{assumption}

\begin{theorem}[Local stability of the selected partition]
\label{thm: local identification}
Let Assumptions~\ref{ass: uniq} and \ref{ass:stable-selection} hold. Then there exists a
neighborhood \(U\) of \(\eta^*\) such that, for all \(\eta\in U\),
\[
S^*(\eta)=S^*(\eta^*).
\]
\end{theorem}

\begin{proof}
See Appendix \ref{appendix:main-proofs}.
\end{proof}

In practice, the maintained requirement is that
\(\widetilde\eta=(\alpha(\widetilde\beta;S^*),\widetilde\beta)\) lies in the neighborhood
\(U\). This can be checked by recomputing the least-sensitive partition at \(\widetilde\eta\).

\subsection{Weak identification and the admissible set of partitions}\label{appendix: weak identification and Sad}

This subsection gives the local-sequence argument underlying the variance-decay restriction
in the definition of \(\mathcal S_{\mathrm{AD}}(\eta^*)\). Let
\[
g_n:\mathcal A_S\times\mathcal B_S\to\mathbb R^{d_g}
\]
denote the population system along the \(n\)th element of a local population sequence, and let
\[
\eta_n^\ast:=(\alpha_n^\ast,\beta_n^\ast)
\]
satisfy
\[
g_n(\alpha_n^\ast,\beta_n^\ast)=0.
\]
The fixed-population case in the main text is obtained by setting
\(g_n\equiv g\), \(\alpha_n^\ast\equiv\alpha^\ast\), and
\(\beta_n^\ast\equiv\beta^\ast\).

Throughout this subsection, the fixed component is held at its population reference value
\(\beta_n^\ast\). Thus, the stochastic object studied here is the conditional minimum-distance
estimator at fixed \(\beta_n^\ast\), which is the object entering the variance-decay
admissibility condition. The use of a sample reference point in plug-in derivative computations
is handled separately.

Define
\[
\widehat\alpha_n
\equiv
\widehat\alpha_n(\beta_n^\ast)
\in
\arg\min_{\alpha\in\mathcal A_S}
\widehat g_n(\alpha,\beta_n^\ast)'
\widehat W
\widehat g_n(\alpha,\beta_n^\ast).
\]
Equivalently, the sample point studied in this subsection is
\[
\widehat\eta_n^\ast
:=
(\widehat\alpha_n(\beta_n^\ast),\beta_n^\ast).
\]

For \(\alpha\in\mathcal A_S\), define the sample equation error at the fixed population
component by
\[
e_n(\alpha,\beta_n^\ast)
:=
\widehat g_n(\alpha,\beta_n^\ast)
-
g_n(\alpha,\beta_n^\ast).
\]
Let
\[
G_n
:=
\frac{\partial g_n(\alpha_n^\ast,\beta_n^\ast)}{\partial\alpha'},
\qquad
C_n
:=
\left(W_n^{1/2}G_n\right)^+.
\]

\begin{assumption}[Local-sequence regularity at fixed component]
\label{ass:local-sequence-md}
There exist finite constants \(K<\infty\), \(\bar r>0\), \(L_\Gamma<\infty\),
\(\rho\in[0,1)\), and a sequence \(q_n\to0\) such that the following conditions hold for all
large \(n\).

\begin{enumerate}
\item[(i)] \(W_n^{1/2}G_n\) has full column rank,
\[
C_nW_n^{1/2}G_n=I,
\]
and
\[
\|W_n^{1/2}\|\leq K.
\]

\item[(ii)] If \(\widehat W^{1/2}\) is the sample weighting matrix, then, with probability one,
\[
K^{-1}\|W_n^{1/2}v\|
\leq
\|\widehat W^{1/2}v\|
\leq
K\|W_n^{1/2}v\|
\]
for every conformable vector \(v\).

\item[(iii)] The estimator is localized. For some \(r_n\in(0,\bar r]\),
\[
P\{\widehat\alpha_n(\beta_n^\ast)\in\mathcal N_n(r_n)\}=1,
\qquad
\mathcal N_n(r):=\{\alpha:\|\alpha-\alpha_n^\ast\|\leq r\}.
\]
The map
\[
\alpha\mapsto g_n(\alpha,\beta_n^\ast)
\]
is continuously differentiable on \(\mathcal N_n(\bar r)\).

\item[(iv)] The sample equation satisfies
\[
\mathbb{E}\left[
\|e_n(\alpha_n^\ast,\beta_n^\ast)\|^2
\right]
=
O(q_n^2),
\]
and
\[
\mathbb{E}\left[
\sup_{\alpha\in\mathcal N_n(\bar r)}
\|e_n(\alpha,\beta_n^\ast)
-
e_n(\alpha_n^\ast,\beta_n^\ast)\|^2
\right]
=
O(q_n^2).
\]

\item[(v)] The Jacobian satisfies
\[
\sup_{\alpha\in\mathcal N_n(r_n)}
\left\|
C_nW_n^{1/2}
\left[
\frac{\partial g_n(\alpha,\beta_n^\ast)}{\partial\alpha'}
-
G_n
\right]
\right\|
\leq
\rho.
\]

\item[(vi)] The target is uniformly locally Lipschitz on a neighborhood of the population
reference point. For some \(\bar s>0\), define
\[
\mathcal B_n(s):=\{\beta:\|\beta-\beta_n^\ast\|\leq s\}.
\]
The map
\[
(\alpha,\beta)\mapsto \Gamma(\alpha,\beta)
\]
is continuously differentiable on
\[
\mathcal N_n(\bar r)\times \mathcal B_n(\bar s),
\]
and
\[
\sup_{(\alpha,\beta)\in \mathcal N_n(\bar r)\times \mathcal B_n(\bar s)}
\left\|
\begin{bmatrix}
\dfrac{\partial\Gamma(\alpha,\beta)}{\partial\alpha'}
&
\dfrac{\partial\Gamma(\alpha,\beta)}{\partial\beta'}
\end{bmatrix}
\right\|
\leq
L_\Gamma .
\]
\end{enumerate}
\end{assumption}

\begin{proposition}[Second-moment bound for the estimated parameters]
\label{prop:l2_alpha_rate_local_sequence}
Under Assumption~\ref{ass:local-sequence-md},
\[
\mathbb{E}\left[
\|\widehat\alpha_n(\beta_n^\ast)-\alpha_n^\ast\|^2
\right]
=
O\left(\|C_n\|^2q_n^2\right).
\]
If \(\|C_n\|q_n=o(1)\), then
\[
\mathbb{E}\left[
\|\widehat\alpha_n(\beta_n^\ast)-\alpha_n^\ast\|^2
\right]
=
o(1).
\]
\end{proposition}

\begin{proposition}[Vanishing second moments for the object of interest]
\label{prop:gamma_l2_rate_local_sequence}
Under Assumption~\ref{ass:local-sequence-md},
define
\[
\widehat\gamma_n
:=
\Gamma\!\left(
\widehat\alpha_n(\beta_n^\ast),
\beta_n^\ast
\right),
\qquad
\gamma_n^\ast
:=
\Gamma(\alpha_n^\ast,\beta_n^\ast).
\]
Then
\[
\mathbb{E}\left[
\|\widehat\gamma_n-\gamma_n^\ast\|^2
\right]
=
O\left(\|C_n\|^2q_n^2\right),
\qquad
\operatorname{tr}\{\operatorname{Var}(\widehat\gamma_n)\}
=
O\left(\|C_n\|^2q_n^2\right).
\]
If \(\|C_n\|q_n=o(1)\), then
\[
\mathbb{E}\left[
\|\widehat\gamma_n-\gamma_n^\ast\|^2
\right]
=
o(1),
\qquad
\operatorname{tr}\{\operatorname{Var}(\widehat\gamma_n)\}
=
o(1).
\]
\end{proposition}

\subsection{Consistency of the estimated least-sensitive partition}
\label{appendix:partition-consistency}
\label{sec: consistency}

\begin{assumption}[High-level consistency conditions]
\label{ass:partition-consistency-highlevel}
\label{ass: consistency of eta}
\label{ass: rank estimation}
\label{ass: u.c of nabla}
The following conditions hold.
\begin{enumerate}[label=(\roman*)]
\item Given the first-stage calibration \((\beta_1,S_1)\),
\(\widehat\eta\overset{p}{\to}\eta^*\).
\item The estimated admissible set is consistent:
\[
\Pr\left\{\widehat{\mathcal S}_{\mathrm{AD}}(\widehat\eta)
=\mathcal S_{\mathrm{AD}}(\eta^*)\right\}\to1.
\]
\item For every candidate partition \(S\),
\[
\sup_{\eta\in N}
\left\|
\frac{\partial \widehat g(\eta;S)}{\partial\eta'}
-
\frac{\partial g(\eta;S)}{\partial\eta'}
\right\|
\overset{p}{\to}0.
\]
\item \(\widehat W\overset{p}{\to}W\), where \(W\) is symmetric positive definite.
\end{enumerate}
\end{assumption}

Assumption~\ref{ass:partition-consistency-highlevel}(iii) is a higher-order uniform convergence
condition for the Jacobians used in both the sensitivity statistic and the admissible-set
computation. For moment functions \(g(\eta;S)=E[m(X_i,\eta;S)]\), primitive sufficient
conditions follow from Lemma 2.4 of \cite{newey1994large}: compactness of \(N\), stationarity
and ergodicity of \(\{X_i\}\), continuity in \(\eta\) of the derivative of \(m\), and an integrable
dominating function for that derivative.

\begin{theorem}[Consistency of the selected partition]
\label{thm: Consistency of partition}
Suppose Assumptions~\ref{ass: regularity conditions},
\ref{ass: regularity assumptions reference point}, \ref{ass: uniq}, and
\ref{ass:partition-consistency-highlevel} hold. Then
\[
\Pr\left\{\widehat S^*(\widehat\eta)=S^*(\eta^*)\right\}\to1.
\]
\end{theorem}

\begin{proof}
See Appendix \ref{appendix:main-proofs}.
\end{proof}

\subsection{Direct and re-estimated derivatives}
\label{appendix:direct-derivative}
\label{app:direct-derivative}

This subsection formalizes why Algorithm in Section \ref{sec: Implementation} can compute the
local derivative at the common reference point rather than re-estimating the model separately
for every candidate partition. Fix a partition \(S\). For a generic fixed value \(\beta\), define
\[
\widehat\alpha(\beta;S)
\in
\arg\min_{\alpha\in A_S}
\widehat g(\alpha,\beta;S)'\widehat W\widehat g(\alpha,\beta;S).
\]
Let
\[
\widehat G_{\alpha}(\eta;S):=\frac{\partial\widehat g(\eta;S)}{\partial\alpha'},
\qquad
\widehat G_{\beta}(\eta;S):=\frac{\partial\widehat g(\eta;S)}{\partial\beta'}.
\]
For the \(S\)-specific ordering of the common reference point, write
\(\widehat\eta=(\widehat\alpha_{S},\widehat\beta_{S})\). The direct derivative used by
the algorithm is
\begin{equation}
\widehat D_{\alpha\beta}(S)
:=-
\left[\widehat G_{\alpha}(\widehat\eta;S)'
\widehat W
\widehat G_{\alpha}(\widehat\eta;S)
\right]^{-1}
\widehat G_{\alpha}(\widehat\eta;S)'
\widehat W
\widehat G_{\beta}(\widehat\eta;S).
\label{eq:direct-derivative-estimator}
\end{equation}

\begin{assumption}[Local regularity for direct derivative computation]
\label{ass:direct-derivative}
For the fixed partition \(S\), let \(\eta^*=(\alpha^\ast,\beta^\ast)\) denote the population
reference point under the \(S\)-specific ordering. The following conditions hold.
\begin{enumerate}[label=(\roman*)]
\item \(\eta^*\in\operatorname{int}(N)\), \(g(\eta^*;S)=0\), and
\(G_{\alpha,S}\) has full column rank.
\item \(\widehat\eta\overset{p}{\to}\eta^*\) and
\(\widehat\beta_{S}\overset{p}{\to}\beta^\ast\). 
\item \(\widehat W\overset{p}{\to}W\), where \(W\) is symmetric positive definite.
\item The local sample solution \(\widehat\alpha(\widehat\beta_{S};S)\) exists with
probability approaching one, satisfies the sample first-order conditions, and
\[
(\widehat\alpha(\widehat\beta_{S};S),\widehat\beta_{S})\overset{p}{\to}(\alpha^\ast,\beta^\ast).
\]
\item In a neighborhood of \(\eta^*\), \(\widehat g\), its first derivatives, and its second
derivatives converge uniformly in probability to their population counterparts.
\end{enumerate}
\end{assumption}

\begin{theorem}[Equivalence of direct and re-estimated derivatives]
\label{thm:direct-reestimated-derivative}
Suppose Assumption~\ref{ass:direct-derivative} holds. Then, with probability approaching one,
the derivative of the sample re-estimation map exists at \(\widehat\beta_{S}\), and
\[
\left.
\frac{\partial\widehat\alpha(\beta;S)}{\partial\beta'}
\right|_{\beta=\widehat\beta_{S}}
=
\widehat D_{\alpha\beta}(S)+o_p(1)
=
D_{\alpha\beta,S}+o_p(1),
\]
where \(D_{\alpha\beta,S}\) is defined in \eqref{eq:appendix-D-definitions}.
\end{theorem}

\begin{proof}
See Appendix \ref{appendix:main-proofs}.
\end{proof}

\begin{corollary}[Derivative of the object of interest]
\label{cor:direct-target-derivative}
Suppose the conditions of Theorem~\ref{thm:direct-reestimated-derivative} hold and
\(\Gamma\) is continuously differentiable in a neighborhood of \(\eta^*\). Define
\[
\widehat D_{\gamma\beta}(S)
:=
\Gamma_{\alpha}(\widehat\eta;S)\widehat D_{\alpha\beta}(S)
+
\Gamma_{\beta}(\widehat\eta;S).
\]
Then
\[
\left.
\frac{\partial\widehat\gamma(\beta;S)}{\partial\beta'}
\right|_{\beta=\widehat\beta_{S}}
=
\widehat D_{\gamma\beta}(S)+o_p(1)
=
D_{\gamma\beta,S}+o_p(1).
\]
\end{corollary}

\begin{proof}
See Appendix \ref{appendix:main-proofs}.
\end{proof}

\begin{remark}[Uniformity over partitions]
\label{rem:uniform-direct-derivative}
Let \(\mathcal S_0\) be a finite collection of candidate partitions. If
Assumption~\ref{ass:direct-derivative} holds uniformly over \(S\in\mathcal S_0\) and
\[
\inf_{S\in\mathcal S_0}
\lambda_{\min}\left(G_{\alpha,S}'WG_{\alpha,S}\right)>0,
\]
then
\[
\max_{S\in\mathcal S_0}
\left\|
\left.
\frac{\partial\widehat\alpha(\beta;S)}{\partial\beta'}
\right|_{\beta=\widehat\beta_{S}}
-
\widehat D_{\alpha\beta}(S)
\right\|=o_p(1),
\]
and the analogous statement holds for \(\widehat D_{\gamma\beta}(S)\). Hence the direct
algorithmic derivative is uniformly asymptotically equivalent to the derivative obtained after
partition-specific re-estimation.
\end{remark}

\section{Extension to MLE}

The examples below clarify when the population system \(g(\eta;S)\) depends on the partition.

\begin{example}[Maximum likelihood]
\label{ex:mle-partition-dependent}
In maximum likelihood estimation, the population moments are score equations. For a
partition \(S\),
\begin{equation}
 g_P(\alpha,\beta;S)
 :=E_P\left[\frac{\partial}{\partial\alpha}\log f(X_i\mid \alpha,\beta)\right].
\label{eq:mle-score-moments}
\end{equation}
Here \(S\) selects the score components included in the system. Hence both the functional
form and the dimension \(d_g(S)\) of \(g_P(\cdot;S)\) may depend on \(S\).
\end{example}

\begin{example}[Classical minimum distance]
\label{ex:cmd-partition-independent}
In the classical minimum-distance framework,
\begin{equation}
 g_P(\alpha,\beta;S)=\theta(P)-\psi(\alpha,\beta).
\label{eq:cmd-moments}
\end{equation}
The same population restrictions are used for all partitions, up to the relabeling of coordinates
into estimated and fixed blocks.
\end{example}

In CMD, and in GMM settings where the moment conditions do not depend on the partition,
a primitive condition for Assumption~\ref{ass: regularity assumptions reference point} to hold is that
the first-stage calibration is compatible with the population equations, i.e., there exists
\(\alpha\) such that \(g(\alpha,\beta_1;S_1)=0\).

In MLE, by contrast, \(g_P(\eta;S)\) typically depends on \(S\) because \(S\) selects a subset
of score equations. A convenient sufficient condition for
Assumption~\ref{ass: regularity assumptions reference point} is that \(\eta^*\) satisfies the
mean-zero restriction for the full score vector, not only for the score components associated
with the first-stage estimated parameters. Appendix~\ref{appendix:mle-reference-point}
constructs such a reference point.

\subsection{Reference point for the MLE case}
\label{appendix:mle-reference-point}
\label{appendix: MLE reference point}

Suppose the researcher first estimates the model by MLE under a calibration strategy
\((S_1,\beta_1)\) and obtains
\[
\widetilde\eta_n:=\bigl(\widehat\alpha_n(\beta_1;S_1),\beta_1\bigr).
\]
As discussed above, the MLE moment map
\(g_P(\eta;S)\) depends on \(S\). Consequently, \(\widetilde\eta_n\) need not satisfy the
mean-zero restrictions for score components corresponding to parameters that were fixed in the
first stage. This may violate Assumption~\ref{ass: regularity assumptions reference point}
when the reference point must be observationally equivalent across candidate partitions.

Let
\[
S_{\mathrm{full}}:=\{1,\ldots,d_\eta\},
\qquad
s(X_i,\eta):=\frac{\partial}{\partial\eta}\log f(X_i\mid \eta),
\]
and define the full-score sample and population moments by
\[
\widehat g_n(\eta;S_{\mathrm{full}}):=\frac{1}{n}\sum_{i=1}^n s(X_i,\eta),
\qquad
g(\eta;S_{\mathrm{full}}):=E_P[s(X_i,\eta)].
\]
Let
\[
Q_n(\eta):=\widehat g_n(\eta;S_{\mathrm{full}})'W_n\widehat g_n(\eta;S_{\mathrm{full}}),
\qquad
Q(\eta):=g(\eta;S_{\mathrm{full}})'Wg(\eta;S_{\mathrm{full}}),
\]
where \(W_n\overset{p}{\to}W\) and \(W\) is symmetric positive definite. To avoid confusion
with the weighting matrix, let \(H\) denote a fixed symmetric positive-definite penalty metric
and write \(\|x\|_H^2:=x'Hx\). Define
\begin{equation}
\widehat\eta_n
\in
\arg\min_{\eta\in N}
\left\{Q_n(\eta)+\lambda_n\|\eta-\widetilde\eta_n\|_H^2\right\}.
\label{eq:penalized_full_score}
\end{equation}
The special case \(H=I_{d_\eta}\) gives the Euclidean penalty.

Define the full-score identified set
\[
N_0(S_{\mathrm{full}}):=\{\eta\in N:g(\eta;S_{\mathrm{full}})=0\}.
\]
Since every partition-specific score system is a subvector of the full score, any
\(\eta\in N_0(S_{\mathrm{full}})\) also belongs to \(N_0(S)\) for every partition \(S\).
If \(\widetilde\eta_n\overset{p}{\to}\widetilde\eta_0\), define
\begin{equation}
\eta^*
:=
\arg\min_{\eta\in N_0(S_{\mathrm{full}})}\|\eta-\widetilde\eta_0\|_H^2.
\label{eq:projection_eta_star}
\end{equation}
Thus \(\eta^*\) is the element of the full-score identified set closest to the probability
limit of the first-stage estimate.

\begin{assumption}[Penalized full-score reference point]
\label{ass:mle-reference}
The following conditions hold.
\begin{enumerate}[label=(\roman*)]
\item \(N\subset\mathbb R^{d_\eta}\) is compact and \(H\) is symmetric positive definite.
\item \(\widetilde\eta_n\overset{p}{\to}\widetilde\eta_0\).
\item \(\sup_{\eta\in N}|Q_n(\eta)-Q(\eta)|\overset{p}{\to}0\).
\item \(N_0(S_{\mathrm{full}})\) is nonempty and separated by \(Q\): for every
\(\varepsilon>0\),
\[
\inf_{\eta\in N:d_H(\eta,N_0(S_{\mathrm{full}}))\ge\varepsilon}Q(\eta)>0,
\qquad
d_H(\eta,N_0):=\inf_{\zeta\in N_0}\|\eta-\zeta\|_H.
\]
\item The projection in \eqref{eq:projection_eta_star} exists and is unique.
\item Along the full-score identified set,
\[
\sup_{\eta\in N_0(S_{\mathrm{full}})}Q_n(\eta)=O_p(n^{-1}).
\]
\item The tuning parameter satisfies \(\lambda_n\to0\) and \(n\lambda_n\to\infty\).
\end{enumerate}
\end{assumption}

Assumption~\ref{ass:mle-reference} collects the conditions needed for the penalized
full-score estimator. Parts (i)--(iii) are standard extremum-estimation conditions. Parts
(iv)--(vi) handle the possibility that the full-score problem is set-identified. Part (vii)
ensures that the penalty vanishes asymptotically but dominates the \(O_p(n^{-1})\) stochastic
variation of \(Q_n\) along flat directions. A convenient deterministic choice is
\(\lambda_n=n^{-\kappa}\) with \(0<\kappa<1\). If one also wants the penalty to be
negligible at the \(\sqrt n\) scale in well-identified directions, one may choose
\(\kappa\in(1/2,1)\), for instance \(\lambda_n=n^{-2/3}\).

\begin{theorem}[Consistency of the penalized full-score reference point]
\label{thm:projection_consistency_penalized_score}
Under Assumption~\ref{ass:mle-reference},
\[
\widehat\eta_n\overset{p}{\to}\eta^*.
\]
Consequently, if \(\eta^*\in\operatorname{int}(N)\), the estimator in
\eqref{eq:penalized_full_score} delivers a reference point satisfying
Assumption~\ref{ass: regularity assumptions reference point} for the MLE partition-selection
problem.
\end{theorem}

\begin{proof}
See Appendix \ref{appendix:main-proofs}.
\end{proof}

\section{Main Proofs}
\label{appendix:main-proofs}

\begin{proof}[Proof of Lemma~\ref{lemma: derivatives}]
Fix a partition \(S\in\mathcal S_{\mathrm{AD}}(\eta^*)\). Under the \(S\)-specific ordering, write
\(\eta^*=(\alpha^\ast,\beta^\ast)\). By
Assumption~\ref{ass: regularity assumptions reference point}, the reference point satisfies
\(g(\alpha^\ast,\beta^\ast;S)=0\). The minimum-distance criterion at the fixed value
\(\beta^\ast\) is
\[
Q(\alpha,\beta^\ast;S):=g(\alpha,\beta^\ast;S)'Wg(\alpha,\beta^\ast;S).
\]
Since \(W\) is positive definite, \(Q(\alpha,\beta^\ast;S)\geq0\) for every
\(\alpha\). At \(\alpha=\alpha^\ast\), the criterion equals zero. Hence
\(\alpha^\ast\) solves the conditional MD problem at \(\beta^\ast\), that is,
\[
\alpha^\ast=\alpha(\beta^\ast;S).
\]
This proves part (i).

To obtain the derivative, define the population first-order condition map
\[
F(\alpha,\beta;S)
:=
\frac{\partial Q(\alpha,\beta;S)}{\partial\alpha}
=
2G_\alpha(\alpha,\beta;S)'Wg(\alpha,\beta;S)
\in\mathbb R^{|S|},
\]
where
\[
G_\alpha(\alpha,\beta;S):=
\frac{\partial g(\alpha,\beta;S)}{\partial\alpha'}.
\]
By Assumption~\ref{ass: regularity conditions}, \(g\) is twice continuously differentiable,
so \(F\) is continuously differentiable in a neighborhood of
\((\alpha^\ast,\beta^\ast)\). Since \(g(\alpha^\ast,\beta^\ast;S)=0\),
\[
F(\alpha^\ast,\beta^\ast;S)=0.
\]
The derivative of \(F\) with respect to \(\alpha\), evaluated at the reference point, is
\[
\frac{\partial F(\alpha^\ast,\beta^\ast;S)}{\partial\alpha'}
=
2G_{\alpha,S}'WG_{\alpha,S},
\]
because all terms involving second derivatives of \(g\) are multiplied by
\(g(\alpha^\ast,\beta^\ast;S)=0\). Similarly,
\[
\frac{\partial F(\alpha^\ast,\beta^\ast;S)}{\partial\beta'}
=
2G_{\alpha,S}'WG_{\beta,S}.
\]
The rank condition in \(S\in\mathcal S_{\mathrm{AD}}(\eta^*)\) gives
\(\operatorname{rank}(G_{\alpha,S})=|S|\). Since \(W\) is positive definite,
\(G_{\alpha,S}'WG_{\alpha,S}\) is nonsingular. Therefore the implicit function theorem
applied to \(F(\alpha,\beta;S)=0\) yields a neighborhood of \(\beta^\ast\) and a unique
differentiable map \(\beta\mapsto\alpha(\beta;S)\) satisfying
\[
F(\alpha(\beta;S),\beta;S)=0,
\qquad
\alpha(\beta^\ast;S)=\alpha^\ast.
\]
Differentiating the identity \(F(\alpha(\beta;S),\beta;S)=0\) at
\(\beta=\beta^\ast\) gives
\[
0
=
\frac{\partial F}{\partial\alpha'}
\frac{\partial\alpha(\beta^\ast;S)}{\partial\beta'}
+
\frac{\partial F}{\partial\beta'}.
\]
Substituting the two derivative blocks above and cancelling the common factor 2,
\[
\frac{\partial\alpha(\beta^\ast;S)}{\partial\beta'}
=
-
\left(G_{\alpha,S}'WG_{\alpha,S}\right)^{-1}
G_{\alpha,S}'WG_{\beta,S}
=
D_{\alpha\beta,S}.
\]
Since \(\Gamma\) is continuously differentiable by Assumption~\ref{ass: regularity conditions},
the chain rule applied to
\[\gamma(\beta;S)=\Gamma(\alpha(\beta;S),\beta)\]
gives
\[
\frac{\partial\gamma(\beta^\ast;S)}{\partial\beta'}
=
\Gamma_{\alpha,S}D_{\alpha\beta,S}+\Gamma_{\beta,S}
=
D_{\gamma\beta,S}.
\]
This proves part (ii).

It remains to prove the first-order expansion. Since \(\gamma(\cdot;S)\) is differentiable at
\(\beta^\ast\), it is Fr\'echet differentiable in the finite-dimensional Euclidean space, and
therefore
\[
\frac{
\left\|
\gamma(\beta;S)-\gamma(\beta^\ast;S)-D_{\gamma\beta,S}(\beta-\beta^\ast)
\right\|_2
}{
\|\beta-\beta^\ast\|_2
}
\longrightarrow0
\qquad\text{as }\beta\to\beta^\ast, \; \beta\neq \beta^*.
\]
With the notation
\[
R_S(\beta)
:=
\gamma(\beta;S)-\gamma(\beta^\ast;S)-D_{\gamma\beta,S}(\beta-\beta^\ast),
\]
this is exactly the pointwise statement
\[
\frac{\|R_S(\beta)\|_2}{\|\beta-\beta^\ast\|_2}\to0
\qquad\text{as }\beta\to\beta^\ast.
\]
We now turn this pointwise statement into the displayed uniform statement. Suppose, toward a
contradiction, that
\[
\sup_{0<\|\beta-\beta^\ast\|_2\leq\delta}
\frac{\|R_S(\beta)\|_2}{\|\beta-\beta^\ast\|_2}
\not\longrightarrow0
\qquad\text{as }\delta\downarrow0.
\]
Then there exist \(c>0\), a sequence \(\delta_m\downarrow0\), and points
\(\beta_m\) with \(0<\|\beta_m-\beta^\ast\|_2\leq\delta_m\) such that
\[
\frac{\|R_S(\beta_m)\|_2}{\|\beta_m-\beta^\ast\|_2}\geq c.
\]
But \(\|\beta_m-\beta^\ast\|_2\leq\delta_m\to0\), so \(\beta_m\to\beta^\ast\),
contradicting the pointwise Fr\'echet differentiability limit. Hence
\[
\sup_{0<\|\beta-\beta^\ast\|_2\leq\delta}
\frac{\|R_S(\beta)\|_2}{\|\beta-\beta^\ast\|_2}
\longrightarrow0
\qquad\text{as }\delta\downarrow0.
\]

Finally, consider the normalized neighborhood
\[
B_\epsilon(\beta^\ast;S)
:=
\left\{\beta_0\in B:
\left\|\Sigma_S^{-1}(\beta^\ast-\beta_0)\right\|_2
\leq \epsilon\sqrt{|S^c|}
\right\}.
\]
For any \(\beta_0\in B_\epsilon(\beta^\ast;S)\),
\[
\|\beta_0-\beta^\ast\|_2
=
\left\|\Sigma_S\Sigma_S^{-1}(\beta_0-\beta^\ast)\right\|_2
\leq
\|\Sigma_S\|_2
\left\|\Sigma_S^{-1}(\beta_0-\beta^\ast)\right\|_2
\leq
\|\Sigma_S\|_2\epsilon\sqrt{|S^c|}.
\]
Thus \(B_\epsilon(\beta^\ast;S)\) is contained in a Euclidean ball around \(\beta^\ast\)
of radius \(C_S\epsilon\), where \(C_S:=\|\Sigma_S\|_2\sqrt{|S^c|}\). Applying the
uniform differentiability bound on this ball gives
\[
\sup_{\beta_0\in B_\epsilon(\beta^\ast;S)}\|R_S(\beta_0)\|_2
\leq
\left(
\sup_{0<\|\beta-\beta^\ast\|_2\leq C_S\epsilon}
\frac{\|R_S(\beta)\|_2}{\|\beta-\beta^\ast\|_2}
\right)C_S\epsilon
=o(\epsilon).
\]
This proves part (iii) and completes the proof of the lemma.
\end{proof}

\begin{proof}[Proof of Theorem~\ref{thm:projection_consistency_penalized_score}]
Let
\[
P_n(\eta):=\|\eta-\widetilde\eta_n\|_H^2,
\qquad
P(\eta):=\|\eta-\widetilde\eta_0\|_H^2.
\]
By definition of \(\widehat\eta_n\) as a minimizer of
\eqref{eq:penalized_full_score}, and because \(\eta^*\in N_0(S_{\mathrm{full}})\),
we have
\begin{equation}
Q_n(\widehat\eta_n)+\lambda_nP_n(\widehat\eta_n)
\leq
Q_n(\eta^*)+\lambda_nP_n(\eta^*).
\label{eq:basic_optimality_appendix_revised}
\end{equation}

By Assumption~\ref{ass:mle-reference}(vi), since \(\eta^*\in N_0(S_{\mathrm{full}})\),
\[
Q_n(\eta^*)=O_p(n^{-1}).
\]
Combining this with Assumption~\ref{ass:mle-reference}(vii), which implies
\(n\lambda_n\to\infty\), gives
\[
\frac{Q_n(\eta^*)}{\lambda_n}=o_p(1).
\]
Using \eqref{eq:basic_optimality_appendix_revised}, we therefore obtain
\[
Q_n(\widehat\eta_n)
\leq
Q_n(\eta^*)+\lambda_nP_n(\eta^*).
\]
The first term is \(O_p(n^{-1})=o_p(1)\). The second term is also \(o_p(1)\), because
\(\lambda_n\to0\) by Assumption~\ref{ass:mle-reference}(vii) and
\(P_n(\eta^*)=O_p(1)\) by compactness of \(N\) and
Assumption~\ref{ass:mle-reference}(ii). Hence
\[
Q_n(\widehat\eta_n)=o_p(1).
\]

Next, Assumption~\ref{ass:mle-reference}(iii) gives uniform convergence of \(Q_n\) to \(Q\),
so
\[
Q(\widehat\eta_n)
\leq
Q_n(\widehat\eta_n)+\sup_{\eta\in N}|Q_n(\eta)-Q(\eta)|
=
o_p(1).
\]
By the separation condition in Assumption~\ref{ass:mle-reference}(iv), this implies
\begin{equation}
d_H(\widehat\eta_n,N_0(S_{\mathrm{full}}))\overset{p}{\to}0.
\label{eq:close_to_identified_set_appendix_revised}
\end{equation}
Indeed, if the distance from \(\widehat\eta_n\) to \(N_0(S_{\mathrm{full}})\) were bounded
away from zero with non-vanishing probability, the population criterion \(Q\) would be
bounded away from zero on that event, contradicting \(Q(\widehat\eta_n)=o_p(1)\).

Return now to \eqref{eq:basic_optimality_appendix_revised} and divide both sides by
\(\lambda_n\):
\[
P_n(\widehat\eta_n)
\leq
P_n(\eta^*)+\frac{Q_n(\eta^*)}{\lambda_n}
=
P_n(\eta^*)+o_p(1).
\]
Since \(\widetilde\eta_n\overset{p}{\to}\widetilde\eta_0\) and \(N\) is compact,
\[
\sup_{\eta\in N}|P_n(\eta)-P(\eta)|\overset{p}{\to}0.
\]
To see this, note that
\[
|P_n(\eta)-P(\eta)|
=
\left|
\|\eta-\widetilde\eta_n\|_H^2-
\|\eta-\widetilde\eta_0\|_H^2
\right|
\]
is uniformly controlled over compact \(N\) by a constant times
\(\|\widetilde\eta_n-\widetilde\eta_0\|_H\). Consequently,
\begin{equation}
P(\widehat\eta_n)
\leq
P(\eta^*)+o_p(1).
\label{eq:penalty_comparison_appendix_revised}
\end{equation}

We now use a subsequence argument. Take any subsequence of \(\{\widehat\eta_n\}\).
By compactness of \(N\), it has a further subsequence, still denoted by
\(\{\widehat\eta_n\}\) for notational simplicity, that converges in distribution to some
limit \(\widehat\eta\). By \eqref{eq:close_to_identified_set_appendix_revised}, every such
limit satisfies
\[
\widehat\eta\in N_0(S_{\mathrm{full}})
\qquad\text{a.s.}
\]
Moreover, by continuity of \(P\) and \eqref{eq:penalty_comparison_appendix_revised},
\[
P(\widehat\eta)\leq P(\eta^*)
\qquad\text{a.s.}
\]
But Assumption~\ref{ass:mle-reference}(v) states that \(\eta^*\) is the unique minimizer
of \(P\) over \(N_0(S_{\mathrm{full}})\). Therefore
\[
\widehat\eta=\eta^*
\qquad\text{a.s.}
\]
Since every subsequence has a further subsequence converging to \(\eta^*\), it follows that
\[
\widehat\eta_n\overset{p}{\to}\eta^*.
\]
This proves the theorem.
\end{proof}

\begin{proof}[Proof of Lemma~\ref{lemma: WorstBias}]
Fix a partition \(S\in\mathcal S_{\mathrm{AD}}(\eta^*)\) if $|S^c| = 0$, there are no fixed parameters, the sensitivity statistic $K_S = 0$, and the worst-case bias is exactly zero, so the equality holds trivially. Assume hereafter that $|S^c| > 0$.
Write
\[
D:=\frac{\partial\gamma(\beta^\ast;S)}{\partial\beta'}=D_{\gamma\beta,S}.
\]
By Lemma~\ref{lemma: unif convergence first-taylor expansion}, for
\(\beta_0\in B_\epsilon(\beta^\ast;S)\) we have the expansion
\[
\gamma(\beta_0;S)-\gamma(\beta^\ast;S)
=
D(\beta_0-\beta^\ast)+R_S(\beta_0),
\qquad
\sup_{\beta_0\in B_\epsilon(\beta^\ast;S)}\|R_S(\beta_0)\|_2=o(\epsilon).
\]
Equivalently,
\[
\gamma(\beta^\ast;S)-\gamma(\beta_0;S)
=
D(\beta^\ast-\beta_0)-R_S(\beta_0).
\]
By the triangle inequality, for each \(\beta_0\in B_\epsilon(\beta^\ast;S)\),
\[
\|\gamma(\beta^\ast;S)-\gamma(\beta_0;S)\|_2
\leq
\|D(\beta^\ast-\beta_0)\|_2+
\|R_S(\beta_0)\|_2.
\]
Taking the supremum over \(B_\epsilon(\beta^\ast;S)\) gives
\begin{equation}
\sup_{\beta_0\in B_\epsilon(\beta^\ast;S)}
\|\gamma(\beta^\ast;S)-\gamma(\beta_0;S)\|_2
\leq
\sup_{\beta_0\in B_\epsilon(\beta^\ast;S)}
\|D(\beta^\ast-\beta_0)\|_2
+
\sup_{\beta_0\in B_\epsilon(\beta^\ast;S)}
\|R_S(\beta_0)\|_2.
\label{eq:WorstBias_first_triangle_revised}
\end{equation}
Conversely, using
\[
D(\beta^\ast-\beta_0)
=
\gamma(\beta^\ast;S)-\gamma(\beta_0;S)+R_S(\beta_0),
\]
the triangle inequality gives, for each \(\beta_0\),
\[
\|D(\beta^\ast-\beta_0)\|_2
\leq
\|\gamma(\beta^\ast;S)-\gamma(\beta_0;S)\|_2+
\|R_S(\beta_0)\|_2.
\]
Taking the supremum gives
\begin{equation}
\sup_{\beta_0\in B_\epsilon(\beta^\ast;S)}
\|D(\beta^\ast-\beta_0)\|_2
\leq
\sup_{\beta_0\in B_\epsilon(\beta^\ast;S)}
\|\gamma(\beta^\ast;S)-\gamma(\beta_0;S)\|_2
+
\sup_{\beta_0\in B_\epsilon(\beta^\ast;S)}
\|R_S(\beta_0)\|_2.
\label{eq:WorstBias_second_triangle_revised}
\end{equation}
Combining \eqref{eq:WorstBias_first_triangle_revised} and
\eqref{eq:WorstBias_second_triangle_revised},
\begin{equation}
\left|
\sup_{\beta_0\in B_\epsilon(\beta^\ast;S)}
\|\gamma(\beta^\ast;S)-\gamma(\beta_0;S)\|_2
-
\sup_{\beta_0\in B_\epsilon(\beta^\ast;S)}
\|D(\beta^\ast-\beta_0)\|_2
\right|
\leq
\sup_{\beta_0\in B_\epsilon(\beta^\ast;S)}
\|R_S(\beta_0)\|_2
=
o(\epsilon).
\label{eq:WorstBias_sup_difference_revised}
\end{equation}
Thus,
\[
\operatorname{WorstBias}(\beta^\ast,\epsilon;S)
:=
\sup_{\beta_0\in B_\epsilon(\beta^\ast;S)}
\|\gamma(\beta^\ast;S)-\gamma(\beta_0;S)\|_2
=
\sup_{\beta_0\in B_\epsilon(\beta^\ast;S)}
\|D(\beta^\ast-\beta_0)\|_2
+o(\epsilon).
\]

Now reparametrize the constraint. For \(\beta_0\in B_\epsilon(\beta^\ast;S)\), define
\[
x:=\frac{1}{\epsilon\sqrt{|S^c|}}\Sigma_S^{-1}(\beta^\ast-\beta_0).
\]
Then \(\|x\|_2\leq1\) and
\[
\beta^\ast-\beta_0=\epsilon\sqrt{|S^c|}\,\Sigma_Sx.
\]
Therefore
\[
\begin{aligned}
\sup_{\beta_0\in B_\epsilon(\beta^\ast;S)}
\|D(\beta^\ast-\beta_0)\|_2
&=
\epsilon\sqrt{|S^c|}
\sup_{\|x\|_2\leq1}\|D\Sigma_Sx\|_2 \\
&=
\epsilon\sqrt{|S^c|}\,\|D\Sigma_S\|_2,
\end{aligned}
\]
where the last equality is the definition of the induced Euclidean operator norm. Recalling
\[
K_S:=\sqrt{|S^c|}\left\|
\frac{\partial\gamma(\beta^\ast;S)}{\partial\beta'}\Sigma_S
\right\|_2
=
\sqrt{|S^c|}\|D\Sigma_S\|_2,
\]
we conclude that
\[
\operatorname{WorstBias}(\beta^\ast,\epsilon;S)
=
\epsilon K_S+o(\epsilon).
\]
This completes the proof.
\end{proof}

\begin{proof}[Proof of Theorem~\ref{thm: sufficiency}]
Let
\[
\mathcal S_0:=\mathcal S_{\mathrm{AD}}(\eta^*),
\qquad
W_S(\epsilon):=\operatorname{WorstBias}(\beta^\ast,\epsilon;S).
\]
Since \(d_\eta<\infty\), the number of possible partitions is finite, and hence
\(\mathcal S_0\) is finite. Let
\[
S^*=\arg\min_{S\in\mathcal S_0}K_S.
\]
By Assumption~\ref{ass: uniq}, this minimizer is unique. By Lemma~\ref{lemma: WorstBias}, for
each \(S\in\mathcal S_0\),
\[
W_S(\epsilon)=\epsilon K_S+o_S(\epsilon).
\]
Equivalently, for \(0<\epsilon\) small,
\[
W_S(\epsilon)=\epsilon\{K_S+r_S(\epsilon)\},
\qquad
r_S(\epsilon)\to0
\quad\text{as }\epsilon\downarrow0.
\]
Because \(\mathcal S_0\) is finite,
\[
\rho(\epsilon):=\max_{S\in\mathcal S_0}|r_S(\epsilon)|
\longrightarrow0
\qquad\text{as }\epsilon\downarrow0.
\]

If \(\mathcal S_0=\{S^*\}\), the result is immediate. Otherwise, define the strict gap
\[
\Delta
:=
\min_{S\in\mathcal S_0,\ S\neq S^*}
\{K_S-K_{S^*}\}.
\]
By uniqueness of \(S^*\), \(\Delta>0\). Since \(\rho(\epsilon)\to0\), there exists
\(E>0\) such that, for all \(0<\epsilon\leq E\),
\[
\rho(\epsilon)<\frac{\Delta}{3}.
\]
Hence, for any \(S\in\mathcal S_0\) with \(S\neq S^*\) and any
\(0<\epsilon\leq E\),
\[
\begin{aligned}
W_S(\epsilon)-W_{S^*}(\epsilon)
&=
\epsilon\{K_S-K_{S^*}\}
+
\epsilon\{r_S(\epsilon)-r_{S^*}(\epsilon)\} \\
&\geq
\epsilon\Delta-2\epsilon\rho(\epsilon) \\
&>
\epsilon\Delta-\frac{2}{3}\epsilon\Delta \\
&=
\frac{1}{3}\epsilon\Delta
>
0.
\end{aligned}
\]
Therefore \(W_{S^*}(\epsilon)<W_S(\epsilon)\) for every
\(S\neq S^*\) and every \(0<\epsilon\leq E\). Thus \(S^*\) is the unique minimizer of
\(\operatorname{WorstBias}(\beta^\ast,\epsilon;S)\) over \(\mathcal S_0\) for all
\(0<\epsilon\leq E\). Since \(S^*\) is also the unique minimizer of \(K_S\), we obtain
\[
\arg\min_{S\in\mathcal S_{\mathrm{AD}}(\eta^*)}K_S
=
\arg\min_{S\in\mathcal S_{\mathrm{AD}}(\eta^*)}
\operatorname{WorstBias}(\beta^\ast,\epsilon;S),
\qquad
0<\epsilon\leq E.
\]
This proves the theorem.
\end{proof}

\begin{proof}[Proof of Theorem~\ref{thm: minimax}]
Let
\[
\mathcal S_0:=\mathcal S_{\mathrm{AD}}(\eta^*),
\qquad
W_S(\epsilon):=\operatorname{WorstBias}(\beta^\ast,\epsilon;S),
\]
and define
\[
R_S(n)
:=
\mathbb E\left[
\left\|
\widehat\gamma(\beta^\ast;S)-\gamma(\beta^\ast;S)
\right\|_2^2
\right].
\]
By the definition of the admissible set, every
\(S\in\mathcal S_0\) satisfies the second-moment decay condition
\[
R_S(n)=o(1).
\]
Since \(\mathcal S_0\) is finite, this convergence is uniform over admissible
partitions. Thus, defining
\[
r_n:=\max_{S\in\mathcal S_0} R_S(n),
\]
we have
\[
r_n=o(1).
\]

For each admissible partition, define the finite-sample mean bias
\[
m_S(n)
:=
\mathbb E\!\left[\widehat\gamma(\beta^\ast;S)\right]
-
\gamma(\beta^\ast;S),
\]
and let
\[
V_S(n)
:=
\operatorname{tr}\!\left\{
\operatorname{Var}\!\left(\widehat\gamma(\beta^\ast;S)\right)
\right\}.
\]
The second-moment identity gives
\[
R_S(n)
=
V_S(n)+\|m_S(n)\|_2^2.
\]
Therefore, uniformly over \(S\in\mathcal S_0\),
\[
V_S(n)\leq R_S(n)\leq r_n=o(1),
\qquad
\|m_S(n)\|_2\leq R_S(n)^{1/2}\leq r_n^{1/2}=o(1).
\]

Now fix \(S\in\mathcal S_0\) and
\(\beta_0\in B_\epsilon(\beta^\ast;S)\). Write
\[
b_S(\beta_0)
:=
\gamma(\beta^\ast;S)-\gamma(\beta_0;S).
\]
Using
\[
\widehat\gamma(\beta^\ast;S)-\gamma(\beta_0;S)
=
b_S(\beta_0)
+
\left[
\widehat\gamma(\beta^\ast;S)-\gamma(\beta^\ast;S)
\right],
\]
we obtain the exact decomposition
\[
\begin{aligned}
\operatorname{MSE}(\beta^\ast,\beta_0,S)
&=
\mathbb E\left[
\left\|
\widehat\gamma(\beta^\ast;S)-\gamma(\beta_0;S)
\right\|_2^2
\right] \\
&=
\|b_S(\beta_0)\|_2^2
+
2b_S(\beta_0)'m_S(n)
+
R_S(n).
\end{aligned}
\]
Hence
\[
\begin{aligned}
&
\left|
\operatorname{MSE}(\beta^\ast,\beta_0,S)
-
\left\|
\gamma(\beta^\ast;S)-\gamma(\beta_0;S)
\right\|_2^2
\right| \\
&\qquad
\leq
2
\left\|
\gamma(\beta^\ast;S)-\gamma(\beta_0;S)
\right\|_2
\|m_S(n)\|_2
+
R_S(n) \\
&\qquad
\leq
2W_S(\epsilon)r_n^{1/2}
+
r_n .
\end{aligned}
\]

For fixed \(0<\epsilon\leq E\), and because \(\mathcal S_0\) is finite,
\[
\bar W_\epsilon
:=
\max_{S\in\mathcal S_0} W_S(\epsilon)
<
\infty.
\]
Therefore
\[
\sup_{S\in\mathcal S_0}
\sup_{\beta_0\in B_\epsilon(\beta^\ast;S)}
\left|
\operatorname{MSE}(\beta^\ast,\beta_0,S)
-
\left\|
\gamma(\beta^\ast;S)-\gamma(\beta_0;S)
\right\|_2^2
\right|
\leq
2\bar W_\epsilon r_n^{1/2}+r_n
=
o(1).
\]
It follows that, uniformly over \(S\in\mathcal S_0\),
\[
\max_{\beta_0\in B_\epsilon(\beta^\ast;S)}
\operatorname{MSE}(\beta^\ast,\beta_0,S)
=
W_S(\epsilon)^2
+
\Delta_S(n,\epsilon),
\]
where
\[
\sup_{S\in\mathcal S_0}
|\Delta_S(n,\epsilon)|
=
o(1).
\]

Therefore,
\[
\begin{aligned}
&
\min_{S\in\mathcal S_{\mathrm{AD}}(\eta^*)}
\max_{\beta_0\in B_\epsilon(\beta^\ast;S)}
\operatorname{MSE}(\beta^\ast,\beta_0,S) \\
&\qquad
=
\min_{S\in\mathcal S_0}
\left\{
W_S(\epsilon)^2+\Delta_S(n,\epsilon)
\right\} \\
&\qquad
=
\min_{S\in\mathcal S_0}
W_S(\epsilon)^2
+
o(1).
\end{aligned}
\]

By Theorem~\ref{thm: sufficiency}, there exists \(E>0\) such that, for every
\(0<\epsilon\leq E\), the least-sensitive partition \(S^*\) minimizes
\(W_S(\epsilon)\) over \(\mathcal S_{\mathrm{AD}}(\eta^*)\). Hence
\[
\min_{S\in\mathcal S_0}
W_S(\epsilon)^2
=
W_{S^*}(\epsilon)^2.
\]
Since
\[
W_{S^*}(\epsilon)
=
\operatorname{WorstBias}(\beta^\ast,\epsilon;S^*),
\]
we conclude that
\[
\min_{S\in\mathcal S_{\mathrm{AD}}(\eta^*)}
\max_{\beta_0\in B_\epsilon(\beta^\ast;S)}
\operatorname{MSE}(\beta^\ast,\beta_0,S)
=
\operatorname{WorstBias}(\beta^\ast,\epsilon;S^*)^2
+
o(1).
\]
This proves the theorem.
\end{proof}

\begin{proof}[Proof of Theorem~\ref{thm: local identification}]
Let
\[
\mathcal S_0:=\mathcal S_{\mathrm{AD}}(\eta^*).
\]
By Assumption~\ref{ass:stable-selection}, there exists a neighborhood \(U_0\) of
\(\eta^*\) such that
\[
\mathcal S_{\mathrm{AD}}(\eta)=\mathcal S_0
\qquad\text{for all }\eta\in U_0.
\]
Since the number of partitions is finite, \(\mathcal S_0\) is finite. Let
\[
S^*:=S^*(\eta^*)
\in
\arg\min_{S\in\mathcal S_0}K_S(\eta^*).
\]
By Assumption~\ref{ass: uniq}, this minimizer is unique. If
\(\mathcal S_0=\{S^*\}\), the result is immediate. Otherwise define the strict gap
\[
\Delta
:=
\min_{S\in\mathcal S_0,\ S\neq S^*}
\{K_S(\eta^*)-K_{S^*}(\eta^*)\}.
\]
By uniqueness, \(\Delta>0\).

By Assumption~\ref{ass:stable-selection}, for every \(S\in\mathcal S_0\), the map
\(\eta\mapsto K_S(\eta)\) is continuous in a neighborhood of \(\eta^*\). Because
\(\mathcal S_0\) is finite, there exists a neighborhood \(U_1\subseteq U_0\) of
\(\eta^*\) such that, for every \(\eta\in U_1\) and every \(S\in\mathcal S_0\),
\[
|K_S(\eta)-K_S(\eta^*)|<\frac{\Delta}{3}.
\]
Then, for any \(S\in\mathcal S_0\) with \(S\neq S^*\) and any \(\eta\in U_1\),
\[
\begin{aligned}
K_S(\eta)-K_{S^*}(\eta)
&=
[K_S(\eta^*)-K_{S^*}(\eta^*)]
+[K_S(\eta)-K_S(\eta^*)]
-[K_{S^*}(\eta)-K_{S^*}(\eta^*)] \\
&\geq
\Delta-\frac{\Delta}{3}-\frac{\Delta}{3}
=
\frac{\Delta}{3}>0.
\end{aligned}
\]
Thus \(K_{S^*}(\eta)<K_S(\eta)\) for every
\(S\in\mathcal S_0\setminus\{S^*\}\) and every \(\eta\in U_1\). Since
\(\mathcal S_{\mathrm{AD}}(\eta)=\mathcal S_0\) on \(U_1\), it follows that
\[
S^*(\eta)=S^*(\eta^*)
\qquad\text{for all }\eta\in U_1.
\]
Taking \(U:=U_1\) proves the theorem.
\end{proof}

\begin{proof}[Proof of Proposition~\ref{prop:l2_alpha_rate_local_sequence}]
Let
\[
\Delta_n
:=
\widehat\alpha_n(\beta_n^\ast)-\alpha_n^\ast.
\]
For this proof, set
\[
A_n:=W_n^{1/2},
\qquad
\widehat A_n:=\widehat W^{1/2},
\qquad
J_n(\alpha)
:=
\frac{\partial g_n(\alpha,\beta_n^\ast)}{\partial\alpha'},
\qquad
G_n:=J_n(\alpha_n^\ast),
\]
so that \(C_n=(A_nG_n)^+\). All pathwise inequalities below are understood on the
probability-one event on which the localization and weight-equivalence conditions in
Assumption~\ref{ass:local-sequence-md} hold.

Since
\[
\widehat\alpha_n(\beta_n^\ast)\in\mathcal N_n(r_n)
\]
and \(\mathcal N_n(r_n)\) is convex, the line segment
\[
\alpha_n^\ast+t\Delta_n,
\qquad
t\in[0,1],
\]
lies inside \(\mathcal N_n(r_n)\). A first-order Taylor expansion of
\(\alpha\mapsto g_n(\alpha,\beta_n^\ast)\) around \(\alpha_n^\ast\) gives
\begin{equation}
g_n(\widehat\alpha_n(\beta_n^\ast),\beta_n^\ast)
=
g_n(\alpha_n^\ast,\beta_n^\ast)
+
G_n\Delta_n
+
R_n(\widehat\alpha_n(\beta_n^\ast),\beta_n^\ast),
\label{eq:first_taylor_expansion_gn_fixed_beta}
\end{equation}
where the Taylor residual has the mean-value integral form
\begin{equation}
R_n(\widehat\alpha_n(\beta_n^\ast),\beta_n^\ast)
=
\int_0^1
\left[
J_n(\alpha_n^\ast+t\Delta_n)-G_n
\right]
\Delta_n\,dt.
\label{eq:taylor_residual_integral_form_fixed_beta}
\end{equation}
Since
\[
g_n(\alpha_n^\ast,\beta_n^\ast)=0,
\]
equation~\eqref{eq:first_taylor_expansion_gn_fixed_beta} becomes
\[
g_n(\widehat\alpha_n(\beta_n^\ast),\beta_n^\ast)
=
G_n\Delta_n
+
R_n(\widehat\alpha_n(\beta_n^\ast),\beta_n^\ast).
\]
Premultiplying by \(A_n\) and then by \(C_n\) gives
\[
C_nA_ng_n(\widehat\alpha_n(\beta_n^\ast),\beta_n^\ast)
=
C_nA_nG_n\Delta_n
+
C_nA_nR_n(\widehat\alpha_n(\beta_n^\ast),\beta_n^\ast).
\]
By Assumption~\ref{ass:local-sequence-md}(i),
\[
C_nA_nG_n=I.
\]
Therefore
\[
\Delta_n
=
C_nA_ng_n(\widehat\alpha_n(\beta_n^\ast),\beta_n^\ast)
-
C_nA_nR_n(\widehat\alpha_n(\beta_n^\ast),\beta_n^\ast).
\]
Taking norms and using the triangle inequality,
\begin{equation}
\|\Delta_n\|
\leq
\left\|
C_nA_ng_n(\widehat\alpha_n(\beta_n^\ast),\beta_n^\ast)
\right\|
+
\left\|
C_nA_nR_n(\widehat\alpha_n(\beta_n^\ast),\beta_n^\ast)
\right\|.
\label{eq:two_terms_inequality_fixed_beta}
\end{equation}

We first control the Taylor residual. From
\eqref{eq:taylor_residual_integral_form_fixed_beta},
\[
C_nA_nR_n(\widehat\alpha_n(\beta_n^\ast),\beta_n^\ast)
=
\int_0^1
C_nA_n
\left[
J_n(\alpha_n^\ast+t\Delta_n)-G_n
\right]
\Delta_n\,dt.
\]
Since \(\alpha_n^\ast+t\Delta_n\in\mathcal N_n(r_n)\) for all \(t\in[0,1]\),
Assumption~\ref{ass:local-sequence-md}(v) implies
\[
\left\|
C_nA_nR_n(\widehat\alpha_n(\beta_n^\ast),\beta_n^\ast)
\right\|
\leq
\rho\|\Delta_n\|.
\]
Using this bound in
\eqref{eq:two_terms_inequality_fixed_beta},
\[
\|\Delta_n\|
\leq
\left\|
C_nA_ng_n(\widehat\alpha_n(\beta_n^\ast),\beta_n^\ast)
\right\|
+
\rho\|\Delta_n\|.
\]
Since \(\rho<1\),
\[
\|\Delta_n\|
\leq
\frac{1}{1-\rho}
\left\|
C_nA_ng_n(\widehat\alpha_n(\beta_n^\ast),\beta_n^\ast)
\right\|.
\]
Consequently,
\begin{equation}
\|\Delta_n\|^2
\leq
\frac{\|C_n\|^2}{(1-\rho)^2}
\left\|
A_ng_n(\widehat\alpha_n(\beta_n^\ast),\beta_n^\ast)
\right\|^2.
\label{eq:delta_bound_fixed_beta}
\end{equation}

It remains to control
\[
\mathbb{E}\left[
\left\|
A_ng_n(\widehat\alpha_n(\beta_n^\ast),\beta_n^\ast)
\right\|^2
\right].
\]
By definition,
\[
e_n(\alpha,\beta_n^\ast)
=
\widehat g_n(\alpha,\beta_n^\ast)
-
g_n(\alpha,\beta_n^\ast),
\]
so
\[
g_n(\widehat\alpha_n(\beta_n^\ast),\beta_n^\ast)
=
\widehat g_n(\widehat\alpha_n(\beta_n^\ast),\beta_n^\ast)
-
e_n(\widehat\alpha_n(\beta_n^\ast),\beta_n^\ast).
\]
Thus,
\[
\begin{aligned}
\left\|
A_ng_n(\widehat\alpha_n(\beta_n^\ast),\beta_n^\ast)
\right\|
&\leq
\left\|
A_n\widehat g_n(\widehat\alpha_n(\beta_n^\ast),\beta_n^\ast)
\right\|  \\
&\quad+
\|A_n\|
\left\|
e_n(\widehat\alpha_n(\beta_n^\ast),\beta_n^\ast)
\right\|.
\end{aligned}
\]

We now bound the first term using the sample criterion. By the uniform equivalence of
\(A_n\) and \(\widehat A_n\) in Assumption~\ref{ass:local-sequence-md}(ii),
\[
\left\|
A_n\widehat g_n(\widehat\alpha_n(\beta_n^\ast),\beta_n^\ast)
\right\|
\leq
K
\left\|
\widehat A_n
\widehat g_n(\widehat\alpha_n(\beta_n^\ast),\beta_n^\ast)
\right\|.
\]
By the minimizing property of
\(\widehat\alpha_n(\beta_n^\ast)\) under the sample criterion,
\[
\left\|
\widehat A_n
\widehat g_n(\widehat\alpha_n(\beta_n^\ast),\beta_n^\ast)
\right\|
\leq
\left\|
\widehat A_n
\widehat g_n(\alpha_n^\ast,\beta_n^\ast)
\right\|.
\]
Using again the uniform equivalence of \(A_n\) and \(\widehat A_n\),
\[
\left\|
\widehat A_n
\widehat g_n(\alpha_n^\ast,\beta_n^\ast)
\right\|
\leq
K
\left\|
A_n
\widehat g_n(\alpha_n^\ast,\beta_n^\ast)
\right\|.
\]
Since
\[
g_n(\alpha_n^\ast,\beta_n^\ast)=0,
\]
we have
\[
\widehat g_n(\alpha_n^\ast,\beta_n^\ast)
=
e_n(\alpha_n^\ast,\beta_n^\ast).
\]
Using \(\|A_n\|\leq K\), we obtain
\[
\left\|
A_n\widehat g_n(\widehat\alpha_n(\beta_n^\ast),\beta_n^\ast)
\right\|
\leq
K^3
\left\|
e_n(\alpha_n^\ast,\beta_n^\ast)
\right\|,
\]
where the constant \(K\) may change from line to line but is always finite and independent of
\(n\).

For the second term,
\[
\begin{aligned}
e_n(\widehat\alpha_n(\beta_n^\ast),\beta_n^\ast)
&=
e_n(\alpha_n^\ast,\beta_n^\ast) \\
&\quad+
\left[
e_n(\widehat\alpha_n(\beta_n^\ast),\beta_n^\ast)
-
e_n(\alpha_n^\ast,\beta_n^\ast)
\right].
\end{aligned}
\]
Since
\[
\widehat\alpha_n(\beta_n^\ast)\in\mathcal N_n(r_n)
\qquad\text{and}\qquad
\mathcal N_n(r_n)\subseteq\mathcal N_n(\bar r),
\]
we have
\[
\begin{aligned}
\left\|
e_n(\widehat\alpha_n(\beta_n^\ast),\beta_n^\ast)
\right\|
&\leq
\left\|
e_n(\alpha_n^\ast,\beta_n^\ast)
\right\| \\
&\quad+
\sup_{\alpha\in\mathcal N_n(\bar r)}
\left\|
e_n(\alpha,\beta_n^\ast)
-
e_n(\alpha_n^\ast,\beta_n^\ast)
\right\|.
\end{aligned}
\]
Combining the previous inequalities and using \(\|A_n\|\leq K\), there exists a finite
constant \(K_1<\infty\), independent of \(n\), such that
\[
\begin{aligned}
\left\|
A_ng_n(\widehat\alpha_n(\beta_n^\ast),\beta_n^\ast)
\right\|
&\leq
K_1
\left\|
e_n(\alpha_n^\ast,\beta_n^\ast)
\right\| \\
&\quad+
K_1
\sup_{\alpha\in\mathcal N_n(\bar r)}
\left\|
e_n(\alpha,\beta_n^\ast)
-
e_n(\alpha_n^\ast,\beta_n^\ast)
\right\|.
\end{aligned}
\]
Hence,
\[
\begin{aligned}
\left\|
A_ng_n(\widehat\alpha_n(\beta_n^\ast),\beta_n^\ast)
\right\|^2
&\leq
2K_1^2
\left\|
e_n(\alpha_n^\ast,\beta_n^\ast)
\right\|^2 \\
&\quad+
2K_1^2
\sup_{\alpha\in\mathcal N_n(\bar r)}
\left\|
e_n(\alpha,\beta_n^\ast)
-
e_n(\alpha_n^\ast,\beta_n^\ast)
\right\|^2.
\end{aligned}
\]
Taking expectations and using Assumption~\ref{ass:local-sequence-md}(iv),
\[
\begin{aligned}
\mathbb{E}\left[
\left\|
A_ng_n(\widehat\alpha_n(\beta_n^\ast),\beta_n^\ast)
\right\|^2
\right]
&\leq
2K_1^2
\mathbb{E}\left[
\left\|
e_n(\alpha_n^\ast,\beta_n^\ast)
\right\|^2
\right] \\
&\quad+
2K_1^2
\mathbb{E}\left[
\sup_{\alpha\in\mathcal N_n(\bar r)}
\left\|
e_n(\alpha,\beta_n^\ast)
-
e_n(\alpha_n^\ast,\beta_n^\ast)
\right\|^2
\right] \\
&=
O(q_n^2).
\end{aligned}
\]
Using \eqref{eq:delta_bound_fixed_beta}, we obtain
\[
\begin{aligned}
\mathbb{E}\left[
\|\widehat\alpha_n(\beta_n^\ast)-\alpha_n^\ast\|^2
\right]
&=
\mathbb{E}\left[\|\Delta_n\|^2\right] \\
&\leq
\frac{\|C_n\|^2}{(1-\rho)^2}
\mathbb{E}\left[
\left\|
A_ng_n(\widehat\alpha_n(\beta_n^\ast),\beta_n^\ast)
\right\|^2
\right] \\
&=
O(\|C_n\|^2q_n^2).
\end{aligned}
\]
Finally, if \(\|C_n\|q_n=o(1)\), then
\[
\mathbb{E}\left[
\|\widehat\alpha_n(\beta_n^\ast)-\alpha_n^\ast\|^2
\right]
=
O\left(\|C_n\|^2q_n^2\right)
=
o(1).
\]
This completes the proof.
\end{proof}

\begin{proof}[Proof of Proposition~\ref{prop:gamma_l2_rate_local_sequence}]
Let
\[
\Delta_{\alpha,n}
:=
\widehat\alpha_n(\beta_n^\ast)-\alpha_n^\ast.
\]
By Assumption~\ref{ass:local-sequence-md}(iii),
\[
\widehat\alpha_n(\beta_n^\ast)\in\mathcal N_n(r_n)
\]
with probability one for all sufficiently large \(n\). Since \(r_n\leq\bar r\), we also have
\[
\widehat\alpha_n(\beta_n^\ast)\in\mathcal N_n(\bar r)
\]
with probability one for all sufficiently large \(n\). Also,
\[
\beta_n^\ast\in\mathcal B_n(\bar s).
\]

Using the mean-value integral form for the map
\[
\alpha\mapsto\Gamma(\alpha,\beta_n^\ast),
\]
we obtain
\[
\begin{aligned}
\widehat\gamma_n-\gamma_n^\ast
&=
\Gamma\!\left(
\widehat\alpha_n(\beta_n^\ast),
\beta_n^\ast
\right)
-
\Gamma(\alpha_n^\ast,\beta_n^\ast) \\
&=
\int_0^1
\Gamma_{\alpha,n}
\!\left(
\alpha_n^\ast+t\Delta_{\alpha,n},
\beta_n^\ast
\right)
\Delta_{\alpha,n}
\,dt,
\end{aligned}
\]
where
\[
\Gamma_{\alpha,n}(\alpha,\beta)
:=
\frac{\partial\Gamma(\alpha,\beta)}{\partial\alpha'}.
\]

Because \(\mathcal N_n(\bar r)\) is convex, and because
\[
\alpha_n^\ast,
\widehat\alpha_n(\beta_n^\ast)
\in
\mathcal N_n(\bar r),
\]
the point
\[
\left(
\alpha_n^\ast+t\Delta_{\alpha,n},
\beta_n^\ast
\right)
\]
belongs to
\[
\mathcal N_n(\bar r)\times\mathcal B_n(\bar s)
\]
for all \(t\in[0,1]\). Therefore, by Assumption~\ref{ass:local-sequence-md}(vi),
\[
\begin{aligned}
\|\widehat\gamma_n-\gamma_n^\ast\|
&\leq
\int_0^1
\left\|
\Gamma_{\alpha,n}
\!\left(
\alpha_n^\ast+t\Delta_{\alpha,n},
\beta_n^\ast
\right)
\right\|
\|\Delta_{\alpha,n}\|
dt \\
&\leq
L_\Gamma
\|\Delta_{\alpha,n}\|.
\end{aligned}
\]
Squaring gives
\[
\|\widehat\gamma_n-\gamma_n^\ast\|^2
\leq
L_\Gamma^2
\|\Delta_{\alpha,n}\|^2.
\]
Taking expectations,
\[
\mathbb{E}\left[
\|\widehat\gamma_n-\gamma_n^\ast\|^2
\right]
\leq
L_\Gamma^2
\mathbb{E}\left[
\|\widehat\alpha_n(\beta_n^\ast)-\alpha_n^\ast\|^2
\right].
\]
By Proposition~\ref{prop:l2_alpha_rate_local_sequence},
\[
\mathbb{E}\left[
\|\widehat\alpha_n(\beta_n^\ast)-\alpha_n^\ast\|^2
\right]
=
O\left(\|C_n\|^2q_n^2\right).
\]
Hence
\[
\mathbb{E}\left[
\|\widehat\gamma_n-\gamma_n^\ast\|^2
\right]
=
O\left(\|C_n\|^2q_n^2\right).
\]

Finally, for any random vector \(Y\) and any nonrandom vector \(y\),
\[
E\|Y-y\|^2
=
\operatorname{tr}\{\operatorname{Var}(Y)\}
+
\|E[Y]-y\|^2.
\]
Applying this identity with
\[
Y=\widehat\gamma_n,
\qquad
y=\gamma_n^\ast,
\]
we obtain
\[
\operatorname{tr}\{\operatorname{Var}(\widehat\gamma_n)\}
\leq
\mathbb{E}\left[
\|\widehat\gamma_n-\gamma_n^\ast\|^2
\right].
\]
Therefore
\[
\operatorname{tr}\{\operatorname{Var}(\widehat\gamma_n)\}
=
O\left(\|C_n\|^2q_n^2\right).
\]
If \(\|C_n\|q_n=o(1)\), then
\[
\mathbb{E}\left[
\|\widehat\gamma_n-\gamma_n^\ast\|^2
\right]
=
O\left(\|C_n\|^2q_n^2\right)
=
o(1),
\]
and
\[
\operatorname{tr}\{\operatorname{Var}(\widehat\gamma_n)\}
=
O\left(\|C_n\|^2q_n^2\right)
=
o(1).
\]
Thus, the variance of the object of interest vanishes as \(n\to\infty\). This completes the proof.
\end{proof}

\begin{lemma}[Consistency of the sensitivity statistic]
\label{lem:sensitivity-stat-consistency}
\label{thm: consistency of Sensitivity statistic}
Suppose Assumptions~\ref{ass: regularity conditions},
\ref{ass: regularity assumptions reference point}, and
\ref{ass:partition-consistency-highlevel} hold. Then, for every
\(S\in\mathcal S_{\mathrm{AD}}(\eta^*)\),
\[
\widehat K_S(\widehat\eta)\overset{p}{\to}K_S(\eta^*).
\]
Moreover, since \(\mathcal S_{\mathrm{AD}}(\eta^*)\) is finite,
\[
\max_{S\in\mathcal S_{\mathrm{AD}}(\eta^*)}
|\widehat K_S(\widehat\eta)-K_S(\eta^*)|
\overset{p}{\to}0.
\]
\end{lemma}

\begin{proof}
Fix \(S\in\mathcal S_{\mathrm{AD}}(\eta^*)\). Write
\[
G_{\eta,S}(\eta^*)
:=
\frac{\partial g(\eta^*;S)}{\partial\eta'},
\qquad
\widehat G_{\eta,S}(\widehat\eta)
:=
\frac{\partial\widehat g(\widehat\eta;S)}{\partial\eta'}.
\]
By Assumption~\ref{ass:partition-consistency-highlevel}(i),
\(\widehat\eta\overset{p}{\to}\eta^*\). By
Assumption~\ref{ass:partition-consistency-highlevel}(iii) and continuity of the population
Jacobian,
\[
\begin{aligned}
\|\widehat G_{\eta,S}(\widehat\eta)-G_{\eta,S}(\eta^*)\|
&\leq
\sup_{\eta\in N}
\left\|
\frac{\partial\widehat g(\eta;S)}{\partial\eta'}
-
\frac{\partial g(\eta;S)}{\partial\eta'}
\right\| \\
&\quad+
\|G_{\eta,S}(\widehat\eta)-G_{\eta,S}(\eta^*)\| \\
&\overset{p}{\to}0.
\end{aligned}
\]
Therefore the partition-specific submatrices satisfy
\[
\widehat G_{\alpha,S}\overset{p}{\to}G_{\alpha,S},
\qquad
\widehat G_{\beta,S}\overset{p}{\to}G_{\beta,S}.
\]

Since \(S\in\mathcal S_{\mathrm{AD}}(\eta^*)\), the conditional local identification
condition gives
\[
\operatorname{rank}(G_{\alpha,S})=|S|.
\]
Because \(W\) is positive definite, 
\[
M_S:=G_{\alpha,S}'WG_{\alpha,S}
\]
is nonsingular. Hence, with probability approaching one,
\[
\widehat M_S
:=
\widehat G_{\alpha,S}'\widehat W\widehat G_{\alpha,S}
\]
is nonsingular. By the continuous mapping theorem and
Assumption~\ref{ass:partition-consistency-highlevel}(iv),
\[
\widehat D_{\alpha\beta,S}
:=
-
\left(\widehat G_{\alpha,S}'\widehat W\widehat G_{\alpha,S}\right)^{-1}
\widehat G_{\alpha,S}'\widehat W\widehat G_{\beta,S}
\overset{p}{\to}
D_{\alpha\beta,S},
\]
where
\[
D_{\alpha\beta,S}
:=
-
\left(G_{\alpha,S}'WG_{\alpha,S}\right)^{-1}
G_{\alpha,S}'WG_{\beta,S}.
\]

By Assumption~\ref{ass: regularity conditions}, \(\Gamma\) is continuously differentiable.
Thus
\[
\widehat\Gamma_{\alpha,S}(\widehat\eta)\overset{p}{\to}\Gamma_{\alpha,S}(\eta^*),
\qquad
\widehat\Gamma_{\beta,S}(\widehat\eta)\overset{p}{\to}\Gamma_{\beta,S}(\eta^*).
\]
Therefore
\[
\widehat D_{\gamma\beta,S}
:=
\widehat\Gamma_{\alpha,S}(\widehat\eta)\widehat D_{\alpha\beta,S}
+
\widehat\Gamma_{\beta,S}(\widehat\eta)
\overset{p}{\to}
D_{\gamma\beta,S},
\]
where
\[
D_{\gamma\beta,S}
:=
\Gamma_{\alpha,S}(\eta^*)D_{\alpha\beta,S}
+
\Gamma_{\beta,S}(\eta^*)
=
\frac{\partial\gamma(\beta^\ast;S)}{\partial\beta'}.
\]
Finally, since the spectral norm is continuous and
Assumption~\ref{ass:partition-consistency-highlevel}(iv),
\[
\widehat K_S(\widehat\eta)
=
\sqrt{|S^c|}
\|\widehat D_{\gamma\beta,S}\Sigma_S\|_2
\overset{p}{\to}
\sqrt{|S^c|}
\|D_{\gamma\beta,S}\Sigma_S\|_2
=
K_S(\eta^*).
\]
This proves pointwise convergence. Since \(\mathcal S_{\mathrm{AD}}(\eta^*)\) is finite,
pointwise convergence implies uniform convergence over \(\mathcal S_{\mathrm{AD}}(\eta^*)\).
\end{proof}

\begin{proof}[Proof of Theorem~\ref{thm: Consistency of partition}]
Let
\[
\mathcal S_0:=\mathcal S_{\mathrm{AD}}(\eta^*),
\qquad
S^*:=S^*(\eta^*)
\in
\arg\min_{S\in\mathcal S_0}K_S(\eta^*).
\]
By Assumption~\ref{ass: uniq}, \(S^*\) is the unique minimizer. If
\(\mathcal S_0=\{S^*\}\), then the result follows directly from the admissible-set
consistency condition in Assumption~\ref{ass:partition-consistency-highlevel}(ii). Hence
suppose \(\mathcal S_0\setminus\{S^*\}\neq\varnothing\), and define the strict gap
\[
\Delta
:=
\min_{S\in\mathcal S_0,\ S\neq S^*}
\{K_S(\eta^*)-K_{S^*}(\eta^*)\}.
\]
By uniqueness, \(\Delta>0\).

Define the events
\[
A_n:=
\{\widehat{\mathcal S}_{\mathrm{AD}}(\widehat\eta)=\mathcal S_0\},
\]
and
\[
B_n:=
\left\{
\max_{S\in\mathcal S_0}
|\widehat K_S(\widehat\eta)-K_S(\eta^*)|
<
\frac{\Delta}{3}
\right\}.
\]
By Assumption~\ref{ass:partition-consistency-highlevel}(ii),
\[
\Pr(A_n)\to1.
\]
By Lemma~\ref{lem:sensitivity-stat-consistency},
\[
\Pr(B_n)\to1.
\]
Therefore
\[
\Pr(A_n\cap B_n)\to1.
\]

On \(A_n\cap B_n\), the estimated admissible set equals the population admissible set.
Moreover, for every \(S\in\mathcal S_0\) with \(S\neq S^*\),
\[
\begin{aligned}
\widehat K_S(\widehat\eta)-\widehat K_{S^*}(\widehat\eta)
&=
[K_S(\eta^*)-K_{S^*}(\eta^*)] \\
&\quad+
[\widehat K_S(\widehat\eta)-K_S(\eta^*)]
-[\widehat K_{S^*}(\widehat\eta)-K_{S^*}(\eta^*)] \\
&>
\Delta-\frac{\Delta}{3}-\frac{\Delta}{3}
=
\frac{\Delta}{3}>0.
\end{aligned}
\]
Thus, on \(A_n\cap B_n\), \(S^*\) is the unique minimizer of
\(\widehat K_S(\widehat\eta)\) over \(\widehat{\mathcal S}_{\mathrm{AD}}(\widehat\eta)\).
Hence
\[
\widehat S^*(\widehat\eta)=S^*(\eta^*)
\]
on \(A_n\cap B_n\). Therefore,
\[
\Pr\{\widehat S^*(\widehat\eta)\neq S^*(\eta^*)\}
\leq
\Pr(A_n^c)+\Pr(B_n^c)
\longrightarrow0.
\]
This proves the consistency of the estimated least-sensitive partition.
\end{proof}

\begin{proof}[Proof of Theorem~\ref{thm:direct-reestimated-derivative}]
Define the sample first-order condition map
\[
\widehat F(\alpha,\beta;S)
:=
\widehat G_{\alpha}(\alpha,\beta;S)'
\widehat W
\widehat g(\alpha,\beta;S).
\]
Let
\[
\widetilde\eta_{S}
:=
(\widehat\alpha(\widehat\beta_{S};S),\widehat\beta_{S}).
\]
By Assumption~\ref{ass:direct-derivative},
\[
\widetilde\eta_{S}\overset{p}{\to}\eta^*.
\]
The sample first-order condition is
\[
\widehat F(\widetilde\eta_{S};S)=0.
\]
Let
\[
\widehat F_{\alpha}(\eta;S)
:=
\frac{\partial\widehat F(\eta;S)}{\partial\alpha'},
\qquad
\widehat F_{\beta}(\eta;S)
:=
\frac{\partial\widehat F(\eta;S)}{\partial\beta'}.
\]
By the product rule, evaluated at \(\widetilde\eta_{S}\),
\[
\widehat F_{\alpha}
=
\widehat G_{\alpha}'\widehat W\widehat G_{\alpha}
+
\widehat R_{\alpha,n},
\]
where \(\widehat R_{\alpha}\) collects the terms involving second derivatives of
\(\widehat g\) multiplied by \(\widehat W\widehat g\). Similarly,
\[
\widehat F_{\beta}
=
\widehat G_{\alpha}'\widehat W\widehat G_{\beta}
+
\widehat R_{\beta},
\]
where \(\widehat R_{\beta}\) collects the mixed second-derivative terms multiplied by
\(\widehat W\widehat g\). All matrices in these two displays are evaluated at
\(\widetilde\eta_{S}\).

Because second derivatives are locally bounded with probability approaching one,
\(\widehat W=O_p(1)\), and
\[
\widehat g(\widetilde\eta_{S};S)
\overset{p}{\to}
g(\eta^*;S)=0,
\]
we have
\[
\widehat R_{\alpha}=o_p(1),
\qquad
\widehat R_{\beta}=o_p(1).
\]
Therefore,
\[
\widehat F_{\alpha}(\widetilde\eta_{S};S)
=
\widehat G_{\alpha}(\widetilde\eta_{S};S)'
\widehat W
\widehat G_{\alpha}(\widetilde\eta_{S};S)+o_p(1),
\]
and
\[
\widehat F_{\beta}(\widetilde\eta_{S};S)
=
\widehat G_{\alpha}(\widetilde\eta_{S};S)'
\widehat W
\widehat G_{\beta}(\widetilde\eta_{S};S)+o_p(1).
\]
Uniform convergence of the Jacobians, consistency of \(\widetilde\eta_{S}\), and
\(\widehat W\overset{p}{\to}W\) imply
\[
\widehat F_{\alpha}(\widetilde\eta_{S};S)
\overset{p}{\to}
G_{\alpha,S}'WG_{\alpha,S},
\]
and
\[
\widehat F_{\beta}(\widetilde\eta_{S};S)
\overset{p}{\to}
G_{\alpha,S}'WG_{\beta,S}.
\]
Since \(G_{\alpha,S}\) has full column rank and \(W\) is positive definite,
\(G_{\alpha,S}'WG_{\alpha,S}\) is nonsingular. Hence
\(\widehat F_{\alpha}(\widetilde\eta_{S};S)\) is nonsingular with probability approaching
one. The sample implicit function theorem gives
\[
\left.
\frac{\partial\widehat\alpha(\beta;S)}{\partial\beta'}
\right|_{\beta=\widehat\beta_{S}}
=
-
\left[
\widehat F_{\alpha}(\widetilde\eta_{S};S)
\right]^{-1}
\widehat F_{\beta}(\widetilde\eta_{S};S).
\]
By the continuous mapping theorem,
\[
\left.
\frac{\partial\widehat\alpha(\beta;S)}{\partial\beta'}
\right|_{\beta=\widehat\beta_{S}}
\overset{p}{\to}
-
(G_{\alpha,S}'WG_{\alpha,S})^{-1}G_{\alpha,S}'WG_{\beta,S}
=
D_{\alpha\beta,S}.
\]

Now consider the direct derivative used in the algorithm. Since
\(\widehat\eta_n\overset{p}{\to}\eta^*\), the same uniform convergence and continuous
mapping argument implies
\[
\widehat D_{\alpha\beta}(S)
\overset{p}{\to}
D_{\alpha\beta,S}.
\]
Consequently,
\[
\left.
\frac{\partial\widehat\alpha(\beta;S)}{\partial\beta'}
\right|_{\beta=\widehat\beta_{S}}
-
\widehat D_{\alpha\beta}(S)
=
o_p(1).
\]
Finally, applying the population implicit function theorem to
\[
F(\alpha,\beta;S):=G_\alpha(\alpha,\beta;S)'Wg(\alpha,\beta;S)
\]
at \((\alpha^\ast,\beta^\ast)\) gives
\[
\frac{\partial\alpha(\beta^\ast;S)}{\partial\beta'}
=
-
(G_{\alpha,S}'WG_{\alpha,S})^{-1}G_{\alpha,S}'WG_{\beta,S}
=
D_{\alpha\beta,S}.
\]
The second-derivative terms in the derivative of the population first-order condition vanish
because \(g(\eta^*;S)=0\). This completes the proof.
\end{proof}

\begin{proof}[Proof of Corollary~\ref{cor:direct-target-derivative}]
By the chain rule,
\[
\left.
\frac{\partial\widehat\gamma(\beta;S)}{\partial\beta'}
\right|_{\beta=\widehat\beta_{S}}
=
\Gamma_\alpha(\widetilde\eta_{S};S)
\left.
\frac{\partial\widehat\alpha(\beta;S)}{\partial\beta'}
\right|_{\beta=\widehat\beta_{S}}
+
\Gamma_\beta(\widetilde\eta_{S};S).
\]
Theorem~\ref{thm:direct-reestimated-derivative} gives the convergence of the derivative of
the re-estimation map. Continuity of \(\Gamma_\alpha\) and \(\Gamma_\beta\), together with
\[
\widetilde\eta_{S}\overset{p}{\to}\eta^*,
\qquad
\widehat\eta\overset{p}{\to}\eta^*,
\]
gives
\[
\left.
\frac{\partial\widehat\gamma(\beta;S)}{\partial\beta'}
\right|_{\beta=\widehat\beta_{S}}
=
D_{\gamma\beta,S}+o_p(1).
\]
The same continuity and continuous-mapping argument applied to the direct derivative gives
\[
\widehat D_{\gamma\beta}(S)
=
D_{\gamma\beta,S}+o_p(1).
\]
Therefore
\[
\left.
\frac{\partial\widehat\gamma(\beta;S)}{\partial\beta'}
\right|_{\beta=\widehat\beta_{S}}
=
\widehat D_{\gamma\beta}(S)+o_p(1),
\]
and the proof is complete.
\end{proof}

\section[Main Tables: Nakamura and Steinsson Application and Simulation]{Main Tables \cite{nakamura2018high} Application and Simulation}\label{appendix:MainresultsNS18}

\renewcommand{\arraystretch}{1.5}

\begin{longtable}{@{}lcccc@{}}
\caption{\textbf{Set $\widehat{\mathcal{S}}_{\mathrm{AD}}(\widehat{\eta})$ and Sensitivity Statistic $\widehat{K}(\widehat{\beta}, S)$}}
\label{tab: Sensitivities}\\

\hline \hline
\textit{Name} & \textit{Estimated} & \textit{Calibrated} &
$\widehat{K}(\widehat{\beta};S)$ & 5\%\textit{ bounds} \\
\hline \hline
\endfirsthead

\multicolumn{5}{c}{\textit{Table \thetable\ (continued)}}\\
\hline \hline
\textit{Name} & \textit{Estimated Parameters} & \textit{Fixed Parameters} &
$\widehat{K}(\widehat{\beta};S)$ & 5\%\textit{ bounds} \\
\hline \hline
\endhead

\hline \hline
\endfoot

\hline \hline
\endlastfoot
\textit{Original} & $( \rho_1, \rho_2, \gamma, \kappa\zeta)$ &$(\rho,\delta,\omega,\theta,\sigma,b) $ & 1.82 & [0.58, 0.768] \\
$S^*$ & $(\delta, \sigma, \rho_{1}, \rho_{2}, \gamma, b)$ & $(\rho, \omega, \theta, \kappa\zeta)$ & 0.06 & [0.67, 0.68] \\
$S_{2}$ & $(\rho, \delta, \rho_{1}, \rho_{2}, \gamma, b)$ & $(\omega, \theta, \sigma, \kappa\zeta)$ & 0.32 & [0.66, 0.69] \\
$S_{3}$ & $(\rho, \rho_{1}, \rho_{2}, \gamma, b)$ & $(\delta, \omega, \theta, \sigma, \kappa\zeta)$ & 0.36 & [0.66, 0.70] \\
$S_{4}$ & $(\rho, \delta, \rho_{2}, \gamma, b)$ & $(\omega, \theta, \sigma, \rho_{1}, \kappa\zeta)$ & 0.38 & [0.66, 0.70] \\
$S_{5}$ & $(\sigma, \rho_{1}, \rho_{2}, \gamma, b)$ & $(\rho, \delta, \omega, \theta, \kappa\zeta)$ & 0.53 & [0.65, 0.70] \\
$S_{6}$ & $(\delta, \rho_{1}, \rho_{2}, \gamma, b)$ & $(\rho, \omega, \theta, \sigma, \kappa\zeta)$ & 0.61 & [0.65, 0.71] \\
$S_{7}$ & $(\rho, \delta, \rho_{1}, \gamma, b)$ & $(\omega, \theta, \sigma, \rho_{2}, \kappa\zeta)$ & 0.61 & [0.65, 0.71] \\
$S_{8}$ & $(\rho_{1}, \rho_{2}, \gamma, b)$ & $(\rho, \delta, \omega, \theta, \sigma, \kappa\zeta)$ & 0.78 & [0.64, 0.72] \\
$S_{9}$ & $(\rho, \rho_{1}, \gamma, b)$ & $(\delta, \omega, \theta, \sigma, \rho_{2}, \kappa\zeta)$ & 0.79 & [0.64, 0.72] \\
$S_{10}$ & $(\delta, \rho_{1}, \gamma, b)$ & $(\rho, \omega, \theta, \sigma, \rho_{2}, \kappa\zeta)$ & 0.84 & [0.64, 0.72] \\
$S_{11}$ & $(\rho_{1}, \gamma, b)$ & $(\rho, \delta, \omega, \theta, \sigma, \rho_{2}, \kappa\zeta)$ & 0.92 & [0.63, 0.72] \\
$S_{12}$ & $(\rho, \delta, \sigma, \rho_{1}, \rho_{2}, \gamma)$ & $(\omega, \theta, \kappa\zeta, b)$ & 1.05 & [0.63, 0.73] \\
$S_{13}$ & $(\rho, \sigma, \rho_{1}, \rho_{2}, \gamma)$ & $(\delta, \omega, \theta, \kappa\zeta, b)$ & 1.12 & [0.62, 0.73] \\
$S_{14}$ & $(\rho, \sigma, \rho_{1}, \rho_{2}, \gamma, b)$ & $(\delta, \omega, \theta, \kappa\zeta)$ & 1.19 & [0.62, 0.74] \\
$S_{15}$ & $(\rho, \delta, \sigma, \rho_{1}, \gamma)$ & $(\omega, \theta, \rho_{2}, \kappa\zeta, b)$ & 1.22 & [0.62, 0.74] \\
$S_{16}$ & $(\rho, \delta, \sigma, \rho_{2}, \gamma)$ & $(\omega, \theta, \rho_{1}, \kappa\zeta, b)$ & 1.23 & [0.62, 0.74] \\
$S_{17}$ & $(\sigma, \rho_{1}, \gamma, b)$ & $(\rho, \delta, \omega, \theta, \rho_{2}, \kappa\zeta)$ & 1.25 & [0.61, 0.74] \\
$S_{18}$ & $(\delta, \sigma, \rho_{1}, \gamma, b)$ & $(\rho, \omega, \theta, \rho_{2}, \kappa\zeta)$ & 1.26 & [0.61, 0.74] \\
$S_{19}$ & $(\delta, \sigma, \rho_{2}, \gamma, b)$ & $(\rho, \omega, \theta, \rho_{1}, \kappa\zeta)$ & 1.34 & [0.61, 0.74] \\
$S_{20}$ & $(\rho, \rho_{2}, \gamma, b)$ & $(\delta, \omega, \theta, \sigma, \rho_{1}, \kappa\zeta)$ & 1.53 & [0.60, 0.75] \\
$S_{21}$ & $(\rho, \sigma, \rho_{2}, \gamma)$ & $(\delta, \omega, \theta, \rho_{1}, \kappa\zeta, b)$ & 1.53 & [0.60, 0.75] \\
$S_{22}$ & $(\rho, \sigma, \rho_{1}, \gamma)$ & $(\delta, \omega, \theta, \rho_{2}, \kappa\zeta, b)$ & 1.59 & [0.60, 0.76] \\
$S_{23}$ & $(\delta, \sigma, \rho_{1}, \rho_{2}, \gamma)$ & $(\rho, \omega, \theta, \kappa\zeta, b)$ & 1.63 & [0.60, 0.76] \\
$S_{24}$ & $(\rho, \sigma, \rho_{1}, \gamma, b)$ & $(\delta, \omega, \theta, \rho_{2}, \kappa\zeta)$ & 1.73 & [0.59, 0.76] \\
$S_{25}$ & $(\delta, \rho_{1}, \rho_{2}, \gamma)$ & $(\rho, \omega, \theta, \sigma, \kappa\zeta, b)$ & 1.82 & [0.59, 0.77] \\
$S_{26}$ & $(\delta, \sigma, \rho_{1}, \gamma)$ & $(\rho, \omega, \theta, \rho_{2}, \kappa\zeta, b)$ & 1.86 & [0.58, 0.77] \\
$S_{27}$ & $(\rho, \delta, \rho_{1}, \rho_{2}, \gamma)$ & $(\omega, \theta, \sigma, \kappa\zeta, b)$ & 1.93 & [0.58, 0.77] \\
$S_{28}$ & $(\rho, \delta, \rho_{1}, \gamma)$ & $(\omega, \theta, \sigma, \rho_{2}, \kappa\zeta, b)$ & 1.94 & [0.58, 0.77] \\
$S_{29}$ & $(\sigma, \rho_{2}, \gamma, b)$ & $(\rho, \delta, \omega, \theta, \rho_{1}, \kappa\zeta)$ & 1.96 & [0.58, 0.78] \\
$S_{30}$ & $(\delta, \rho_{1}, \gamma)$ & $(\rho, \omega, \theta, \sigma, \rho_{2}, \kappa\zeta, b)$ & 1.96 & [0.58, 0.78] \\
$S_{31}$ & $(\sigma, \rho_{1}, \rho_{2}, \gamma)$ & $(\rho, \delta, \omega, \theta, \kappa\zeta, b)$ & 1.98 & [0.58, 0.78] \\
$S_{32}$ & $(\rho_{1}, \rho_{2}, \gamma)$ & $(\rho, \delta, \omega, \theta, \sigma, \kappa\zeta, b)$ & 1.99 & [0.58, 0.78] \\
$S_{33}$ & $(\rho, \rho_{1}, \rho_{2}, \gamma)$ & $(\delta, \omega, \theta, \sigma, \kappa\zeta, b)$ & 2.07 & [0.57, 0.78] \\
$S_{34}$ & $(\sigma, \rho_{1}, \gamma)$ & $(\rho, \delta, \omega, \theta, \rho_{2}, \kappa\zeta, b)$ & 2.09 & [0.57, 0.78] \\
$S_{35}$ & $(\rho_{1}, \gamma)$ & $(\rho, \delta, \omega, \theta, \sigma, \rho_{2}, \kappa\zeta, b)$ & 2.13 & [0.57, 0.78] \\
$S_{36}$ & $(\rho, \rho_{1}, \gamma)$ & $(\delta, \omega, \theta, \sigma, \rho_{2}, \kappa\zeta, b)$ & 2.16 & [0.57, 0.79] \\
$S_{37}$ & $(\rho, \delta, \rho_{2}, \gamma)$ & $(\omega, \theta, \sigma, \rho_{1}, \kappa\zeta, b)$ & 2.39 & [0.56, 0.80] \\
$S_{38}$ & $(\rho, \rho_{2}, \gamma)$ & $(\delta, \omega, \theta, \sigma, \rho_{1}, \kappa\zeta, b)$ & 2.89 & [0.53, 0.82] \\
$S_{39}$ & $(\rho, \sigma, \rho_{2}, \gamma, b)$ & $(\delta, \omega, \theta, \rho_{1}, \kappa\zeta)$ & 3.70 & [0.49, 0.86] \\
$S_{40}$ & $(\delta, \sigma, \rho_{2}, \gamma)$ & $(\rho, \omega, \theta, \rho_{1}, \kappa\zeta, b)$ & 3.71 & [0.49, 0.86] \\
$S_{41}$ & $(\sigma, \rho_{2}, \gamma)$ & $(\rho, \delta, \omega, \theta, \rho_{1}, \kappa\zeta, b)$ & 4.46 & [0.45, 0.90] \\
$S_{42}$ & $(\delta, \rho_{2}, \gamma, b)$ & $(\rho, \omega, \theta, \sigma, \rho_{1}, \kappa\zeta)$ & 4.85 & [0.43, 0.92] \\
$S_{43}$ & $(\delta, \rho_{2}, \gamma)$ & $(\rho, \omega, \theta, \sigma, \rho_{1}, \kappa\zeta, b)$ & 4.94 & [0.43, 0.92] \\
$S_{44}$ & $(\rho_{2}, \gamma)$ & $(\rho, \delta, \omega, \theta, \sigma, \rho_{1}, \kappa\zeta, b)$ & 5.56 & [0.40, 0.96] \\
$S_{45}$ & $(\rho_{2}, \gamma, b)$ & $(\rho, \delta, \omega, \theta, \sigma, \rho_{1}, \kappa\zeta)$ & 5.57 & [0.40, 0.96] \\
$S_{46}$ & $(\rho, \delta, \sigma, \gamma)$ & $(\omega, \theta, \rho_{1}, \rho_{2}, \kappa\zeta, b)$ & 9.25 & [0.22, 1.14] \\
$S_{47}$ & $(\rho, \delta, \gamma, b)$ & $(\omega, \theta, \sigma, \rho_{1}, \rho_{2}, \kappa\zeta)$ & 9.74 & [0.19, 1.16] \\
$S_{48}$ & $(\delta, \sigma, \gamma)$ & $(\rho, \omega, \theta, \rho_{1}, \rho_{2}, \kappa\zeta, b)$ & 9.83 & [0.19, 1.17] \\
$S_{49}$ & $(\rho, \delta, \gamma)$ & $(\omega, \theta, \sigma, \rho_{1}, \rho_{2}, \kappa\zeta, b)$ & 9.95 & [0.18, 1.17] \\
$S_{50}$ & $(\rho, \sigma, \gamma)$ & $(\delta, \omega, \theta, \rho_{1}, \rho_{2}, \kappa\zeta, b)$ & 10.00 & [0.18, 1.18] \\
$S_{51}$ & $(\delta, \gamma, b)$ & $(\rho, \omega, \theta, \sigma, \rho_{1}, \rho_{2}, \kappa\zeta)$ & 10.15 & [0.17, 1.19] \\
$S_{52}$ & $(\delta, \sigma, \gamma, b)$ & $(\rho, \omega, \theta, \rho_{1}, \rho_{2}, \kappa\zeta)$ & 10.23 & [0.17, 1.19] \\
$S_{53}$ & $(\rho, \sigma, \gamma, b)$ & $(\delta, \omega, \theta, \rho_{1}, \rho_{2}, \kappa\zeta)$ & 10.43 & [0.16, 1.20] \\
$S_{54}$ & $(\delta, \gamma)$ & $(\rho, \omega, \theta, \sigma, \rho_{1}, \rho_{2}, \kappa\zeta, b)$ & 10.43 & [0.16, 1.20] \\
$S_{55}$ & $(\sigma, \gamma)$ & $(\rho, \delta, \omega, \theta, \rho_{1}, \rho_{2}, \kappa\zeta, b)$ & 10.52 & [0.15, 1.20] \\
$S_{56}$ & $(\rho, \gamma, b)$ & $(\delta, \omega, \theta, \sigma, \rho_{1}, \rho_{2}, \kappa\zeta)$ & 10.55 & [0.15, 1.20] \\
$S_{57}$ & $(\rho, \gamma)$ & $(\delta, \omega, \theta, \sigma, \rho_{1}, \rho_{2}, \kappa\zeta, b)$ & 10.65 & [0.15, 1.21] \\
$S_{58}$ & $(\gamma, b)$ & $(\rho, \delta, \omega, \theta, \sigma, \rho_{1}, \rho_{2}, \kappa\zeta)$ & 10.88 & [0.13, 1.22] \\
$S_{59}$ & $(\sigma, \gamma, b)$ & $(\rho, \delta, \omega, \theta, \rho_{1}, \rho_{2}, \kappa\zeta)$ & 11.02 & [0.13, 1.23] \\
$S_{60}$ & $(\gamma)$ & $(\rho, \delta, \omega, \theta, \sigma, \rho_{1}, \rho_{2}, \kappa\zeta, b)$ & 11.07 & [0.12, 1.23] \\

\end{longtable}

\renewcommand{\arraystretch}{1.5}
\begin{longtable}{@{}lcccccccccc@{}}
\caption{\textbf{Contribution of Fixed Parameters to Sensitivity Statistic}}\label{tab: ContributionFixed}\\
\hline \hline
\textit{Name} & $\rho$ & $\delta$ & $\omega$ & $\theta$ & $\sigma$ & $\rho_{1}$ & $\rho_{2}$ & $\kappa\zeta$ & $\gamma$ & $b$ \\
\hline \hline
\endfirsthead
\multicolumn{11}{c}{\textit{Table \thetable\ (continued)}}\\
\hline \hline
\textit{Name} & $\rho$ & $\delta$ & $\omega$ & $\theta$ & $\sigma$ & $\rho_{1}$ & $\rho_{2}$ & $\kappa\zeta$ & $\gamma$ & $b$ \\
\hline \hline
\endhead
\hline \hline
\endfoot
\hline \hline
\endlastfoot
\textit{Original} & 0.0\% & 0.0\% & 0.0\% & 0.0\% & 0.0\% & - & - & - & - & 100.0\% \\
$S^*$ & 99.6\% & - & 0.4\% & 0.0\% & - & - & - & 0.0\% & - & - \\
$S_{2}$ & - & - & 0.0\% & 0.0\% & 100.0\% & - & - & 0.0\% & - & - \\
$S_{3}$ & - & 22.9\% & 2.0\% & 4.3\% & 70.7\% & - & - & 0.2\% & - & - \\
$S_{4}$ & - & - & 0.0\% & 0.0\% & 87.6\% & 12.4\% & - & 0.0\% & - & - \\
$S_{5}$ & 0.6\% & 77.7\% & 6.6\% & 14.6\% & - & - & - & 0.5\% & - & - \\
$S_{6}$ & 0.6\% & - & 0.0\% & 0.0\% & 99.4\% & - & - & 0.0\% & - & - \\
$S_{7}$ & - & - & 0.0\% & 0.0\% & 33.8\% & - & 66.2\% & 0.0\% & - & - \\
$S_{8}$ & 0.7\% & 2.6\% & 0.2\% & 0.5\% & 96.1\% & - & - & 0.0\% & - & - \\
$S_{9}$ & - & 1.2\% & 0.1\% & 0.2\% & 27.8\% & - & 70.6\% & 0.0\% & - & - \\
$S_{10}$ & 0.8\% & - & 0.0\% & 0.0\% & 86.8\% & - & 12.4\% & 0.0\% & - & - \\
$S_{11}$ & 0.8\% & 1.6\% & 0.1\% & 0.3\% & 93.2\% & - & 4.1\% & 0.0\% & - & - \\
$S_{12}$ & - & - & 0.0\% & 0.0\% & - & - & - & 0.0\% & - & 100.0\% \\
$S_{13}$ & - & 0.2\% & 0.0\% & 0.0\% & - & - & - & 0.0\% & - & 99.7\% \\
$S_{14}$ & - & 78.0\% & 6.8\% & 14.7\% & - & - & - & 0.5\% & - & - \\
$S_{15}$ & - & - & 0.0\% & 0.0\% & - & - & 7.5\% & 0.0\% & - & 92.5\% \\
$S_{16}$ & - & - & 0.0\% & 0.0\% & - & 10.2\% & - & 0.0\% & - & 89.8\% \\
$S_{17}$ & 0.3\% & 0.3\% & 0.0\% & 0.0\% & - & - & 99.4\% & 0.0\% & - & - \\
$S_{18}$ & 0.3\% & - & 0.0\% & 0.0\% & - & - & 99.7\% & 0.0\% & - & - \\
$S_{19}$ & 0.2\% & - & 0.0\% & 0.0\% & - & 99.8\% & - & 0.0\% & - & - \\
$S_{20}$ & - & 8.9\% & 0.7\% & 1.7\% & 2.7\% & 85.9\% & - & 0.1\% & - & - \\
$S_{21}$ & - & 2.9\% & 0.3\% & 0.6\% & - & 52.0\% & - & 0.0\% & - & 44.2\% \\
$S_{22}$ & - & 1.3\% & 0.1\% & 0.3\% & - & - & 13.6\% & 0.0\% & - & 84.7\% \\
$S_{23}$ & 0.0\% & - & 0.0\% & 0.0\% & - & - & - & 0.0\% & - & 100.0\% \\
$S_{24}$ & - & 46.2\% & 3.9\% & 8.7\% & - & - & 40.9\% & 0.3\% & - & - \\
$S_{25}$ & 0.0\% & - & 0.0\% & 0.0\% & 0.0\% & - & - & 0.0\% & - & 100.0\% \\
$S_{26}$ & 0.0\% & - & 0.0\% & 0.0\% & - & - & 0.2\% & 0.0\% & - & 99.7\% \\
$S_{27}$ & - & - & 0.0\% & 0.0\% & 1.4\% & - & - & 0.0\% & - & 98.6\% \\
$S_{28}$ & - & - & 0.0\% & 0.0\% & 1.0\% & - & 2.1\% & 0.0\% & - & 97.0\% \\
$S_{29}$ & 0.0\% & 16.4\% & 1.4\% & 3.1\% & - & 79.0\% & - & 0.1\% & - & - \\
$S_{30}$ & 0.0\% & - & 0.0\% & 0.0\% & 0.1\% & - & 0.6\% & 0.0\% & - & 99.3\% \\
$S_{31}$ & 0.0\% & 0.3\% & 0.0\% & 0.1\% & - & - & - & 0.0\% & - & 99.6\% \\
$S_{32}$ & 0.1\% & 0.6\% & 0.0\% & 0.1\% & 0.3\% & - & - & 0.0\% & - & 98.9\% \\
$S_{33}$ & - & 0.0\% & 0.0\% & 0.0\% & 1.4\% & - & - & 0.0\% & - & 98.5\% \\
$S_{34}$ & 0.0\% & 0.5\% & 0.0\% & 0.1\% & - & - & 0.4\% & 0.0\% & - & 98.9\% \\
$S_{35}$ & 0.0\% & 0.6\% & 0.0\% & 0.1\% & 0.2\% & - & 0.0\% & 0.0\% & - & 98.9\% \\
$S_{36}$ & - & 0.5\% & 0.0\% & 0.1\% & 0.6\% & - & 5.3\% & 0.0\% & - & 93.5\% \\
$S_{37}$ & - & - & 0.0\% & 0.0\% & 1.6\% & 7.7\% & - & 0.0\% & - & 90.7\% \\
$S_{38}$ & - & 1.6\% & 0.1\% & 0.3\% & 1.7\% & 33.9\% & - & 0.0\% & - & 62.4\% \\
$S_{39}$ & - & 10.1\% & 0.9\% & 1.9\% & - & 87.1\% & - & 0.1\% & - & - \\
$S_{40}$ & 0.0\% & - & 0.0\% & 0.0\% & - & 47.8\% & - & 0.0\% & - & 52.2\% \\
$S_{41}$ & 0.0\% & 0.5\% & 0.0\% & 0.1\% & - & 65.2\% & - & 0.0\% & - & 34.2\% \\
$S_{42}$ & 0.2\% & - & 0.0\% & 0.0\% & 8.8\% & 91.0\% & - & 0.0\% & - & - \\
$S_{43}$ & 0.1\% & - & 0.0\% & 0.0\% & 3.5\% & 86.9\% & - & 0.0\% & - & 9.5\% \\
$S_{44}$ & 0.1\% & 0.1\% & 0.0\% & 0.0\% & 2.7\% & 89.0\% & - & 0.0\% & - & 8.1\% \\
$S_{45}$ & 0.1\% & 0.1\% & 0.0\% & 0.0\% & 6.7\% & 93.0\% & - & 0.0\% & - & - \\
$S_{46}$ & - & - & 0.0\% & 0.0\% & - & 80.0\% & 18.1\% & 0.0\% & - & 1.9\% \\
$S_{47}$ & - & - & 0.0\% & 0.0\% & 0.2\% & 80.2\% & 19.6\% & 0.0\% & - & - \\
$S_{48}$ & 0.0\% & - & 0.0\% & 0.0\% & - & 80.6\% & 16.7\% & 0.0\% & - & 2.7\% \\
$S_{49}$ & - & - & 0.0\% & 0.0\% & 0.0\% & 79.8\% & 17.9\% & 0.0\% & - & 2.3\% \\
$S_{50}$ & - & 0.0\% & 0.0\% & 0.0\% & - & 80.1\% & 18.0\% & 0.0\% & - & 1.9\% \\
$S_{51}$ & 0.0\% & - & 0.0\% & 0.0\% & 0.6\% & 84.8\% & 14.6\% & 0.0\% & - & - \\
$S_{52}$ & 0.0\% & - & 0.0\% & 0.0\% & - & 74.4\% & 25.6\% & 0.0\% & - & - \\
$S_{53}$ & - & 1.5\% & 0.1\% & 0.3\% & - & 89.5\% & 8.6\% & 0.0\% & - & - \\
$S_{54}$ & 0.0\% & - & 0.0\% & 0.0\% & 0.0\% & 82.8\% & 15.1\% & 0.0\% & - & 2.1\% \\
$S_{55}$ & 0.0\% & 0.0\% & 0.0\% & 0.0\% & - & 81.1\% & 16.3\% & 0.0\% & - & 2.6\% \\
$S_{56}$ & - & 0.0\% & 0.0\% & 0.0\% & 0.1\% & 81.0\% & 18.9\% & 0.0\% & - & - \\
$S_{57}$ & - & 0.0\% & 0.0\% & 0.0\% & 0.0\% & 80.0\% & 17.7\% & 0.0\% & - & 2.3\% \\
$S_{58}$ & 0.0\% & 0.0\% & 0.0\% & 0.0\% & 0.5\% & 85.2\% & 14.2\% & 0.0\% & - & - \\
$S_{59}$ & 0.0\% & 0.0\% & 0.0\% & 0.0\% & - & 77.4\% & 22.6\% & 0.0\% & - & - \\
$S_{60}$ & 0.0\% & 0.0\% & 0.0\% & 0.0\% & 0.0\% & 82.9\% & 15.0\% & 0.0\% & - & 2.1\% \\
\end{longtable}

\section{Monte Carlo appendix}
\label{app:mc-implementation}

The Monte Carlo exercise is implemented as follows. Let $\eta_0$ denote the reference structural parameter vector used in the empirical application. For each value of $\eta_0$, we solve the structural model and compute the model-implied vector
\[
m(\eta_0)\in\mathbb{R}^{33},
\]
whose entries correspond to the same reduced-form responses matched in the CMD estimation.

The simulated observations are draws of this stacked reduced-form moment vector. Specifically, for each Monte Carlo replication and sample size $n$, we draw
\[
m_i = m(\eta_0) + u_i,
\qquad
u_i \sim N(0,\widehat{\Omega}),
\qquad i=1,\ldots,n,
\]
where $\widehat{\Omega}$ is the covariance matrix computed from the bootstrap distribution of the empirical reduced-form moments. We then compute the sample mean and sample covariance of $\{m_i\}_{i=1}^n$ and use these objects to construct the same moment structure that enters the empirical CMD estimator.

Thus, the Monte Carlo sample size $n$ indexes the precision of the simulated reduced-form moment vector, rather than the number of raw daily or tick observations. We report results for $n=150$, which is of the same order as the number of FOMC-event observations underlying the empirical reduced-form moments, and for $n=500$, which provides a larger-sample benchmark in which sampling noise is smaller.

For each partition $S$ and each miscalibration size $\varepsilon$, we compute the partition-specific worst-case perturbation of the fixed parameters. The fixed parameters are held at these miscalibrated values, while the estimated parameters are re-estimated using the simulated CMD moments. Across Monte Carlo replications, this delivers the finite-sample distribution of $\widehat{\gamma}$ under each partition and miscalibration size. We summarize these distributions using bias, variance, and mean squared error relative to the reference value $\gamma_0$.

\subsection{Robustness to the weak-identification rank threshold}
\label{app:weak_id_threshold_robustness}\label{appendix: weakid_tolerance}

This appendix checks whether the Monte Carlo conclusions depend on the
hard-threshold used to estimate the conditional rank of each partition. In the
baseline implementation, a partition is kept whenever
\[
    \widehat r_S = |S|,
    \qquad
    \widehat r_S
    =
    \#\{j:\widehat\sigma_{j,S}>\tau_n\},
    \qquad
    \tau_n=\left(\frac{\log n}{n}\right)^{1/2}.
\]
To study sensitivity to this implementation choice, we replace the baseline cutoff by
\[
    \tau_n(a)
    =
    \left(\frac{\log n}{n}\right)^a,
    \qquad
    a\in\{0.50,0.60,0.70,0.80,0.90,1.00\}.
\]
The scalar \(a\) is not a structural parameter. It is the exponent applied to
\(\log n/n\) in the rank threshold. Since \(\log n/n<1\) in the samples considered
here, larger values of \(a\) imply a smaller cutoff and therefore a more permissive
admissibility rule. For each value of \(a\), I recompute the admissible set, rank
admissible partitions by the sensitivity statistic \(K_S\), and repeat the same
worst-case Monte Carlo exercise for the top-ranked partitions.

Table~\ref{tab:rank_threshold_robustness} reports the results in compressed form.
The first panel shows how the top-five sensitivity ranking changes with \(a\).
The second panel reports Monte Carlo MSEs only for the distinct partitions that
appear in one of those rankings. This avoids repeating identical rows: for
\(a\in\{0.50,0.60,0.70,0.80\}\), both the number of admissible partitions and the
top-five ranking are exactly the same. The main conclusion is stable. The selected
partition \(S^*\) remains the least-sensitive partition for every threshold
exponent. When the cutoff is made more permissive, additional partitions enter the
admissible set, but they do not overturn the main finite-sample message: the
least-sensitive partition remains unchanged and, when miscalibration is large
relative to the sample size, delivers the lowest MSE.

\begin{landscape}
\begin{table}[!t]
\centering
\scriptsize
\setlength{\tabcolsep}{3pt}
\renewcommand{\arraystretch}{1.12}
\begin{threeparttable}
\caption{Robustness to the rank-threshold exponent}
\label{tab:rank_threshold_robustness}

\begin{tabular}{@{}p{0.15\linewidth}c p{0.18\linewidth} p{0.32\linewidth} p{0.20\linewidth}@{}}
\toprule
Threshold exponent \(a\)
& \(\#\widehat{\mathcal S}_{AD}(a)\)
& Top five by \(K_S\)
& Sensitivities
& Main change relative to baseline \\
\midrule
\(a\in\{0.50,0.60,0.70,0.80\}\)
& 60
& \(S^*,\ S_2,\ S_3,\ S_4,\ S_5\)
& \(0.0582,\ 0.3159,\ 0.3602,\ 0.3772,\ 0.5324\)
& No change: same admissible set size and same top-five ranking. \\

\(a=0.90\)
& 86
& \(S^*,\ S_2,\ S_3,\ S_4,\ S_{0.9}\)
& \(0.0582,\ 0.3159,\ 0.3602,\ 0.3772,\ 0.3860\)
& More permissive threshold admits additional partitions; the first four positions are unchanged. \\

\(a=1.00\)
& 95
& \(S^*,\ S_{1a},\ S_{1b},\ S_{1c},\ S_2\)
& \(0.0582,\ 0.0591,\ 0.0600,\ 0.3153,\ 0.3159\)
& Near-ties appear only under the most permissive threshold, but \(S^*\) remains first. \\
\bottomrule
\end{tabular}

\vspace{0.8em}

\begin{tabularx}{\linewidth}{@{}l
>{\raggedright\arraybackslash}p{0.19\linewidth}
>{\raggedright\arraybackslash}p{0.16\linewidth}
r
rrr
rrr@{}}
\toprule
& & & &
\multicolumn{3}{c}{MSE, \(n=150\)}
& \multicolumn{3}{c}{MSE, \(n=500\)} \\
\cmidrule(lr){5-7}\cmidrule(lr){8-10}
Partition
& Estimated parameters
& Fixed parameters
& \(K_S\)
& \(\varepsilon=5\%\)
& \(\varepsilon=10\%\)
& \(\varepsilon=15\%\)
& \(\varepsilon=5\%\)
& \(\varepsilon=10\%\)
& \(\varepsilon=15\%\) \\
\midrule
\(S^*\)
& \(\delta,\sigma,\rho_1,\rho_2,\gamma,b\)
& \(\rho,\omega,\theta,\kappa\zeta\)
& 0.0582
& 4.7 & 5.5 & 6.7
& 1.4 & 1.8 & 2.5 \\

\(S_2\)
& \(\rho,\delta,\rho_1,\rho_2,\gamma,b\)
& \(\omega,\theta,\sigma,\kappa\zeta\)
& 0.3159
& 5.7 & 8.6 & 22.4
& 4.5 & 7.0 & 21.0 \\

\(S_3\)
& \(\rho,\rho_1,\rho_2,\gamma,b\)
& \(\delta,\omega,\theta,\sigma,\kappa\zeta\)
& 0.3602
& 10.7 & 8.8 & 12.9
& 9.6 & 7.3 & 11.7 \\

\(S_4\)
& \(\rho,\delta,\rho_2,\gamma,b\)
& \(\omega,\theta,\sigma,\rho_1,\kappa\zeta\)
& 0.3772
& 21.1 & 349.3 & 1492.8
& 20.4 & 345.1 & 1494.3 \\

\(S_5\)
& \(\sigma,\rho_1,\rho_2,\gamma,b\)
& \(\rho,\delta,\omega,\theta,\kappa\zeta\)
& 0.5324
& 79.7 & 448.9 & 602.8
& 78.1 & 437.9 & 605.3 \\

\(S_{0.9}\)
& \(\rho,\omega,\rho_2,\gamma,b\)
& \(\delta,\theta,\sigma,\rho_1,\kappa\zeta\)
& 0.3860
& 56.8 & 653.2 & 3449.9
& 56.9 & 648.5 & 3540.0 \\

\(S_{1a}\)
& \(\rho,\delta,\sigma,\rho_1,\rho_2,\gamma,b\)
& \(\omega,\theta,\kappa\zeta\)
& 0.0591
& 54.1 & 51.7 & 52.2
& 15.7 & 16.9 & 17.2 \\

\(S_{1b}\)
& \(\omega,\sigma,\rho_1,\rho_2,\gamma,b\)
& \(\rho,\delta,\theta,\kappa\zeta\)
& 0.0600
& 4.0 & 6.7 & 18.2
& 1.2 & 4.1 & 16.0 \\

\(S_{1c}\)
& \(\rho,\omega,\rho_1,\rho_2,\gamma,b\)
& \(\delta,\theta,\sigma,\kappa\zeta\)
& 0.3153
& 5.7 & 9.0 & 31.3
& 4.5 & 7.3 & 29.3 \\
\bottomrule
\end{tabularx}

\begin{tablenotes}[flushleft]
\footnotesize
\item Notes: The rank-threshold exponent \(a\) defines
\(\tau_n(a)=(\log n/n)^a\). Larger \(a\) corresponds to a smaller singular-value
cutoff and therefore a more permissive admissibility rule. \(K_S\) is the
sensitivity statistic used to rank admissible partitions. MSEs are multiplied by
\(10^4\) and rounded to one decimal place. The table reports MSEs, rather than
bias and variance separately, to focus on the robustness of the substantive
ranking; the worst-case direction is recomputed separately for each partition and
each value of \(\varepsilon\), so the sign of the bias is not the object of
comparison.
\end{tablenotes}
\end{threeparttable}
\end{table}
\end{landscape}


\bibliographystyle{qe} 
\bibliography{bibliography}  


\end{appendix}
\end{document}